\documentclass[11pt, journal, final]{IEEEtran}
\onecolumn
\usepackage[numbers,sort&compress]{natbib}
\usepackage{graphics}
\usepackage{enumerate}
\usepackage{setspace}
\usepackage{epstopdf}
\usepackage{amsmath}
\usepackage{graphicx,amsfonts,amscd,amssymb,bm,epsfig,epsf,url,color}
\usepackage{subfigure}
\usepackage{bbm}
\usepackage{microtype}
\usepackage{algorithm}
\usepackage{algorithmic}
\usepackage{amsthm}
\usepackage{verbatim}
\usepackage{tablefootnote}

\numberwithin{equation}{section}

\usepackage{hyperref}

\newtheorem{theorem}{Theorem}[section]
\newtheorem{lemma}[theorem]{Lemma}
\newtheorem{corollary}[theorem]{Corollary}
\newtheorem{proposition}[theorem]{Proposition}

\newtheorem{remark}[theorem]{Remark}

\newcommand{\eps}{\varepsilon}
\newcommand{\R}{\mathbb{R}}

\newcommand{\E}{\mathbb{E}}
\newcommand{\N}{\mathcal{N}}

\newcommand{\reals}{\bb R}
\newcommand{\mb}{\mathbf}
\newcommand{\mc}{\mathcal}

\newcommand{\bb}{\mathbb}

\newcommand{\indicator}[1]{\mathbbm 1_{#1}}
\newcommand{\norm}[1]{\left\lVert#1\right\rVert}
\newcommand{\ol}{\overline}

\newcommand{ \abs }[1]{\left| #1 \right|}
\newcommand{ \Brac }[1]{\left\lbrace #1 \right\rbrace}
\newcommand{ \brac }[1]{\left[ #1 \right]}
\newcommand{ \paren }[1]{ \left( #1 \right) }
\newcommand{\innerprod}[2]{\left\langle #1,  #2 \right\rangle}
\newcommand{\prob}[1]{\bb P\left[ #1 \right]}
\newcommand{\expect}[1]{\bb E\left[ #1 \right]}

\newcommand{\event}{\mc E}

\newcommand{\jw}[1]{{\color{blue}{\bf John: #1}}}

\begin{document}

{

\title{Finding a sparse vector in a subspace: linear sparsity using alternating directions}

\author{Qing~Qu, ~\IEEEmembership{Student Member,~IEEE,}
				Ju~Sun, ~\IEEEmembership{Student Member,~IEEE,}
				and~John~Wright,~\IEEEmembership{Member,~IEEE}
        \thanks{This work was partially supported by grants ONR N00014-13-1-0492, NSF 1343282, NSF 1527809, and funding from the Moore and Sloan Foundations. Q. Qu, J. Sun and J. Wright are all with the Electrical Engineering Department, Columbia University, New York, NY, 10027, USA (e-mail: \{qq2105, js4038, jw2966\}@columbia.edu). This paper is an extension of our previous conference version \cite{qu2014finding}.
}
}


%


\markboth{IEEE Transaction on Information Theory,~Vol.~xx, No.~xx, xxxx~2015}%
{Qu \MakeLowercase{\textit{et al.}}: Finding a sparse vector in a subspace}

\maketitle

\begin{abstract}

Is it possible to find the sparsest vector (direction) in a generic subspace $\mathcal{S} \subseteq \R^p$ with $\text{dim}\paren{\mathcal{S}}=n < p$? This problem can be considered a homogeneous variant of the sparse recovery problem, and finds connections to sparse dictionary learning, sparse PCA, and many other problems in signal processing and machine learning. In this paper, we focus on a \emph{planted sparse model} for the subspace: the target sparse vector is embedded in an otherwise random subspace. Simple convex heuristics for this planted recovery problem provably break down when the fraction of nonzero entries in the target sparse vector substantially exceeds $O(1/\sqrt{n})$. In contrast, we exhibit a relatively simple nonconvex approach based on alternating directions, which provably succeeds even when the fraction of nonzero entries is $\Omega(1)$. To the best of our knowledge, this is the first practical algorithm to achieve linear scaling under the planted sparse model. Empirically, our proposed algorithm also succeeds in more challenging data models, e.g., sparse dictionary learning. 
\end{abstract}



\section{Introduction}

Suppose that a linear subspace $\mc S$ embedded in $\R^p$ contains a sparse vector $\mb x_0 \ne \mb 0$. Given an arbitrary basis of $\mc S$, can we efficiently recover $\mb x_0$ (up to scaling)? Equivalently, provided a matrix $\mb A \in \R^{\left(p - n\right) \times p}$ with $\text{Null}(\mb A) = \mc S$, \footnote{ $\text{Null}(\mb A) \doteq \Brac{\mb x\in \bb R^p \mid \mb A \mb x = \mb 0 } $ denotes the null space of $\mb A$.} can we efficiently find a nonzero sparse vector $\mb x$ such that $\mb A \mb x = \mb 0$? In the language of sparse recovery, can we solve
\begin{equation}\label{eqn:sparse-null}
\min_{\mb x} \; \norm{\mb x}_0 \quad \text{s.t.} \quad \mb A \mb x = \mb 0, \; \mb x \ne \mb 0 \qquad {\text ?}
\end{equation} 
In contrast to the standard sparse recovery problem ($\mb A \mb x = \mb b$, $\mb b \ne \mb 0$), for which convex relaxations perform nearly optimally for broad classes of designs $\mb A$~\cite{candes2005decoding, donoho2006most}, the computational properties of problem \eqref{eqn:sparse-null} are not nearly as well understood. It has been known for several decades that the basic formulation 
\begin{equation} \label{eqn:L0}
\min_{\bf x}\; \norm{\bf x}_0, \quad \text{s.t.} \quad {\bf x}\in \mathcal{S} \setminus \{{\bf 0}\},
\end{equation}
is NP-hard for an arbitrary subspace~\cite{mccormick1983combinatorial, coleman1986null}. In this paper, we assume a specific random \emph{planted sparse model} for the subspace $\mc S$: a target sparse vector is embedded in an otherwise random subspace. We will show that under the specific random model, problem~\eqref{eqn:L0} is tractable by an efficient algorithm based on nonconvex optimization.

\subsection{Motivation}
The general version of Problem \eqref{eqn:L0}, in which $\mathcal{S}$ can be an arbitrary subspace, takes several forms in numerical computation and computer science, and underlies several important problems in modern signal processing and machine learning. Below we provide a sample of these applications.  
\par\smallskip
\noindent\textbf{Sparse Null Space and Matrix Sparsification:} The \emph{sparse null space} problem is finding the sparsest matrix $\mb N$ whose columns span the null space of a given matrix $\mb A$. The problem arises in the context of solving linear equality problems in constrained optimization~\cite{coleman1986null}, null space methods for quadratic programming \cite{berry85algorithm}, and solving underdetermined linear equations~\cite{gilbert86computing}. The \emph{matrix sparsification} problem is of similar flavor, the task is finding the sparsest matrix $\mb B$ which is equivalent to a given full rank matrix $\mb A$ under elementary column operations. Sparsity helps simplify many fundamental matrix operations (see~\cite{duff86direct}), and the problem has applications in areas such as machine learning~\cite{Smola00sparsegreedy} and in discovering cycle bases of graphs~\cite{Kavitha04afaster}. \cite{gottlieb2010matrix} discusses connections between the two problems and also to other problems in complexity theory. 
\par\smallskip
\noindent\textbf{Sparse (Complete) Dictionary Learning:} In dictionary learning, given a data matrix $\mb Y$, one seeks an approximation $\mb Y \approx \mb A \mb X$, such that $\mb A$ is a representation dictionary with certain desired structure and $\mb X$ collects the representation coefficients with maximal sparsity. Such compact representation naturally allows signal compression, and also facilitates efficient signal acquisition and classification (see relevant discussion in~\cite{mairal2014sparse}). When $\mb A$ is required to be complete (i.e., square and invertible), by linear algebra, we have\footnote{Here, $\mathrm{row}(\cdot)$ denotes the row space.} $\mathrm{row}(\mb Y) = \mathrm{row}(\mb X)$~\cite{spielman2013exact}. Then the problem reduces to finding sparsest vectors (directions) in the known subspace $\mathrm{row}(\mb Y)$, i.e.~\eqref{eqn:L0}. Insights into this problem have led to new theoretical developments on complete dictionary learning~\cite{spielman2013exact, hand2013recovering, sun2015complete}. 
\par\smallskip
\noindent\textbf{Sparse Principal Component Analysis (Sparse PCA):} In geometric terms, Sparse PCA (see, e.g.,~\cite{zou2006sparse,johnstone2009consistency,d2007direct} for early developments and~\cite{krauthgamer2015semidefinite, ma2015sum} for discussion of recent results) concerns stable estimation of a linear subspace spanned by a sparse basis, in the data-poor regime, i.e., when the available data are not numerous enough to allow one to decouple the subspace estimation and sparsification tasks. Formally, given a data matrix $\mb Z = \mb U_0 \mb X_0 + \mb E$,\footnote{Variants of multiple-component formulations often add an additional orthonormality constraint on $\mb U_0$ but involve a different notation of sparsity; see, e.g.,~\cite{zou2006sparse, vu2013fantope, lei2015sparsistency, wang2014nonconvex}. } where $\mb Z \in \R^{p \times n}$ collects column-wise $n$ data points, $\mb U_0 \in \R^{p \times r}$ is the sparse basis, and $\mb E$ is a noise matrix, one is asked to estimate $\mb U_0$ (up to sign, scale, and permutation). Such a factorization finds applications in gene expression, financial data analysis and pattern recognition~\cite{Aspremont07sparse}. When the subspace is known (say by the PCA estimator with enough data samples), the problem again reduces to instances of~\eqref{eqn:L0} and is already nontrivial\footnote{\cite{hand2013recovering} has also discussed this data-rich sparse PCA setting. }. The full geometric sparse PCA can be treated as finding sparse vectors in a subspace that is subject to perturbation. 
\par\smallskip
In addition, variants and generalizations of the problem \eqref{eqn:L0} have also been studied in applications regarding control and optimization~\cite{zhao2013rank}, nonrigid structure from motion~\cite{dai2012simple}, spectral estimation and Prony's problem~\cite{beylkin2005approximation}, outlier rejection in PCA~\cite{manolis2015dualPCA}, blind source separation~\cite{zibulevsky2001blind}, graphical model learning~\cite{anandkumar2013overcomplete}, and sparse coding on manifolds~\cite{ho2013nonlinear}; see also~\cite{nakatsukasa15finding} and the references therein. 

\subsection{Prior Arts}
Despite these potential applications of problem \eqref{eqn:L0}, it is only very recently that efficient computational surrogates with nontrivial recovery guarantees have been discovered for some cases of practical interest. In the context of sparse dictionary learning, Spielman et al.\ \cite{spielman2013exact} introduced a convex relaxation which replaces the nonconvex problem \eqref{eqn:L0} with a sequence of linear programs:
\begin{equation}
\ell^1/\ell^\infty \text{ Relaxation:}\qquad  \min_{\bf x} \norm{\bf x}_1,\quad \text{s.t.} \quad x(i)=1,~{\bf x}\in \mathcal{S},~1\le i \le p. \label{l1inf}
\end{equation}
They proved that when $\mc S$ is generated as a span of $n$ random sparse vectors, with high probability (w.h.p.), the relaxation recovers these vectors, provided the probability of an entry being nonzero is at most $\theta \in O\paren{ 1/\sqrt{n}}$. In the {\em planted sparse model}, in which $\mathcal{S}$ is formed as direct sum of a single sparse vector $\mb x_0$ and a ``generic'' subspace, Hand and Demanet proved that \eqref{l1inf} also correctly recovers $\mb x_0$, provided the fraction of nonzeros in $\mb x_0$ scales as $\theta \in O\paren{1/ \sqrt{n}}$ \cite{hand2013recovering}. One might imagine improving these results by tightening the analyses. Unfortunately, the results of~\cite{spielman2013exact,hand2013recovering} are essentially sharp: {\em when $\theta$ substantially exceeds $\Omega(1/\sqrt{n})$, in both models the relaxation \eqref{l1inf} provably breaks down.} Moreover, the most natural semidefinite programming (SDP) relaxation of \eqref{eqn:sparse-null},
\begin{equation} \label{eqn:SDP-relax}
\min_{\mb X} \norm{\mb X}_1, \quad \text{s.t.} \quad \left\langle \mb A^\top \mb A, \mb X \right\rangle = 0, \; \mathrm{trace}[\mb X] = 1, \; \mb X \succeq \mb 0. 
\end{equation}
also breaks down at exactly the same threshold of $\theta \sim O(1/\sqrt{n})$.\footnote{This breakdown behavior is again in sharp contrast to the standard sparse approximation problem (with $\mb b \ne \mb 0$), in which it is possible to handle very large fractions of nonzeros (say, $\theta = \Omega(1/\log n)$, or even $\theta = \Omega(1)$) using a very simple $\ell^1$ relaxation~\cite{candes2005decoding, donoho2006most}}

One might naturally conjecture that this $1/\sqrt{n}$ threshold is simply an intrinsic price we must pay for having an efficient algorithm, even in these random models. Some evidence towards this conjecture might be borrowed from the superficial similarity of \eqref{eqn:L0}-\eqref{eqn:SDP-relax} and {\em sparse PCA}~\cite{zou2006sparse}. In sparse PCA, there is a substantial gap between what can be achieved with currently available efficient algorithms and the information theoretic optimum~\cite{berthet2013complexity, krauthgamer2015semidefinite}. Is this also the case for recovering a sparse vector in a subspace? {\em Is $\theta \in O\paren{1/\sqrt{n}}$ simply the best we can do with efficient, guaranteed algorithms?}
 
\begin{table}
\center 
\caption{Comparison of existing methods for recovering a planted sparse vector in a subspace}
\label{table:comparison}
\begin{tabular}{c|c|c}
\hline 
Method & Recovery Condition & Time Complexity\tablefootnote{All estimates here are based on the standard interior point methods for solving linear and semidefinite programs. Customized solvers may result in order-wise speedup for specific problems. $\eps$ is the desired numerical accuracy. } \\
\hline 
$\ell^1/\ell^\infty$ Relaxation \cite{hand2013recovering} & $\theta \in O(1/\sqrt{n})$ & $O(n^3 p \log(1/\eps))$ \\
SDP Relaxation & $\theta \in O(1/\sqrt{n})$ & $O\paren{p^{3.5}\log\paren{1/\eps} }$ \\
SOS Relaxation \cite{barak2013rounding} & $p\geq \Omega(n^2), \theta\in O(1)$  & $\sim O(p^7 \log(1/\eps))$ \tablefootnote{Here our estimation is based on the degree-4 SOS hierarchy used in~\cite{barak2013rounding} to obtain an initial approximate recovery. }  \\
Spectral Method \cite{hopkins2015speeding} & $p \geq \Omega(n^2 \text{poly} \log(n)), \theta \in O(1) $ & $O\paren{np \log(1/\epsilon) }$ \\
This work & $p \geq \Omega(n^4\log n),~\theta\in O(1)$  & $O(n^5p^2\log n+n^3 p \log(1/\eps))$ \\
\hline 
\end{tabular}
\end{table}

Remarkably, this is not the case. Recently, Barak et al.\ introduced a new rounding technique for sum-of-squares relaxations, and showed that the sparse vector $\mb x_0$ in the planted sparse model can be recovered when $p \ge \Omega\paren{ n^2}$ and $\theta = \Omega(1)$ \cite{barak2013rounding}. It is perhaps surprising that this is possible at all with a polynomial time algorithm. Unfortunately, the runtime of this approach is a high-degree polynomial in $p$ (see Table \ref{table:comparison}); for machine learning problems in which $p$ is often either the feature dimension or the sample size, this algorithm is mostly of theoretical interest only. However, it raises an interesting algorithmic question: 
 {\em Is there a practical algorithm that provably recovers a sparse vector with $\theta \gg 1/ \sqrt{n}$ portion of nonzeros from a generic subspace $\mc S$?}
 
\subsection{Contributions and Recent Developments}
In this paper, we address the above problem under the planted sparse model. We allow $\mb x_0$ to have up to $\theta_0 p$ nonzero entries, where $\theta_0 \in \paren{0, 1}$ is a constant. We provide a relatively simple algorithm which, w.h.p., exactly recovers $\mb x_0$, provided that $p \ge \Omega\paren{ n^4 \log n}$. A comparison of our results with existing methods is shown in Table \ref{table:comparison}. After initial submission of our paper, Hopkins et al. \cite{hopkins2015speeding} proposed a different simple algorithm based on the spectral method. This algorithm guarantees recovery of the planted sparse vector also up to linear sparsity, whenever $p \ge \Omega(n^2 \mathrm{polylog} (n))$, and comes with better time complexity.\footnote{Despite these improved guarantees in the planted sparse model, our method still produces more appealing results on real imagery data -- see Section~\ref{sec:face_exp} for examples. } 

Our algorithm is based on alternating directions, with two special twists. First, we introduce a special data driven initialization, which seems to be important for achieving $\theta = \Omega(1)$. Second, our theoretical results require a second, linear programming based rounding phase, which is similar to \cite{spielman2013exact}.  Our core algorithm has very simple iterations, of linear complexity in the size of the data, and hence should be scalable to moderate-to-large scale problems. 
 
Besides enjoying the $\theta \sim \Omega(1)$ guarantee that is out of the reach of previous practical algorithms, our algorithm performs well in simulations -- empirically succeeding  with $p \geq \Omega\paren{n \; \mathrm{polylog}(n)}$. It also performs well empirically on more challenging data models, such as the complete dictionary learning model, in which the subspace of interest contains not one, but $n$ random target sparse vectors. This is encouraging, as breaking the $O(1/\sqrt{n})$ sparsity barrier with a practical algorithm and optimal guarantee is an important problem in theoretical dictionary learning~\cite{arora2013new,agarwal2013exact, agarwal2013learning,arora2014more,arora2015simple}. In this regard, our recent work~\cite{sun2015complete} presents an efficient algorithm based on Riemannian optimization that guarantees recovery up to linear sparsity. However, the result is based on different ideas: a different nonconvex formulation, optimization algorithm, and analysis methodology.  

\subsection{Paper Organization, Notations and Reproducible Research}
The rest of the paper is organized as follows. In Section~\ref{sec:optimality}, we provide a nonconvex formulation and show its capability of recovering the sparse vector. Section~\ref{sec:algorithm} introduces the alternating direction algorithm. In Section~\ref{sec:analysis}, we present our main results and sketch the proof ideas. Experimental evaluation of our method is provided in Section~\ref{sec:exp}. We conclude the paper by drawing connections to related work and discussing potential improvements in Section~\ref{sec:discussion}. Full proofs are all deferred to the appendix sections.

For a matrix $\mb X$, we use $\mb x_i$ and $\mb x^j$ to denote its $i$-th column and $j$-th row, respectively, all in column vector form. Moreover, we use $x(i)$ to denote the $i$-th component of a vector $\mb x$. We use the compact notation $[k] \doteq \left\{1, \dots, k\right\}$ for any positive integer $k$, and use $c$ or $C$, and their indexed versions to denote absolute numerical constants. The scope of these constants are always local, namely within a particular lemma, proposition, or proof, such that the apparently same constant in different contexts may carry different values. For probability events, sometimes we will just say the event holds ``with high probability'' (w.h.p.) if the probability of failure is dominated by $p^{-\kappa}$ for some $\kappa > 0$. 

The codes to reproduce all the figures and experimental results can be found online at:
\begin{quote} 
\centering
\url{https://github.com/sunju/psv}.
\end{quote}


\section{Problem Formulation and Global Optimality}\label{sec:optimality}

We study the problem of recovering a sparse vector $\mb x_0 \ne \mb 0$ (up to scale), which is an element of a known subspace $\mc S \subset \R^p$ of dimension $n$, provided an arbitrary orthonormal basis $\mb Y \in \R^{p \times n}$ for $\mc S$.
Our starting point is the nonconvex formulation \eqref{eqn:L0}. Both the objective and the constraint set are nonconvex, and hence it is not easy to optimize over. We relax \eqref{eqn:L0} by replacing the $\ell^0$ norm with the $\ell^1$ norm. For the constraint $\mb x \ne \mb 0$, since in most applications we only care about the solution up to scaling, it is natural to force $\mb x$ to live on the unit sphere $\bb S^{n-1}$, giving 
\begin{equation} \label{eqn:l1-l2}
\min_{\mb x} \; \norm{\mb x }_1, \quad \text{s.t.} \quad \mb x \in \mc S, \; \norm{\mb x}_{2} = 1. 
\end{equation}
This formulation is still nonconvex, and for general nonconvex problems it is known to be NP-hard to find even a local minimizer~\cite{murty1987some}. Nevertheless, the geometry of the sphere is benign enough, such that for well-structured inputs it actually {\em will} be possible to give algorithms that find the global optimizer.

The formulation \eqref{eqn:l1-l2} can be contrasted with \eqref{l1inf}, in which effectively we optimize the $\ell^1$ norm subject to the constraint $\norm{\mb x}_\infty = 1$: because the set $\{\mb x: \norm{\mb x}_\infty = 1\}$ is polyhedral, the $\ell^\infty$-constrained problem immediately yields a sequence of linear programs. This is very convenient for computation and analysis. However, it suffers from the aforementioned breakdown behavior around $\norm{\mb x_0}_0 \sim p / \sqrt{n}$. In contrast, though the sphere $\norm{\mb x}_{2} = 1$ is a more complicated geometric constraint, it will allow much larger number of nonzeros in $\mb x_0$. Indeed, if we consider the global optimizer of a reformulation of \eqref{eqn:l1-l2}: 
\begin{align} \label{eqn:syn_l1_l2}
\min_{\mb q \in \R^n} \; \norm{\mb Y \mb q }_1, \quad \text{s.t.} \quad \norm{\mb q}_{2} = 1, 
\end{align} 
where $\mb Y$ is any orthonormal basis for $\mc S$, the sufficient condition that guarantees exact recovery under the planted sparse model for the subspace is as follows: 
\begin{theorem}[$\ell^1/\ell^2$ recovery, planted sparse model] \label{thm:global}
There exists a constant $\theta_0 > 0$, such that if the subspace $\mc S$ follows the planted sparse model
\begin{equation*}
\mc S = \mathrm{span}\left( \mb x_0, \mb g_1, \dots, \mb g_{n-1} \right) \;\subset\; \R^p,
\end{equation*}
where $\mb g_i \sim_{\text{i.i.d.}} \mc N(\mb 0,\frac{1}{p} \mb I)$, and $\mb x_0 \sim_{\text{i.i.d.}} \tfrac{1}{\sqrt{\theta p}} \mathrm{Ber}(\theta)$ are all jointly independent and $1/\sqrt{n} < \theta < \theta_0$, then the unique (up to sign) optimizer $\mb q^\star$ to~\eqref{eqn:syn_l1_l2}, for any orthonormal basis $\mb Y$ of $\mc S$, produces $\mb Y \mb q^\star = \xi \mb x_0$ for some $\xi \neq 0$ with probability at least $1- cp^{-2}$, provided $p \geq Cn$. Here $c$ and $C$ are positive constants. 
\end{theorem}

Hence, {\em if} we could find the global optimizer of \eqref{eqn:syn_l1_l2}, we would be able to recover $\mb x_0$ whose number of nonzero entries is quite large -- even linear in the dimension $p$ ($\theta = \Omega(1)$). On the other hand, it is not obvious that this should be possible: \eqref{eqn:syn_l1_l2} is nonconvex. In the next section, we will describe a simple heuristic algorithm for approximately solving a relaxed version of the $\ell^1/\ell^2$ problem \eqref{eqn:syn_l1_l2}. More surprisingly, we will then prove that for a class of random problem instances, this algorithm, plus an auxiliary rounding technique, actually recovers the global optimizer -- the target sparse vector $\mb x_0$. The proof requires a detailed probabilistic analysis, which is sketched in Section~\ref{sec:analysis_sketch}. 

Before continuing, it is worth noting that the formulation \eqref{eqn:l1-l2} is in no way novel -- see, e.g., the work of~\cite{zibulevsky2001blind} in blind source separation for precedent. However, our algorithms and subsequent analysis are novel. 

\section{Algorithm based on Alternating Direction Method (ADM)}\label{sec:algorithm}
To develop an algorithm for solving \eqref{eqn:syn_l1_l2}, it is useful to consider a slight relaxation of \eqref{eqn:syn_l1_l2}, in which we introduce an auxiliary variable $\mb x \approx \mb Y \mb q$: 
\begin{equation} \label{eqn:huber-l2}
\min_{ {\bf q}, {\bf x} } f(\mb q,\mb x) \doteq \frac{1}{2}\norm{{\bf Yq} - {\bf x}}_2^2+\lambda \norm{\bf x}_1,\quad\text{s.t.}\quad \norm{\bf q}_2=1. 
\end{equation}
Here, $\lambda > 0$ is a penalty parameter. It is not difficult to see that this problem is equivalent to minimizing the {\em Huber} M-estimator over $\mb Y \mb q$. This relaxation makes it possible to apply the alternating direction method to this problem. This method starts from some initial point $\mb q^{(0)}$, alternates between optimizing with respect to (w.r.t.) $\mb x$ and optimizing w.r.t. $\mb q$:
\begin{eqnarray}
{\bf x}^{(k+1)} &=& \mathop{\arg\min}_{\mb x} \frac{1}{2}\norm{{\bf Yq}^{(k)}-{\bf x}}_2^2 + \lambda\norm{\bf x}_1, \label{update_x_k}\\
{\bf q}^{(k+1)} &=& \mathop{\arg\min}_{\mb q} \frac{1}{2}\norm{{\bf Yq} - {\bf x}^{(k+1)}}_2^2 \; \text{s.t.} \; \norm{\mb q}_2 = 1,  \label{update_q_l}
\end{eqnarray}
where ${\bf x}^{(k)}$ and ${\bf q}^{(k)}$ denote the values of ${\bf x}$ and ${\bf q}$ in the $k$-th iteration. 
Both \eqref{update_x_k} and \eqref{update_q_l} have simple closed form solutions:
\begin{eqnarray}
{\bf x}^{(k+1)} = S_{\lambda}[{\bf Yq}^{(k)}],\qquad {\bf q}^{(k+1)} = \frac{{\bf Y}^\top{\bf x}^{(k+1)}}{\norm{{\bf Y}^\top{\bf x}^{(k+1)}}_2},\label{eqn:closed-form}
\end{eqnarray}
where $S_{\lambda}\brac{x} = \mathrm{sign}(x)\max\left\{\abs{x} - \lambda,0\right\} $ is the soft-thresholding operator. The proposed ADM algorithm is summarized in Algorithm \ref{ADM}.

\begin{algorithm}
\caption{Nonconvex ADM for solving~\eqref{eqn:huber-l2}}
\label{ADM}
\begin{algorithmic}[1]
\renewcommand{\algorithmicrequire}{\textbf{Input:}}
\renewcommand{\algorithmicensure}{\textbf{Output:}}
\REQUIRE~~\
A matrix ${\bf Y}\in\mathbb{R}^{p\times n}$ with $\mb Y^\top \mb Y = \mb I$, initialization ${\bf q}^{(0)}$, threshold parameter $\lambda > 0$.
\ENSURE~~\
The recovered sparse vector $\hat{\bf x}_0= {\bf Yq}^{(k)}$
\FOR{$k = 0, \dots, O\paren{n^4 \log n}$}
\STATE ${\bf x}^{(k+1)} = S_{\lambda}[{\bf Yq}^{(k)}]$,
\STATE ${\bf q}^{(k+1)} = \frac{{\bf Y}^\top{\bf x}^{(k+1)}}{\norm{{\bf Y}^\top{\bf x}^{(k+1)}}_2}$,
\ENDFOR
\end{algorithmic}
\end{algorithm}

The algorithm is simple to state and easy to implement. However, if our goal is to recover the {\em sparsest} vector $\mb x_0$, some additional tricks are needed. 
\par\smallskip

\noindent\textbf{Initialization.} Because the problem \eqref{eqn:syn_l1_l2} is nonconvex, an arbitrary or random initialization may not produce a global minimizer.\footnote{More precisely, in our models, random initialization {\em does} work, but only when the subspace dimension $n$ is \emph{extremely} low compared to the ambient dimension $p$.} In fact, good initializations are critical for the proposed ADM algorithm to succeed in the linear sparsity regime. For this purpose, we suggest using every normalized row of ${\bf Y}$ as initializations for ${\bf q}$, and solving a sequence of $p$ nonconvex programs \eqref{eqn:syn_l1_l2} by the ADM algorithm.

To get an intuition of why our initialization works, recall the planted sparse model
$\mc S = \mathrm{span}( \mb x_0, \mb g_1, \dots, \mb g_{n-1} )$ and suppose
\begin{align}\label{eqn:Y-bar}
	\overline{\mb Y}\;= \;\brac{ \mb x_0 \mid \mb g_1 \mid \dots \mid \mb g_{n-1} } \;\in\; \R^{p \times n}.
\end{align}
If we take a row $\overline{\mb y}^i$ of $\overline{\mb Y}$, in which $x_0(i)$ is nonzero, then $x_0(i) = \Theta\left(1/\sqrt{\theta p}\right)$. Meanwhile, the entries of $\mb g_1(i), \dots \mb g_{n-1}(i)$ are all $\mc N(0,1/p)$, and so their magnitude have size about $1/\sqrt{p}$. Hence, when $\theta$ is not too large, $x_0(i)$ will be somewhat bigger than most of the other entries in $\overline{\mb y}^i$. Put another way, {\em $\overline{\mb y}^i$ is biased towards the first standard basis vector $\mb e_1$.} Now, under our probabilistic model assumptions, $\overline{\mb Y}$ is very well conditioned: $\overline{\mb Y}^\top \overline{\mb Y} \approx \mb I$.\footnote{This is the common heuristic that ``tall random matrices are well conditioned'' {\cite{veryshynin2011matrix}}.} Using the Gram-Schmidt process\footnote{...QR decomposition in general with restriction that $R_{11} = 1$.}, we can find an orthonormal basis $\mb Y$ for $\mc S$ via: 
\begin{eqnarray} \label{eqn:Y-Z-R}
\overline{\mb Y} &=& \mb Y \mb R,
\end{eqnarray}
where $\mb R$ is upper triangular, and $\mb R$ is itself well-conditioned: $\mb R \approx \mb I$. Since the $i$-th row $\overline{\mb y}^i$ of $\overline{\mb Y}$ is biased in the direction of $\mb e_1$ and $\mb R$ is well-conditioned, the $i$-th row $\mb y^i$ of $\mb Y$ is also biased in the direction of $\mb e_1$. In other words, with this canonical orthobasis $\mb Y$ for the subspace, {\em the $i$-th row of $\mb Y$ is biased in the direction of the global optimizer}. The heuristic arguments are made rigorous in Appendix~\ref{sec:app_bases} and Appendix~\ref{app:initialization}. 

What if we are handed some other basis $\widehat{\mb Y} = \mb Y \mb U$, where $\mb U$ is an arbitary orthogonal matrix? Suppose $\mb q_\star$ is a global optimizer to \eqref{eqn:syn_l1_l2} with the input matrix $\mb Y$, then it is easy to check that, $\mb U^\top \mb q_\star$ is a global optimizer to \eqref{eqn:syn_l1_l2} with the input matrix $\widehat{\mb Y}$. Because 
\begin{align*}
\innerprod{(\mb Y \mb U)^\top \mb e_i}{\mb U^\top \mb q_\star} = \innerprod{\mb Y^\top \mb e_i}{\mb q_\star}, 
\end{align*}
our initialization is {\em invariant} to any rotation of the orthobasis. Hence, {\em even if we are handed an arbitrary orthobasis for $\mc S$, the $i$-th row is still biased in the direction of the global optimizer}. 
\par\smallskip

\noindent\textbf{Rounding by linear programming (LP). }Let $\overline{\mb q}$ denote the output of Algorithm \ref{ADM}. As illustrated in Fig. \ref{fig:proof_sketch}, we will prove that with our particular initialization and an appropriate choice of $\lambda$, ADM algorithm uniformly moves towards the optimal over a large portion of the sphere, and its solution falls within a certain small radius of the globally optimal solution $\mb q_\star$ to \eqref{eqn:syn_l1_l2}. To exactly recover $\mb q_\star$, or equivalently to recover the exact sparse vector $\mb x_0 = \gamma \mb Y \mb q_\star$ for some $\gamma \neq 0$, we solve the linear program 
\begin{equation} \label{eqn:rounding}
\min_{\mb q} \norm{\mb Y \mb q}_1 \quad \text{s.t.} \quad \left\langle \mb r, \mb q \right\rangle = 1 
\end{equation}
with $\mb r = \overline{\mb q}$. Since the feasible set $\{\mb q \mid \left\langle \overline{\mb q}, \mb q \right\rangle = 1\}$ is essentially the tangent space of the sphere $\bb S^{n-1}$ at $\overline{\mb q}$, whenever $\overline{\mb q}$ is close enough to $\mb q_\star$, one should expect that the optimizer of \eqref{eqn:rounding} exactly recovers $\mb q_\star$ and hence $\mb x_0$ up to scale. We will prove that this is indeed true under appropriate conditions.


\section{Main Results and Sketch of Analysis}\label{sec:analysis}
\subsection{Main Results}

In this section, we describe our main theoretical result, which shows that w.h.p. the algorithm described in the previous section succeeds. 

\begin{theorem} \label{thm:recovery} 
Suppose that $\mc S$ obeys the planted sparse model, and let the columns of $\mb Y$ form an arbitrary orthonormal basis for the subspace $\mc S$. Let $\mb y^1, \dots, \mb y^p \in \R^n$ denote the (transposes of) the rows of $\mb Y$. Apply Algorithm \ref{ADM} with $\lambda = 1/\sqrt{p}$, using initializations $\mb q^{(0)} = \mb y^1/\norm{\mb y^1}_2 , \dots, \mb y^p / \norm{\mb y^p}_2$, to produce outputs $\overline{\mb q}_1, \dots, \overline{\mb q}_p$. Solve the linear program \eqref{eqn:rounding} with $\mb r = \overline{\mb q}_1, \dots, \overline{\mb q}_p$, to produce $\widehat{\mb q}_1, \dots, \widehat{\mb q}_p$. Set $i^\star \in \arg \min_i \norm{\mb Y \widehat{\mb q}_i}_1$. Then 
\begin{equation}
\mb Y \widehat{\mb q}_{i^\star} = \gamma \mb x_0,
\end{equation}
for some $\gamma \ne 0$ with probability at least $1- cp^{-2}$, provided
\begin{equation}
p \ge Cn^4\log n, \qquad\text{and} \qquad  \frac{1}{\sqrt{n}} \le \theta\le \theta_0.
\end{equation}
Here $C, c$ and $\theta_0$ are positive constants. 
\end{theorem}

\begin{remark} 
We can see that the result in Theorem~\ref{thm:recovery} is suboptimal in sample complexity compared to the global optimality result in Theorem~\ref{thm:global} and Barak et al.'s result \cite{barak2013rounding} (and the subsequent work~\cite{hopkins2015speeding}). For successful recovery, we require $p \ge \Omega\paren{n^4 \log n}$, while the global optimality and Barak et al. demand $p \geq \Omega\paren{n}$ and $p \ge \Omega\paren{n^2}$, respectively. Aside from possible deficiencies in our current analysis, compared to Barak et al., we believe this is still the first practical and efficient method which is guaranteed to achieve $\theta\sim \Omega(1)$ rate. The lower bound on $\theta$ in Theorem~\ref{thm:recovery} is mostly for convenience in the proof; in fact, the LP rounding stage of our algorithm already succeeds w.h.p. when $\theta \in O\paren{ 1/\sqrt{n}}$.
\end{remark}

\subsection{A Sketch of Analysis} \label{sec:analysis_sketch}
In this section, we briefly sketch the main ideas of proving our main result in Theorem~\ref{thm:recovery}, to show that the ``initialization + ADM + LP rounding'' pipeline recovers $\mb x_0$ under the stated technical conditions, as illustrated in Fig. \ref{fig:proof_sketch}. The proof of our main result requires rather detailed technical analysis of the iteration-by-iteration properties of Algorithm \ref{ADM}, most of which is deferred to the appendices. 

\begin{figure}[!htbp]
\begin{center}
\includegraphics[width=0.8\linewidth]{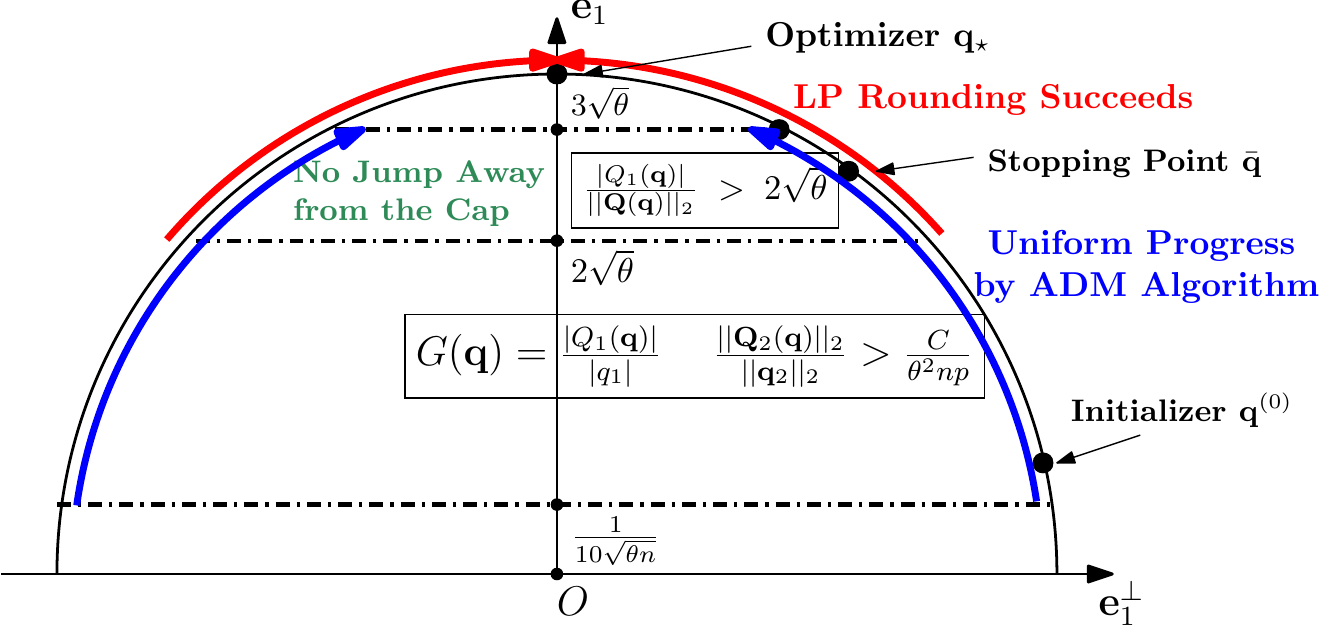}
\end{center}
\caption{An illustration of the proof sketch for our ADM algorithm.} 
\label{fig:proof_sketch}
\end{figure}

As noted in Section~\ref{sec:algorithm}, the ADM algorithm is invariant to change of basis. So w.l.o.g., let us assume 
$\overline{\mb Y} = \brac{\mb x_0\mid \mb g_1\mid \dots \mid \mb g_{n-1} }$ and let $\mb Y$ to be its orthogonalization, i.e., \footnote{Note that with probability one, the inverse matrix square-root in $\mb Y$ is well defined. So $\mb Y$ is well defined w.h.p. (i.e., except for $\mb x_0 = 0$). See more quantitative characterization of $\mb Y$ in Appendix~\ref{sec:app_bases}. }
\begin{align} \label{eqn:Y_orth}
\mb Y = \brac{\frac{\mb x_0}{\norm{\mb x_0}_2} \mid \mc P_{\mb x_0^\perp} \mb G\paren{\mb G^\top \mc P_{\mb x_0^{\perp}} \mb G}^{-1/2} }. 
\end{align}
When $p$ is large, $\overline{\mb Y}$ is nearly orthogonal, and hence $\ol{\mb Y}$ is very close to $\mb Y$. Thus, in our proofs, whenever convenient, we make the arguments on $\ol{\mb Y}$ first and then ``propagate'' the quantitative results onto $\mb Y$ by perturbation arguments. With that noted, let $\mb y^1,\cdots,\mb y^p$ be the transpose of the rows of $\mb Y$, and note that these are all independent random vectors. To prove the result of Theorem \ref{thm:recovery}, we need the following results. First, given the specified $\mb Y$, we show that our initialization is biased towards the global optimum:
\begin{proposition}[Good initialization]\label{prop:initialization}
Suppose $\theta > 1/\sqrt{n}$ and $ p \ge Cn$. It holds with probability at least $1 -cp^{-2}$ that at least one of our $p$ initialization vectors suggested in Section \ref{sec:algorithm}, say $\mb q_i^{(0)} = \mb y^i /\norm{\mb y^i}_2$, obeys 
\begin{align}
\abs{\innerprod{\frac{\mb y^i}{\norm{\mb y^i}_2}}{\mb e_1}} \ge \frac{1}{10\sqrt{\theta n}}. 
\end{align}
Here $C, c$ are positive constants. 
\end{proposition}
\begin{proof}
	See Appendix \ref{app:initialization}.
\end{proof}
Second, we define a vector-valued random process $\mb Q(\mb q)$ on $\mb q\in \bb S^{n-1}$, via
\begin{align}\label{eqn:Q-process}
	\mb Q(\mb q) = \frac{1}{p} \sum_{i=1}^p {\mb y}^i S_\lambda\brac{  \mb q^\top {\mb y}^i  },
\end{align} 
so that based on \eqref{eqn:closed-form}, one step of the ADM algorithm takes the form:
\begin{align}\label{eqn:adm-one-step}
	\mb q^{(k+1)} = \frac{\mb Q\paren{\mb q^{(k)} }}{ \norm{\mb Q\paren{\mb q^{(k)} } }_2}
\end{align}
This is a very favorable form for analysis: the term in the numerator $\mb Q\paren{\mb q^{(k)} }$ is a sum of $p$ independent random vectors with $\mb q^{(k)}$ viewed as fixed. We study the behavior of the iteration \eqref{eqn:adm-one-step} through the random process $\mb Q\paren{\mb q^{(k)} }$. We want to show that w.h.p. the ADM iterate sequence $\mb q^{(k)}$ converges to some small neighborhood of $\pm \mb e_1$, so that the ADM algorithm plus the LP rounding (described in Section~\ref{sec:algorithm}) successfully retrieves the sparse vector $\mb x_0/\|\mb x_0\| = \mb Y \mb e_1$. Thus, we hope that in general, $\mb Q(\mb q)$ is more concentrated on the first coordinate than $\mb q\in \bb S^{n-1}$. Let us partition the vector $\mb q$ as $\mb q = [q_1; \mb q_2]$, with $q_1 \in \reals$ and $\mb q_2 \in \reals^{n-1}$; and correspondingly $\mb Q(\mb q) = [Q_1(\mb q); \mb Q_2(\mb q)]$. The inner product of $\mb Q(\mb q) / \norm{\mb Q(\mb q)}_2$ and $\mb e_1$ is strictly larger than the inner product of $\mb q$ and $\mb e_1$ if and only if 
\begin{equation*}
\frac{ \abs{ Q_1(\mb q) }}{ \abs{q_1} } > \frac{ \norm{\mb Q_2(\mb q) }_{2} }{ \norm{\mb q_2}_{2} }. 
\end{equation*}
In the following proposition, we show that w.h.p., this inequality holds uniformly over a significant portion of the sphere 
\begin{align}\label{eqn:Gamma-set}
	\Gamma \doteq \Brac{\mb q\in \bb S^{n-1} \mid \frac{1}{10\sqrt{n \theta}} \le \abs{q_1} \le 3\sqrt{\theta}, \norm{\mb q_2}_2 \geq \frac{1}{10}},
\end{align}
so the algorithm moves in the correct direction. Let us define the gap $G(\mb q)$ between the two quantities $\abs{Q_1(\mb q)}/\abs{q_1}$ and $\norm{\mb Q_2(\mb q)}_2/ \norm{\mb q_2}_2$ as
\begin{align}\label{eqn:gap-G}
	G(\mb q) \doteq \frac{\abs{Q_1(\mb q)}}{\abs{q_1}} - \frac{\norm{\mb Q_2(\mb q) }_2}{\norm{\mb q_2}_2},
\end{align}
and we show that the following result is true:

\begin{proposition}[Uniform lower bound for finite sample gap]\label{prop:gap-bound-Y'}
There exists a constant $\theta_0 \in (0, 1)$, such that when $p \ge Cn^4 \log n$, the estimate 
\begin{align*}
\inf_{\mb q \in \Gamma} G(\mb q) \;&\geq \;\frac{1}{10^4\theta^2 np } 
\end{align*}
holds with probability at least $1 -cp^{-2}$, provided $\theta\in \paren{1/\sqrt{n},\theta_0}$. Here $C, c$ are positive constants. 
\end{proposition}
\begin{proof}
	See Appendix~\ref{app:gap-finite}.
\end{proof}

Next, we show that whenever $\abs{q_1}\geq 3\sqrt{\theta}$, w.h.p. the iterates stay in a ``safe region'' with $\abs{q_1}\geq 2\sqrt{\theta}$ which is enough for LP rounding \eqref{eqn:rounding} to succeed. 

\begin{proposition}[Safe region for rounding] \label{lem:safe}
There exists a constant $\theta_0 \in (0, 1)$, such that when $p \ge Cn^4 \log n$, it holds with probability at least $1 - cp^{-2}$ that 
\begin{align*}
\frac{\abs{Q_1(\mb q)}}{\norm{\mb Q(\mb q)}_2} \; \geq \; 2\sqrt{\theta}
\end{align*}
for all $\mb q\in \bb S^{n-1}$ satisfying $\abs{q_1} > 3 \sqrt{\theta}$, provided $\theta\in \paren{1/\sqrt{n},\theta_0}$. Here $C, c$ are positive constants.
\end{proposition}
\begin{proof}
	See Appendix \ref{app:safe-region}.
\end{proof}
In addition, the following result shows that the number of iterations for the ADM algorithm to reach the safe region can be bounded grossly by $O(n^4\log n)$ w.h.p..
\begin{proposition}[Iteration complexity of reaching the safe region]\label{prop:iter-complexity}
There is a constant $\theta_0 \in (0, 1)$, such that when $p \ge Cn^4 \log n$, it holds with probability at least $1 -cp^{-2}$ that the ADM algorithm in Algorithm \ref{ADM}, with any initialization $\mb q^{(0)}\in \bb S^{n-1}$ satisfying $\abs{q_1^{(0)}}\geq \frac{1}{10\sqrt{\theta n}}$, will produce some iterate $\overline{\mb q}$ with $\abs{\bar{q}_1}> 3\sqrt{\theta} $ at least once in at most $O(n^4 \log n)$ iterations, provided $\theta \in \paren{1/\sqrt{n}, \theta_0 }$. Here $C, c$ are positive constants. 
\end{proposition}
\begin{proof}
	See Appendix \ref{app:iter_cplx}.
\end{proof}
Moreover, we show that the LP rounding \eqref{eqn:rounding} with input $\mb r = \ol{\mb q}$ exactly recovers the optimal solution w.h.p., whenever the ADM algorithm returns a solution $\ol{\mb q}$ with first coordinate $\abs{\ol{q}_1}>2\sqrt{\theta}$.

\begin{proposition}[Success of rounding] \label{lem:rounding}
There is a constant $\theta_0 \in (0, 1)$, such that when $p \geq C n$, the following holds with probability at least $1 - cp^{-2}$ provided $\theta \in (1/\sqrt{n}, \theta_0)$: Suppose the input basis is $\mb Y$ defined in~\eqref{eqn:Y_orth} and the ADM algorithm produces an output $\ol{\mb q}\in \bb S^{n-1}$ with $|\overline{q}_1|>2\sqrt{\theta}$. Then the rounding procedure with $\mb r= \overline{\mb q}$ returns the desired solution $\pm \mb e_1$. Here $C, c$ are positive constants. 
\end{proposition}
\begin{proof}
	See Appendix~\ref{app:rounding}.
\end{proof}


Finally, given $p \ge Cn^4 \log n$ for a sufficiently large constant $C$, we combine all the results above to complete the proof of Theorem \ref{thm:recovery}.

\begin{proof}[Proof of Theorem \ref{thm:recovery}]

W.l.o.g., let us again first consider $\ol{\mb Y}$ as defined in~\eqref{eqn:Y-bar} and its orthogonalization $\mb Y$ in a ``natural/canonical'' form \eqref{eqn:Y_orth}. We show that w.h.p. our algorithmic pipeline described in Section \ref{sec:algorithm} exactly recovers the optimal solution up to scale, via the following argument:
\begin{enumerate}
\item \textbf{Good initializers}. Proposition \ref{prop:initialization} shows that w.h.p., at least one of the $p$ initialization vectors, say $\mb q_i^{(0)} = \mb y^i/\norm{\mb y^i}_2 $, obeys
\begin{align*}
\abs{\innerprod{ \mb q_i^{(0)} }{\mb e_1}} \ge \frac{1}{10\sqrt{\theta n}},
\end{align*}
which implies that $\mb q_i^{(0)}$ is biased towards the global optimal solution. 

\item \textbf{Uniform progress away from the equator}. By Proposition \ref{prop:gap-bound-Y'}, for any $\theta \in (1/\sqrt{n}, \theta_0)$  with a constant $\theta_0 \in (0, 1)$, 
\begin{align}\label{eqn:gap-main}
G(\mb q) = \frac{\abs{Q_1(\mb q)}}{\abs{q_1}} - \frac{\norm{\mb Q_2(\mb q)}_2}{\norm{\mb q}_2}\;&\geq \;\frac{1}{10^4\theta^2 np }
\end{align}
holds uniformly for all $\mb q\in \bb S^{n-1}$ in the region $\frac{1}{10\sqrt{\theta n}} \leq \abs{q_1} \leq 3 \sqrt{\theta}$ w.h.p.. This implies that with an input $\mb q^{(0)}$ such that $\abs{q_1^{(0)} }\geq \frac{1}{10 \sqrt{\theta n} } $, the ADM algorithm will eventually obtain a point $\mb q^{(k)}$ for which $\abs{q^{(k)} } \geq 3\sqrt{\theta}$, if sufficiently many iterations are allowed.

\item \textbf{No jumps away from the caps}. Proposition \ref{lem:safe} shows that for any $\theta \in (1/\sqrt{n}, \theta_0)$ with a  constant $\theta_0 \in (0, 1)$, w.h.p., 
\begin{align*}
\frac{Q_1(\mb q)}{\norm{\mb Q(\mb q)}_2}\;\geq \; 2\sqrt{\theta}
\end{align*}
holds for all $\mb q\in \bb S^{n-1}$ with $\abs{q_1}\geq 3 \sqrt{\theta}$. This implies that once $|q_1^{(k)}| \geq 3\sqrt{\theta}$ for some iterate $k$, all the future iterates produced by the ADM algorithm stay in a ``spherical cap'' region around the optimum with $\abs{q_1} \geq 2 \sqrt{\theta} $.
\item \textbf{Location of stopping points}. As shown in Proposition \ref{prop:iter-complexity}, w.h.p., the strictly positive gap $G(\mb q)$ in \eqref{eqn:gap-main} ensures that one needs to run at most $O\paren{n^4 \log n}$ iterations to first encounter an iterate $\mb q^{(k)}$ such that $|q^{(k)}_1| \ge 3\sqrt{\theta}$. Hence, the steps above imply that, w.h.p., Algorithm~\ref{ADM} fed with the proposed initialization scheme successively produces iterates $\ol{\mb q}\in \bb S^{n-1}$ with its first coordinate $\abs{\ol{q}_1} \ge 2\sqrt{\theta}$ after $O\paren{n^4 \log n}$ steps.  
\item \textbf{Rounding succeeds when $|r_1| \geq 2 \sqrt{\theta}$}. Proposition \ref{lem:rounding} proves that w.h.p., the LP rounding \eqref{eqn:rounding} with an input $\mb r = \ol{\mb q}$ produces the solution $\pm \mb x_0$ up to scale.
\end{enumerate}

Taken together, these claims imply that from at least one of the initializers $\mb q^{(0)}$, the ADM algorithm will produce an output $\ol{\mb q}$ which is accurate enough for LP rounding to exactly return $\mb x_0/\|\mb x_0\|_2$. On the other hand, our $\ell^1/\ell^2$ optimality theorem (Theorem~\ref{thm:global}) implies that $\pm \mb x_0$ are the unique vectors with the smallest $\ell^1$ norm among all unit vectors in the subspace. Since w.h.p. $\mb x_0/\|\mb x_0\|_2$ is among the $p$ unit vectors $\widehat{\mb q}_1, \dots, \widehat{\mb q}_p$ our $p$ row initializers finally produce, our minimal $\ell^1$ norm selector will successfully locate $\mb x_0/\|\mb x_0\|_2$ vector. 

For the general case when the input is an arbitrary orthonormal basis $\widehat{\mb Y} = \mb Y \mb U$ for some orthogonal matrix $\mb U$, the target solution is $\mb U^\top \mb e_1$. The following technical pieces are perfectly parallel to the argument above for $\mb Y$. 
\begin{enumerate}
\item Discussion at the end of Appendix~\ref{app:initialization} implies that w.h.p., at least one row of $\widehat{\mb Y}$ provides an initial point $\mb q^{(0)}$ such that $\abs{\innerprod{\mb q^{(0)}}{\mb U^\top \mb e_1}} \ge \frac{1}{10\sqrt{\theta n}}$. 
\item Discussion following Proposition~\ref{prop:gap-bound-Y'} in Appendix~\ref{app:gap-finite} indicates that for all $\mb q$ such that $\frac{1}{10\sqrt{\theta n}} \le \abs{\innerprod{\mb q}{\mb U^\top \mb e_1}} \le 3\sqrt{\theta}$, there is a strictly positive gap, indicating steady progress towards a point $\mb q^{(k)}$ such that $\abs{\innerprod{\mb q^{(k)}}{\mb U^\top \mb e_1}} \ge 3\sqrt{\theta}$. 
\item Discussion at the end of Appendix~\ref{app:safe-region} implies that once $\mb q$ satisfies $\abs{\innerprod{\mb q}{\mb U^\top \mb e_1}}$, the next iterate will not move far away from the target: 
\begin{align*}
\abs{\innerprod{\mb Q\paren{\mb q; \widehat{\mb Y}}/\norm{\mb Q\paren{\mb q; \widehat{\mb Y}}}_2 }{\mb U^\top \mb e_1}} \;\geq\; 2\sqrt{\theta}. 
\end{align*}
\item Repeating the argument in Appendix~\ref{app:iter_cplx} for general input $\widehat{\mb Y}$ shows it is enough to run the ADM algorithm $O\paren{n^4 \log n}$ iterations to cross the range $\frac{1}{10\sqrt{\theta n}} \le \abs{\innerprod{\mb q}{\mb U^\top \mb e_1}} \le 3\sqrt{\theta}$. So the argument above together dictates that with the proposed initialization, w.h.p., the ADM algorithm produces an output $\ol{\mb q}$ that satisfies $\abs{\innerprod{\ol{\mb q}}{\mb U^\top \mb e_1}} \ge 2\sqrt{\theta}$, if we run at least $O\paren{n^4 \log n}$ iterations.  
\item Since the ADM returns $\ol{\mb q}$ satisfying $\abs{\innerprod{\overline{\mb q}}{\mb R^\top \mb e_1}} \ge 2\sqrt{\theta}$, discussion at the end of Appendix~\ref{app:rounding} implies that we will obtain a solution $\mb q_\star = \pm \mb U^\top \mb e_1$ up to scale as the optimizer of the rounding program, exactly the target solution. 
\end{enumerate}
Hence, we complete the proof. 
\end{proof}

\begin{remark}
Under the planted sparse model, in practice the ADM algorithm with the proposed initialization converges to a global optimizer of~\eqref{eqn:huber-l2} that correctly recovers $\mb x_0$. In fact, simple calculation shows such desired point for successful recovery is indeed the only critical point of~\eqref{eqn:huber-l2} near the pole in Fig.~\ref{fig:proof_sketch}.  Unfortunately, using the current analytical framework, we did not succeed in proving such convergence in theory. Proposition~\ref{lem:safe} and~\ref{prop:iter-complexity} imply that after $O(n^4 \log n)$ iterations, however, the ADM sequence will stay in a small neighborhood of the target. Hence, we proposed to stop after $O(n^4 \log n)$ steps, and then round the output using the LP that provable recover the target, as implied by Proposition~\ref{lem:safe} and~\ref{lem:rounding}. So the LP rounding procedure is for the purpose of completing the theory, and seems not necessary in practice. We suspect alternative analytical strategies, such as the geometrical analysis that we will discuss in Section~\ref{sec:discussion}, can likely get around the artifact.
\end{remark}

\section{Experimental Results}\label{sec:exp}

In this section, we show the performance of the proposed ADM algorithm on both synthetic and real datasets. On the synthetic dataset, we show the phase transition of our algorithm on both the planted sparse and the dictionary learning models; for the real dataset, we demonstrate how seeking sparse vectors can help discover interesting patterns on face images. 

\subsection{Phase Transition on Synthetic Data}
For the planted sparse model, for each pair of $(k,p)$, we generate the $n$ dimensional subspace $\mathcal{S}\subset \mathbb{R}^p$ by direct sum of $\mb x_0$ and $\mb G$: $\mb x_0 \in \R^p$ is a $k$-sparse vector with uniformly random support and all nonzero entries equal to $1$, and $\mb G \in \R^{p \times (n-1)}$ is an i.i.d. Gaussian matrix distributed by $\mc N(0, 1/p)$. So one basis $\mb Y$ of the subspace $\mc S$ can be constructed by 
$
{\mb Y} = \mathtt{GS}\paren{\brac{{\mb x}_0,{\mb G}}}{\mb U},
$
where $\mathtt{GS}\paren{\cdot}$ denotes the Gram-Schmidt orthonormalization operator and ${\mb U}\in \R^{n\times n}$ is an arbitrary orthogonal matrix. For each $p$, we set the regularization parameter in \eqref{eqn:huber-l2} as $\lambda = 1/\sqrt{p}$, use all the normalized rows of ${\mb Y}$ as initializations of ${\mb q}$ for the proposed ADM algorithm, and run the alternating steps for $10^4$ iterations. We determine the recovery to be successful whenever $\norm{\mb x_0/\norm{\mb x_0}_2- \mb Y \mb q}_2 \le 10^{-2}$ for at least one of the $p$ trials (we set the tolerance relatively large as we have shown that LP rounding exactly recovers the solutions with approximate input). To determine the empirical recovery performance of our ADM algorithm, first we fix the relationship between $n$ and $p$ as $p=5n \log n$, and plot out the phase transition between $k$ and $p$. Next, we fix the sparsity level $\theta = 0.2$ (or $k = 0.2p$), and plot out the phase transition between $p$ and $n$. For each pair of $(p,k)$ or $(n,p)$, we repeat the simulation for $10$ times. Fig.~\ref{phase_transition:psv} shows both phase transition plots.
\begin{figure}[!htbp]
\centering
\includegraphics[width = 0.45\linewidth]{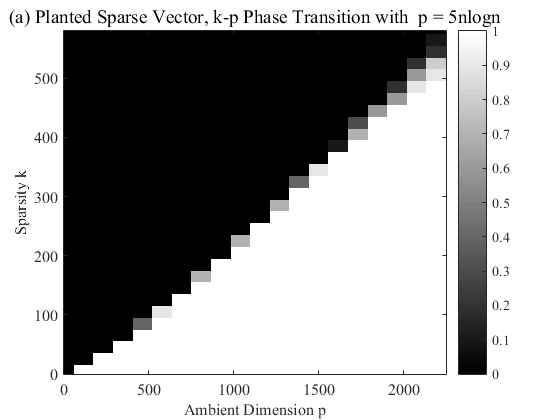}
\includegraphics[width = 0.45\linewidth]{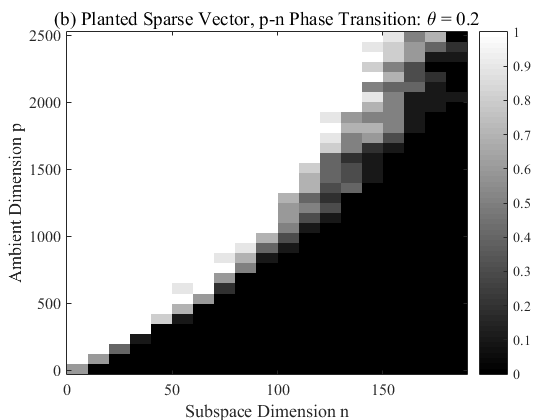}
\caption{Phase transition for the planted sparse model using the ADM algorithm: (a) with fixed relationship between $p$ and $n$: $p = 5n\log n$; (b) with fixed relationship between $p$ and $k$: $k = 0.2 p$. White indicates success and black indicates failure.} 
\label{phase_transition:psv}
\end{figure}

\begin{figure}[!htbp]
\centering
\includegraphics[width = 0.45\linewidth]{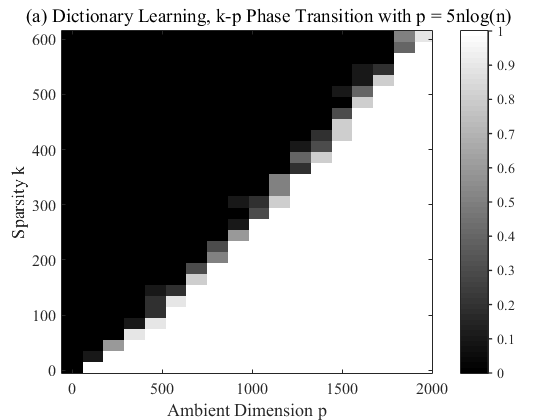}
\includegraphics[width = 0.45\linewidth]{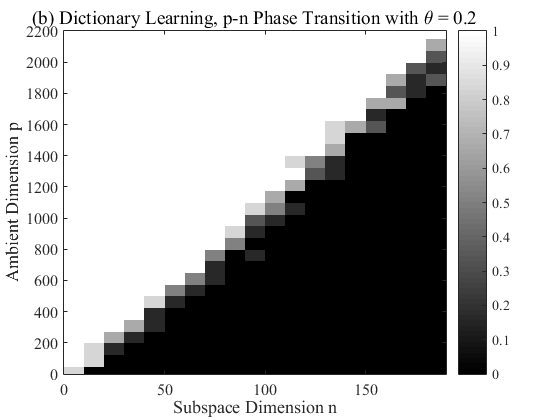}
\caption{Phase transition for the dictionary learning model using the ADM algorithm: (a) with fixed relationship between $p$ and $n$: $p = 5n\log n$; (b) with fixed relationship between $p$ and $k$: $k = 0.2 p$. White indicates success and black indicates failure.} 
\label{phase_transition:dl}
\end{figure}

We also experiment with the complete dictionary learning model as in~\cite{spielman2013exact} (see also~\cite{sun2015complete}). Specifically, the observation is assumed to be $\mb Y = \mb A_0 \mb X_0$, where $\mb A_0$ is a square, invertible matrix, and $\mb X_0$ a $n \times p$ sparse matrix. Since $\mb A_0$ is invertible, the row space of $\mb Y$ is the same as that of $\mb X_0$. For each pair of $(k,n)$, we generate ${\mb X}_0 = \brac{{\mb x}_1,\cdots,{\mb x}_n}^\top$, where each vector ${\mb x}_i\in \bb R^p$ is $k$-sparse with every nonzero entry following i.i.d. Gaussian distribution, and construct the observation by
$
{\mb Y}^\top = \mathtt{GS}\paren{{\mb X}_0^\top}{\mb U}^\top. 
$
We repeat the same experiment as for the planted sparse model described above. The only difference is that here we determine the recovery to be successful as long as one sparse row of $\mb X_0$ is recovered by one of those $p$ programs. Fig.~\ref{phase_transition:dl} shows both phase transition plots.

Fig.~\ref{phase_transition:psv}(a) and Fig.~\ref{phase_transition:dl}(a) suggest our ADM algorithm could work into the linear sparsity regime for both models, provided $p \ge \Omega(n \log n)$. Moreover, for both models, the $\log  n$ factor seems necessary for working into the linear sparsity regime, as suggested by Fig.~\ref{phase_transition:psv}(b) and Fig.~\ref{phase_transition:dl}(b): there are clear nonlinear transition boundaries between success and failure regions. For both models, $O(n \log n)$ sample requirement is near optimal: for the planted sparse model, obviously $p \ge \Omega(n)$ is necessary; for the complete dictionary learning model, \cite{spielman2013exact} proved that $p\geq \Omega(n\log n)$ is required for exact recovery. For the planted sparse model, our result $p \ge \Omega(n^4 \log n)$ is far from this much lower empirical requirement. Fig~\ref{phase_transition:psv}(b) further suggests that alternative reformulation and algorithm are needed to solve~\eqref{eqn:l1-l2} so that the optimal recovery guarantee as depicted in Theorem~\ref{thm:global} can be obtained. 

\subsection{Exploratory Experiments on Faces} \label{sec:face_exp}
It is well known in computer vision that the collection of images of a convex object only subject to illumination changes can be well approximated by a low-dimensional subspaces in raw-pixel space~\cite{basri2003lambertian}. We will play with face subspaces here. First, we extract face images of one person ($65$ images) under different illumination conditions. Then we apply \emph{robust principal component analysis} \cite{Candes2011-JACM} to the data and get a low dimensional subspace of dimension $10$, i.e., the basis $\mb Y \in \R^{32256\times 10}$. We apply the ADM + LP algorithm to find the sparsest elements in such a subspace, by randomly selecting $10\%$ rows of $\mb Y$ as initializations for ${\mb q}$. We judge the sparsity in the $\ell^1/\ell^2$ sense, that is, the sparsest vector $\widehat{\mb x}_0={\mb Y\mb q}^\star$ should produce the smallest $\norm{\mb Y \mb q}_1/\norm{\mb Y \mb q}_2$ among all results. Once some sparse vectors are found, we project the subspace onto orthogonal complement of the sparse vectors already found\footnote{The idea is to build a sparse, orthonormal basis for the subspace in a greedy manner. }, and continue the seeking process in the projected subspace. Fig.~\ref{face_exp_1}(Top) shows the first four sparse vectors we get from the data. We can see they correspond well to different extreme illumination conditions. We also implemented the spectral method (with the LP post-processing) proposed in~\cite{hopkins2015speeding} for comparison under the same protocol. The result is presented as Fig.~\ref{face_exp_1}(Bottom): the ratios $\norm{\cdot}_{\ell^1}/\norm{\cdot}_{\ell^2}$ are significantly higher, and the ratios $\norm{\cdot}_{\ell^4}/\norm{\cdot}_{\ell^2}$ (this is the metric to be maximized in~\cite{hopkins2015speeding} to promote sparsity) are significantly lower. By these two criteria the spectral method with LP rounding consistently produces vectors with higher sparsity levels under our evaluation protocol. Moreover, the resulting images are harder to interpret physically. 

\begin{figure}[!htbp]
\begin{center}
\includegraphics[width=0.8\textwidth]{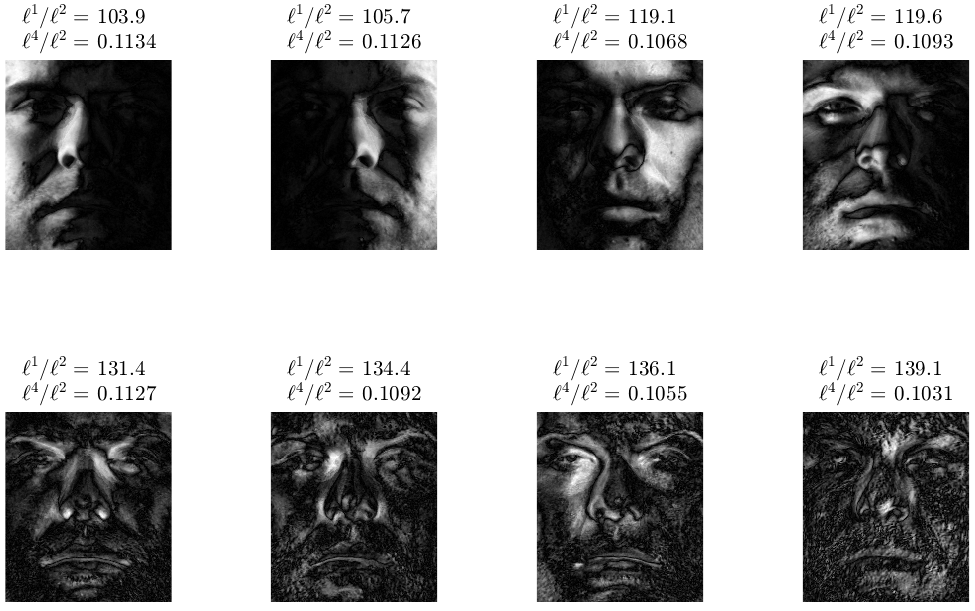}\\
\end{center}
\caption{The first four sparse vectors extracted for one person in the Yale B database under different illuminations. (Top) by our ADM algorithm; (Bottom) by the speeding-up SOS algorithm proposed in~\cite{hopkins2015speeding}. } 
\label{face_exp_1}
\end{figure}

Second, we manually select ten different persons' faces under the normal lighting condition. Again, the dimension of the subspace is $10$ and ${\mb Y}\in\R^{32256\times 10}$. We repeat the same experiment as stated above. Fig.~\ref{face_exp_2} shows four sparse vectors we get from the data. Interestingly, the sparse vectors roughly correspond to differences of face images concentrated around facial parts that different people tend to differ from each other, e.g., eye brows, forehead hair, nose, etc. By comparison, the vectors returned by the spectral method~\cite{hopkins2015speeding} are relatively denser and the sparsity patterns in the images are less structured physically. 

\begin{figure}[!htbp]
\begin{center}
\includegraphics[width=0.8\textwidth]{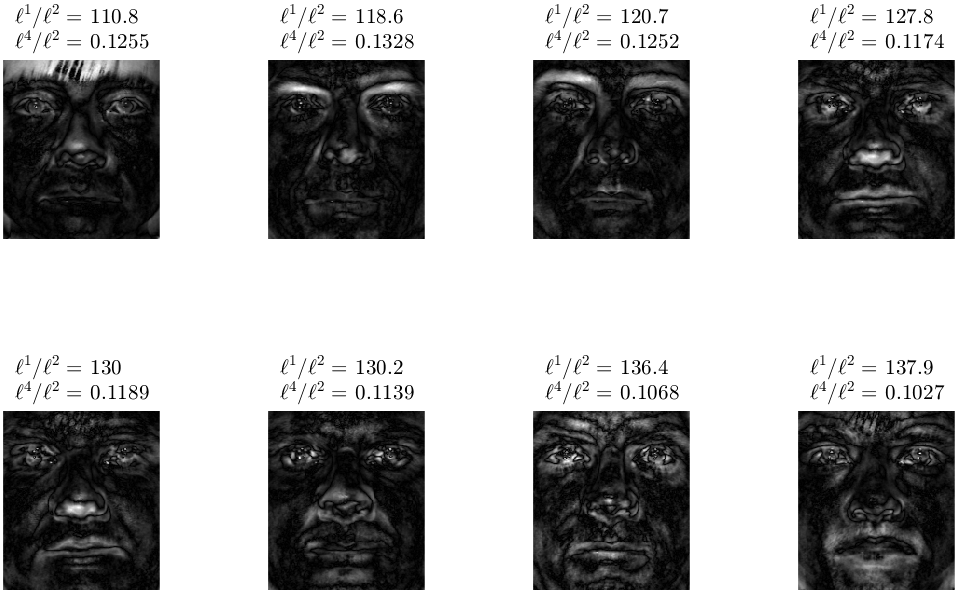} 
\end{center}
\caption{The first four sparse vectors extracted for $10$ persons in the Yale B database under normal illuminations. (Top) by our ADM algorithm; (Bottom) by the speeding-up SOS algorithm proposed in~\cite{hopkins2015speeding}. } 
\label{face_exp_2}
\end{figure}
In sum, our algorithm seems to find useful sparse vectors for potential applications, such as peculiarity discovery in first setting, and locating differences in second setting. Nevertheless, the main goal of this experiment is to invite readers to think about similar pattern discovery problems that might be cast as the problem of seeking sparse vectors in a subspace. The experiment also demonstrates in a concrete way the practicality of our algorithm, both in handling data sets of realistic size and in producing meaningful results even beyond the (idealized) planted sparse model that we adopted for analysis. 


\section{Connections and Discussion}\label{sec:discussion}
For the planted sparse model, there is a substantial performance gap in terms of $p$-$n$ relationship between the our optimality theorem (Theorem~\ref{thm:global}), empirical simulations, and guarantees we have obtained via efficient algorithm (Theorem~\ref{thm:recovery}). More careful and tighter analysis based on decoupling~\cite{de1999decoupling} and chaining~\cite{talagrand2014upper, luh15dictionary} and geometrical analysis described below can probably help bridge the gap between our theoretical and empirical results. Matching the theoretical limit depicted in Theorem~\ref{thm:global} seems to require novel algorithmic ideas. The random models we assume for the subspace can be extended to other random models, particularly for dictionary learning where all the bases are sparse (e.g., Bernoulli-Gaussian random model). 

This work is part of a recent surge of research efforts on deriving provable and practical nonconvex algorithms to central problems in modern signal processing and machine learning. These problems include low-rank matrix recovery/completion \cite{jain2013low,hardt2013provable,hardt2014fast,hardt2014understanding,jain2014fast,netrapalli2014non,zheng2015convergent,tu2015low,chen2015fast}, tensor recovery/decomposition \cite{jain2014provable,anandkumar2014guaranteed,anandkumar2014analyzing,anandkumar2015tensor,ge2015escaping}, phase retrieval \cite{netrapalli2013phase,candes2014wirtinger,chen2015solving,sun2016geometric}, dictionary learning \cite{arora2013new,agarwal2013learning,agarwal2013exact,arora2014more,arora2015simple,sun2015complete}, and so on.\footnote{The webpage \url{http://sunju.org/research/nonconvex/} maintained by the second author contains pointers to the growing list of work in this direction. } Our approach, like the others, is to start with a carefully chosen, problem-specific initialization, and then perform a local analysis of the subsequent iterates to guarantee convergence to a good solution. In comparison, our subsequent work on complete dictionary learning~\cite{sun2015complete} and generalized phase retrieval~\cite{sun2016geometric} has taken a geometrical approach by characterizing the function landscape and designing efficient algorithm accordingly. The geometric approach has allowed provable recovery via efficient algorithms, with an \emph{arbitrary initialization}. The article~\cite{sun2015nonconvex} summarizes the geometric approach and its applicability to several other problems of interest. 

A hybrid of the initialization and the geometric approach discussed above is likely to be a powerful computational framework. To see it in action for the current planted sparse vector problem, in Fig.~\ref{fig:landscape}
\begin{figure}[!htbp]
\centering
\begin{subfigure}
    \centering
    	\includegraphics[width = 0.45\linewidth]{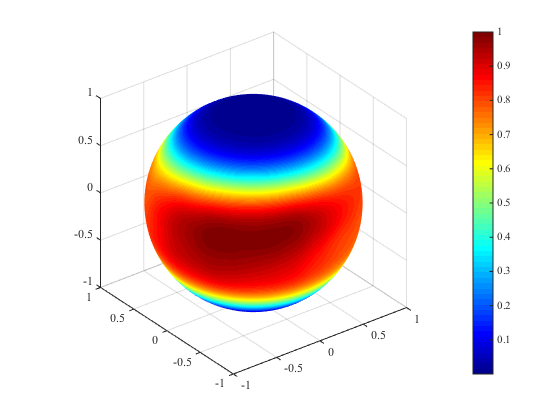}
\end{subfigure}
\begin{subfigure}
    \centering
    	\includegraphics[width = 0.45\linewidth]{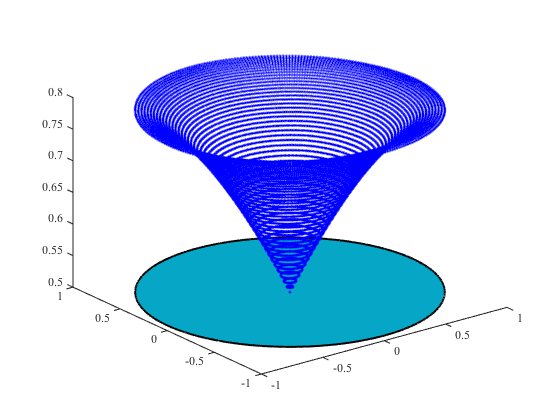}
\end{subfigure}
\caption{Function landscape of $f(\mb q)$ with $\theta =0.4$ for $n = 3$. (Left) $f(\mb q)$ over the sphere $\bb S^2$. Note that near the spherical caps around the north and south poles, there are no critical points and the gradients are always nonzero; (Right) Projected function landscape by projecting the upper hemisphere onto the equatorial plane. Mathematically the function $g(\mb w) : \mb e_3^\perp \mapsto \bb R$ obtained via the reparameterization $\mb q(\mb w) =  [\mb w; \sqrt{1 - \|\mb w\|^2 }]$. Corresponding to the left, there is no undesired critical point around $\mb 0$ within a large radius. } 
\label{fig:landscape}
\end{figure}
we provide the asymptotic function landscape (i.e., $p \to \infty$) of the Huber loss on the sphere $\bb S^{2}$ (aka the relaxed formulation we tried to solve~\eqref{eqn:huber-l2}). It is clear that with an initialization that is biased towards either the north or the south pole, we are situated in a region where the gradients are always nonzero and points to the favorable directions such that many reasonable optimization algorithms can take the gradient information and make steady progress towards the target. This will probably ease the algorithm development and analysis, and help yield tight performance guarantees.  

We provide a very efficient algorithm for finding a sparse vector in a subspace, with strong guarantee. Our algorithm is practical for handling large datasets---in the experiment on the face dataset, we successfully extracted some meaningful features from the human face images. However, the potential of seeking sparse/structured element in a subspace seems largely unexplored, despite the cases we mentioned at the start. We hope this work could inspire more application ideas.

%

\section*{Acknowledgement}
JS thanks the Wei Family Private Foundation for their generous support. We thank Cun Mu, IEOR Department of Columbia University, for helpful discussion and input regarding this work. We thank the anonymous reviewers for their constructive comments that helped improve the manuscript. This work was partially supported by grants ONR N00014-13-1-0492, NSF 1343282, NSF 1527809, and funding from the Moore and Sloan Foundations. 

\appendices
\section{Technical Tools and Preliminaries}

In this appendix, we record several lemmas that are useful for our analysis.

\begin{lemma}\label{lem:gaussian-integral}
Let $\psi(x)$ and $\Psi(x)$ to denote the probability density function (pdf) and the cumulative distribution function (cdf) for the standard normal distribution:
\begin{eqnarray*}
&(\text{Standard Normal pdf})&\psi(x) = \frac{1}{\sqrt{2\pi}}\exp\Brac{-\frac{x^2}{2}}\\
&(\text{Standard Normal cdf})&\Psi(x) = \frac{1}{\sqrt{2\pi}}\int_{-\infty}^{x}\exp\Brac{-\frac{t^2}{2}}dt,
\end{eqnarray*}
Suppose a random variable $X\sim \N(0,\sigma^2)$, with the pdf $f_{\sigma}(x) = \frac{1}{\sigma}\psi\paren{\frac{x}{\sigma}}$, then for any $t_2 > t_1$ we have
\begin{eqnarray*}
\int_{t_1}^{t_2}f_{\sigma}(x)dx   &=& \Psi\paren{\frac{t_2}{\sigma}} - \Psi\paren{\frac{t_1}{\sigma}}, \\
\int_{t_1}^{t_2}x f_{\sigma}(x)dx  &=& -\sigma\brac{\psi\paren{\frac{t_2}{\sigma}} - \psi\paren{\frac{t_1}{\sigma}}}, \\
\int_{t_1}^{t_2}x^2 f_{\sigma}(x)dx &=&\sigma^2\brac{\Psi\paren{\frac{t_2}{\sigma}} - \Psi\paren{\frac{t_1}{\sigma}}} -\sigma\brac{t_2\psi\paren{\frac{t_2}{\sigma}} - t_1\psi\paren{\frac{t_1}{\sigma}}}. 
\end{eqnarray*}
\end{lemma}


\begin{lemma}[Taylor Expansion of Standard Gaussian \emph{cdf} and \emph{pdf}]\label{lem:gaussian-taylor}
Assume $\psi(x)$ and $\Psi(x)$ be defined as above. There exists some universal constant $C_{\psi} > 0$ such that for any $x_0,~x\in \bb R$,
\begin{align*}
\abs{\psi(x) - \brac{\psi(x_0) - x_0\psi\paren{x_0}\paren{x - x_0}}} & \le C_\psi(x-x_0)^2, \\
\abs{\Psi(x) - \brac{\Psi(x_0) + \psi(x_0)(x-x_0)}} & \le C_\psi(x-x_0)^2.
\end{align*}
\end{lemma}

\begin{lemma}[Matrix Induced Norms]\label{lem:matrix-norms}
For any matrix $\mb A \in \bb R^{p \times n}$, the induced matrix norm from $\ell^p \to \ell^q$ is defined as 
\begin{align*}
\norm{\mb A}_{\ell^p\to \ell^q} \doteq \sup_{\norm{\mb x}_p = 1} \norm{\mb A\mb x}_q.
\end{align*}
In particular, let $\mb A = \brac{\mb a_1,\cdots, \mb a_n } = \brac{ \mb a^1,\cdots, \mb a^p }^\top $ , we have 
\begin{align*}
&\norm{\mb A}_{\ell^2 \to \ell^1} = \sup_{\norm{\mb x}_2=1} \sum_{k=1}^p \abs{\mb a_k^\top \mb x},\quad 
 \norm{\mb A}_{\ell^2 \to \ell^\infty} = \max_{1\leq k\leq p} \norm{\mb a^k}_2,\\
&\norm{\mb A\mb B}_{\ell^p \to \ell^r} \leq \norm{\mb A}_{\ell^q \to \ell^r} \norm{\mb B}_{\ell^p \to \ell^q},
\end{align*}
and $\mb B$ is any matrix of size compatible with $\mb A$. 
\end{lemma}

\begin{lemma}[Moments of the Gaussian Random Variable] \label{lem:gaussian_moment}
If $X \sim \mc N\left(0, \sigma_X^2\right)$, then it holds for all integer $m \geq 1$ that
\begin{align*}
\expect{\abs{X}^m} = \sigma_X^m \paren{m -1}!! \brac{ \sqrt{\frac{2}{\pi}} \indicator{m=2k+1 }+\indicator{m= 2k} } \leq \sigma_X^m \paren{m -1}!!,~k = \lfloor m/2 \rfloor.
\end{align*}
\end{lemma}



\begin{lemma}[Moments of the $\chi$ Random Variable] \label{lem:chi_moment}
If $X \sim \mc \chi\paren{n}$, i.e., $X = \norm{\mb x}_2$ for $\mb x \sim \mc N\paren{\mb 0, \mb I}$, then it holds for all integer $m \geq 1$ that 
\begin{align*}
\expect{X^m} = 2^{m/2} \frac{\Gamma\paren{m/2 + n/2}}{\Gamma\paren{n/2}} \leq m!! \; n^{m/2}. 
\end{align*}
\end{lemma}

\begin{lemma}[Moments of the $\chi^2$ Random Variable] \label{lem:chi_sq_moment}
If $X \sim \mc \chi^2\paren{n}$, i.e., $X = \norm{\mb x}_2^2$ for $\mb x \sim \mc N\paren{\mb 0, \mb I}$, then it holds for all integer $m \geq 1$ that 
\begin{align*}
\expect{X^m} = 2^m \frac{\Gamma\paren{m + n/2}}{\Gamma\paren{n/2}} =  \prod_{k=1}^m (n+2k-2)\leq \frac{m!}{2} (2n)^m .
\end{align*}
\end{lemma}


\begin{lemma}[Moment-Control Bernstein's Inequality for Random Variables \cite{foucart2013mathematical}] \label{lem:mc_bernstein_scalar}
Let $X_1, \dots, X_p$ be i.i.d.\ real-valued random variables. Suppose that there exist some positive numbers $R$ and $\sigma_X^2$ such that
\begin{align*}
\expect{\abs{X_k}^m} \leq \frac{m!}{2} \sigma_X^2 R^{m-2}, \; \; \text{for all integers $m \ge 2$}.
\end{align*} 
Let $S \doteq \frac{1}{p}\sum_{k=1}^p X_k$, then for all $t > 0$, it holds  that 
\begin{align*}
\prob{\abs{S - \expect{S}} \ge t} \leq 2\exp\left(-\frac{pt^2}{2\sigma_X^2 + 2Rt}\right).   
\end{align*}
\end{lemma}

\begin{lemma}[Moment-Control Bernstein's Inequality for Random Vectors \cite{sun2015complete}] \label{cor:vector-bernstein} Let $\mb x_1, \dots, \mb x_p \in \bb R^d$ be i.i.d. random vectors. Suppose there exist some positive number $R$ and $\sigma_X^2$ such that
\begin{align*}
\bb E\left[ \norm{\mb x_k }_2^m \right] \;\le\; \frac{m!}{2} \sigma_X^2R^{m-2}, \quad \text{for all integers $m \ge 2$}. 
\end{align*}
Let $\mb s = \frac{1}{p}\sum_{k=1}^p \mb x_k$, then for any $t > 0$, it holds that
\begin{align*}
\bb P\brac{\norm{\mb s - \bb E\brac{\mb s}}_2 \geq t} \; \leq \; 2(d+1)\exp\paren{-\frac{pt^2}{2\sigma_X^2+2Rt}}.
\end{align*}
\end{lemma}


\begin{lemma}[Gaussian Concentration Inequality] \label{lem:gaussian-concentration}
Let $\mb x \sim \mc N\paren{\mb 0, \mb I_p}$. Let $f:\bb R^p \mapsto \bb R$ be an $L$-Lipschitz function. Then we have for all $t > 0$ that 
\begin{align*}
\bb P\brac{ f(\mb X) - \bb Ef(\mb X) \geq t }\leq \exp\paren{-\frac{t^2}{2L^2} }. 
\end{align*}
\end{lemma}


\begin{lemma}[Bounding Maximum Norm of Gaussian Vector Sequence] \label{lem:gauss_seq_norm}
Let $\mb x_1, \dots, \mb x_{n_1}$ be a sequence of (not necessarily independent) standard Gaussian vectors in $\R^{n_2}$. It holds that 
\begin{align*}
\prob{\max_{i \in [n_1]} \norm{\mb x_i}_2 > \sqrt{n_2} + 2\sqrt{2 \log(2n_1)}} \le (2n_1)^{-3}.  
\end{align*}
\end{lemma}
\begin{proof}
Since the function $\norm{\cdot}_2$ is $1$-Lipschitz, by Gaussian concentration inequality, for any $i \in [n_1]$, we have 
\begin{align*}
\prob{\norm{\mb x_i}_2 - \sqrt{\bb E \norm{\mb x_i}_2^2} > t} \le  \prob{\norm{\mb x_i}_2 - \bb E \norm{\mb x_i}_2 > t} \le \exp\paren{-\frac{t^2}{2}}
\end{align*}
for all $t > 0$. Since $\bb E \norm{\mb x_i}_2^2 = n_2$, by a simple union bound, we obtain 
\begin{align*}
\prob{\max_{i \in [n_1]} \norm{\mb x_i} > \sqrt{n_2} + t} \le \exp\paren{-\frac{t^2}{2}+ \log n_1}
\end{align*}
for all $t > 0$. Taking $t = 2\sqrt{2\log (2n_1)}$ gives the claimed result. 
\end{proof}

\begin{corollary}\label{cor:gs_l2_linf}
	Let $\mb \Phi \in \R^{n_1 \times n_2} \sim_{i.i.d.} \mc N\paren{0, 1}$. It holds that
	\begin{align*}
		\norm{ \mb \Phi \mb x }_\infty \leq \paren{\sqrt{n_2} + 2\sqrt{2 \log(2n_1)}} \norm{\mb x}_2 \quad \text{for all } \mb x \in \bb R^{n_2}, 
	\end{align*}
	with probability at least $1- (2n_1)^{-3}$.
\end{corollary}

\begin{proof}
	Let $\mb \Phi = \brac{\bm \phi^1, \cdots, \bm \phi^{n_1}  }^\top $ . Without loss of generality, let us only consider $\mb x \in \bb S^{n_2-1}$, we have
	\begin{align}
		\norm{\mb \Phi \mb x }_{\infty} = \max_{i\in [n_1]} \abs{\mb x^\top \bm \phi^i } \leq \max_{i\in [n_1]} \norm{\bm \phi^i }_2.
	\end{align}
	Invoking Lemma \ref{lem:gauss_seq_norm} returns the claimed result.
\end{proof}


\begin{lemma} [Covering Number of a Unit Sphere \cite{veryshynin2011matrix}] \label{lem:eps-net-cover}
Let $\bb S^{n-1} =\Brac{\mb x\in \bb R^n\;|\; \norm{\mb x}_2 = 1 }$ be the unit sphere. For any $\eps \in \paren{0, 1}$, there exists some $\eps$ cover of $\bb S^{n-1}$ w.r.t. the $\ell^2$ norm, denoted as $\mc N_{\eps}$, such that 
\begin{align*}
\abs{\mc N_{\eps}} \le \paren{1+\frac{2}{\eps}}^n \le \paren{\frac{3}{\eps}}^n. 
\end{align*}
\end{lemma}

\begin{lemma}[Spectrum of Gaussian Matrices, \cite{veryshynin2011matrix}] \label{lem:tall-gaussian-spectral}
Let $\mb \Phi \in \bb R^{n_1\times n_2}$ ($n_1 > n_2$) contain i.i.d. standard normal entries. Then for every $t\geq 0$, with probability at least $1 - 2\exp\paren{-t^2/2}$, one has 
\begin{align*}
\sqrt{n_1} - \sqrt{n_2} - t \leq \sigma_{\min}(\mb \Phi) \leq \sigma_{\max}(\mb \Phi) \leq \sqrt{n_1}+\sqrt{n_2}+t. 
\end{align*}
\end{lemma}
 
\begin{lemma}
\label{lem:gs_l2_l1}
For any $\eps \in \paren{0, 1}$, there exists a constant $C\paren{\eps} > 1$, such that provided $n_1 > C\paren{\eps} n_2$, the random matrix  $\mb \Phi \in \R^{n_1 \times n_2} \sim_{i.i.d.} \mc N\paren{0, 1}$ obeys  
\begin{align*}
\paren{1-\eps} \sqrt{\frac{2}{\pi}} n_1 \norm{\mb x}_2 \le \norm{\mb \Phi \mb x}_1 \le \paren{1+\eps} \sqrt{\frac{2}{\pi}} n_1 \norm{\mb x}_2 \quad \text{for all} \;  \mb x \in \R^{n_2}, 
\end{align*}
with probability at least $1-2\exp\paren{-c\paren{\eps} n_1}$ for some $c\paren{\eps} > 0$. 
\end{lemma}
Geometrically, this lemma roughly corresponds to the well known almost spherical section theorem~\cite{figiel1977dimension, garnaev1984widths}, see also~\cite{gluskin2003note}. A slight variant of this version has been proved in~\cite{donoho2006most}, borrowing ideas from~\cite{pisier1999volume}. 
\begin{proof}
By homogeneity, it is enough to show that the bounds hold for every $\mb x$ of unit $\ell^2$ norm. For a fixed $\mb x_0$ with $\norm{\mb x_0}_2 = 1$, $\mb \Phi \mb x_0 \sim \mc N\paren{\mb 0, \mb I}$. So $\bb E \norm{\mb \Phi \mb x}_1 = \sqrt{\frac{2}{\pi}} n_1$. Note that $\norm{\cdot }_1$ is $\sqrt{n_1}$-Lipschitz, by concentration of measure for Gaussian vectors in Lemma \ref{lem:gaussian-concentration}, we have
\begin{align*}
\prob{\abs{\norm{\mb \Phi \mb x}_1 - \expect{\norm{\mb \Phi \mb x}_1}} > t} \le 2\exp\paren{-\frac{t^2}{2n_1}}
\end{align*}
for any $t > 0$. For a fixed $\delta \in \paren{0, 1}$, $\mc S^{n_2 -1}$ can be covered by a $\delta$-net $N_{\delta}$ with cardinality $\# N_{\delta} \le \paren{1+2/\delta}^{n_2}$. Now consider the event 
\begin{align*}
\event \doteq \Brac{\paren{1-\delta} \sqrt{\frac{2}{\pi}} n_1 \le \norm{\mb \Phi \mb x}_1 \le \paren{1+\delta}\sqrt{\frac{2}{\pi}} n_1\; \forall\; \mb x \in N_\delta }. 
\end{align*}
A simple application of union bound yields
\begin{align*}
\prob{\event^c} \le 2\exp\paren{-\frac{\delta^2 n_1}{\pi} + n_2 \log\paren{1+\frac{2}{\delta}}}. 
\end{align*}
Choosing $\delta$ small enough such that 
\begin{align*}
\paren{1-3\delta}\paren{1-\delta}^{-1} \ge 1-\eps \; \text{and} \; \paren{1+\delta}\paren{1-\delta}^{-1} \le 1+\eps, 
\end{align*}
then conditioned on $\event$, we can conclude that
\begin{align*}
\paren{1-\eps} \sqrt{\frac{2}{\pi}} n_1 \le \norm{\mb \Phi \mb x}_1 \le \paren{1+\eps}\sqrt{\frac{2}{\pi}} n_1\; \forall\; \mb x \in \bb S^{n_2 -1}. 
\end{align*} 
Indeed, suppose $\event$ holds. Then it can easily be seen that any $\mb z \in \bb S^{n_2 -1}$ can be written as 
\begin{align*}
\mb z = \sum_{k = 0}^\infty \lambda_k \mb x_k,  \qquad \text{with}\; \abs{\lambda_k} \le \delta^k, \mb x_k \in N_\delta \; \text{for all}\; k. 
\end{align*}
Hence we have 
\begin{align*}
\norm{\mb \Phi \mb z}_1 = \norm{\mb \Phi \sum_{k = 0}^\infty \lambda_k \mb x_k}_1 \le \sum_{k=0}^\infty \delta^k \norm{\mb \Phi \mb x_k}_1 \le \paren{1+\delta}\paren{1-\delta}^{-1}\sqrt{\frac{2}{\pi}} n_1. 
\end{align*}
Similarly, 
\begin{align*}
\norm{\mb \Phi \mb z}_1 = \norm{\mb \Phi \sum_{k = 0}^\infty \lambda_k \mb x_k}_1 \ge \brac{1-\delta - \delta\paren{1+\delta}\paren{1-\delta}^{-1}} \sqrt{\frac{2}{\pi}} n_1 = \paren{1-3\delta}\paren{1-\delta}^{-1} \sqrt{\frac{2}{\pi}} n_1. 
\end{align*}
Hence, the choice of $\delta$ above leads to the claimed result. Finally, given $n_1>C n_2$, to make the probability $\prob{\event^c}$ decaying in $n_1$, it is enough to set $C =  \frac{2\pi}{\delta^2} \log\paren{1+\frac{2}{\delta}}$. This completes the proof. 
\end{proof}

\section{The Random Basis vs. Its Orthonormalized Version} \label{sec:app_bases}


In this appendix, we consider the planted sparse model
\begin{align*}
	\overline{\mb Y} \;=\; \brac{ \mb x_0 \mid \mb g_1 \mid \cdots \mid \mb g_{n-1} } \;=\; \brac{\mb x_0 \mid \mb G } \in \bb R^{p\times n}
\end{align*}
as defined in \eqref{eqn:Y-bar}, where
\begin{align}\label{eqn:x_0-G}
x_0(k) \sim_{i.i.d.} \frac{1}{\sqrt{\theta p}}\mathrm{Ber}\paren{\theta}, \quad \mb g_\ell \sim_{i.i.d.} \mc N\paren{\mb 0, \frac{1}{p}\mb I},\quad 1\leq k \leq p,\;1\leq \ell \leq n-1. 
\end{align}
Recall that one ``natural/canonical'' orthonormal basis for the subspace spanned by columns of $\ol{\mb Y}$ is 
\begin{align*}
\mb Y = \brac{\frac{\mb x_0}{\norm{\mb x_0}_2} \mid \mc P_{\mb x_0^\perp} \mb G\paren{\mb G^\top \mc P_{\mb x_0^{\perp}} \mb G}^{-1/2} },
\end{align*}
which is well-defined with high probability as $\mc P_{\mb x_0^{\perp}} \mb G$ is well-conditioned (proved in Lemma \ref{lem:app_basis_op1}). We write 
\begin{align}\label{eqn:G'}
\mb G' \doteq \mc P_{\mb x_0^\perp} \mb G\paren{\mb G^\top \mc P_{\mb x_0^{\perp}} \mb G}^{-1/2}
\end{align}
for convenience. When $p$ is large, $\ol{\mb Y}$ has nearly orthonormal columns, and so we expect that $\mb Y$ closely approximates $\ol{\mb Y}$. In this section, we make this intuition rigorous. We prove several results that are needed for the proof of Theorem~\ref{thm:global}, and for translating results for $\ol{\mb Y}$ to results for $\mb Y$ in Appendix~\ref{app:q_prime_approx}.

\newcommand{\set}[1]{ \left\{ #1 \right\} }

For any realization of $\mb x_0$, let $\mc I = \mathrm{supp}(\mb x_0) = \set{ i \mid \mb x_0(i) \ne 0 }$. By Bernstein's inequality in Lemma \ref{lem:mc_bernstein_scalar} with $\sigma_X^2 = 2\theta$ and $R=1$, the event 
\begin{align} \label{eq:global_support_event}
\event_0 \doteq \Brac{\frac{1}{2} \theta p \le \abs{\mc I} \le 2\theta p}
\end{align} 
holds with probability at least $1 - 2\exp\paren{-\theta p/16}$. Moreover, we show the following: 
\begin{lemma}\label{lem:frac-x_0}
When $p \ge Cn$ and $\theta > 1/\sqrt{n}$, the bound
\begin{align}\label{eqn:event-0}
\abs{1-\frac{1}{\norm{\mb x_0}_2}}\;\leq\; \frac{4\sqrt{2}}{5} \sqrt{\frac{n\log p}{\theta^2 p}}
\end{align}
holds with probability at least $1-cp^{-2}$. Here $C, c$ are positive constants.  
\end{lemma}

\begin{proof}
Because $\expect{\norm{\mb x_0}_2^2} = 1$, by Bernstein's inequality in Lemma \ref{lem:mc_bernstein_scalar} with $\sigma_X^2 = 2/(\theta p^2) $ and $R=1/(\theta p)$, we have 
\begin{align*}
\prob{\abs{\norm{\mb x_0}_2^2 - \expect{\norm{\mb x_0}_2^2} } > t }= \prob{\abs{\norm{\mb x_0}_2^2 - 1} > t} \le 2\exp\paren{-  \frac{\theta p t^2}{4+ 2t } }
\end{align*}
for all $t > 0$, which implies 
\begin{align*}
\prob{ \abs{ \norm{\mb x_0}_{2} - 1 } > \frac{t}{ \norm{\mb x_0}_{2} + 1 } } = \prob{\abs{\norm{\mb x_0}_2 - 1} \paren{\norm{\mb x_0}_2 + 1} > t} \le 2\exp\paren{-  \frac{\theta p t^2}{4+ 2t } }. 
\end{align*}
On the intersection with $\event_0$, $\norm{\mb x_0}_2 + 1 \geq \frac{1}{ \sqrt{2}} + 1 \geq  5/4$ and setting $t = \sqrt{ \frac{n\log p}{\theta^2 p}  }$, we obtain
\begin{align*}
	\prob{  \abs{\norm{\mb x_0}_2 - 1} \geq \frac{4}{5} \sqrt{ \frac{n\log p}{\theta^2 p}  }  \;\Big|\; \event_0  } \leq 2 \exp\paren{ -\sqrt{n p\log p } }. 
\end{align*}
Unconditionally, this implies that with probability at least $1 - 2 \exp\paren{-p\theta/16 } - 2 \exp\paren{ -\sqrt{n p\log p } }$, we have
\begin{align*}
\abs{1-\frac{1}{\norm{\mb x_0}_2}} = \frac{\abs{1-\norm{\mb x_0}_2}}{\norm{\mb x_0}_2} \le \frac{4\sqrt{2}}{5} \sqrt{\frac{n\log p}{\theta^2 p}},
\end{align*}
as desired.
\end{proof}

Let $\mb M \doteq \paren{\mb G^\top \mc P_{\mb x_0^{\perp}} \mb G}^{-1/2}$. Then $\mb G' = \mb G \mb M - \frac{\mb x_0 \mb x_0^\top}{\norm{\mb x_0}_2^2} \mb G \mb M$. We show the following results hold:

\begin{lemma}\label{lem:app_basis_op1}
Provided $ p \ge Cn$, it holds that  
\begin{align*}
\norm{\mb M} \le 2, \quad \norm{\mb M - \mb I} \le 4\sqrt{\frac{n}{p}} + 4\sqrt{\frac{\log(2p)}{p}}
\end{align*}
with probability at least $1 - (2p)^{-2}$. Here $C$ is a positive constant.  
\end{lemma}
\begin{proof}
First observe that 
\begin{align*}
\norm{\mb M} = \paren{\sigma_{\min}\paren{\mb G^\top \mc P_{\mb x_0^\perp} \mb G}}^{-1/2} = \sigma_{\min}^{-1}\paren{\mc P_{\mb x_0^\perp} \mb G}. 
\end{align*}
Now suppose $\mb B$ is an orthonormal basis spanning $\mb x_0^\perp$. Then it is not hard to see the spectrum of $\mc P_{\mb x_0^\perp} \mb G$ is the same as that of $\mb B^\top \mb G \in \R^{\paren{p-1} \times \paren{n-1}}$; in particular, 
\begin{align*}
\sigma_{\min}\paren{\mc P_{\mb x_0^\perp} \mb G} = \sigma_{\min}\paren{\mb B^\top \mb G}. 
\end{align*}
Since each entry of $\mb G\sim_{i.i.d.} \mc N\paren{0, \frac{1}{p}}$, and $\mb B^\top$ has orthonormal rows, $\mb B^\top \mb G \sim_{i.i.d.} \mc N\paren{0, \frac{1}{p}}$, we can invoke the spectrum results for Gaussian matrices in Lemma~\ref{lem:tall-gaussian-spectral} and obtain that  
\begin{align*}
	 \sqrt{\frac{p-1}{p}}-  \sqrt{\frac{n-1}{p}} - 2\sqrt{\frac{\log\paren{2 p}}{p}}  \le \sigma_{\min}\paren{\mb B^\top \mb G} \le \sigma_{\max}\paren{\mb B^\top \mb G} \le \sqrt{\frac{p-1}{p}} +  \sqrt{\frac{n-1}{p}}+ 2\sqrt{\frac{\log(2 p)}{p}}
\end{align*}
with probability at least $1 - (2p)^{-2} $. Thus, when $p \ge C_1 n$ for some sufficiently large constant $C_1$, by using the results above we have 
\begin{align*}
\norm{\mb M} & = \sigma_{\min}^{-1}\paren{\mb B^\top \mb G}=\paren{\sqrt{\frac{p-1}{p}}-  \sqrt{\frac{n-1}{p}} - 2\sqrt{\frac{\log\paren{2 p}}{p}}}^{-1} \le 2, \\
\norm{\mb I - \mb M} & = \max\paren{\abs{\sigma_{\max}\paren{\mb M} - 1}, \abs{\sigma_{\min}\paren{\mb M} - 1}}  \\
    & =  \max\paren{\abs{\sigma_{\min}^{-1}\paren{\mb B^\top \mb G} - 1}, \abs{\sigma_{\max}^{-1}\paren{\mb B^\top \mb G} - 1}}  \\
    & \le \max\Brac{\paren{\sqrt{\frac{p-1}{p}}-  \sqrt{\frac{n-1}{p}} - 2\sqrt{\frac{\log\paren{2 p}}{p}}}^{-1} - 1, 1- \paren{ \sqrt{\frac{p-1}{p}} +  \sqrt{\frac{n-1}{p}}+ 2\sqrt{\frac{\log(2 p)}{p}}}^{-1}} \\
    & = \max \left\{\paren{1- \sqrt{\frac{p-1}{p}}+  \sqrt{\frac{n-1}{p}} + 2\sqrt{\frac{\log\paren{2 p}}{p}}} \paren{\sqrt{\frac{p-1}{p}}-  \sqrt{\frac{n-1}{p}} - 2\sqrt{\frac{\log\paren{2 p}}{p}}}^{-1}, \right. \\
    & \qquad \qquad \left.  \paren{\sqrt{\frac{p-1}{p}} -1 +  \sqrt{\frac{n-1}{p}}+ 2\sqrt{\frac{\log(2 p)}{p}}} \paren{ \sqrt{\frac{p-1}{p}} +  \sqrt{\frac{n-1}{p}}+ 2\sqrt{\frac{\log(2 p)}{p}}}^{-1} \right\} \\
    &\le 2 \paren{1- \sqrt{\frac{p-1}{p}}+  \sqrt{\frac{n-1}{p}} + 2\sqrt{\frac{\log\paren{2 p}}{p}}} \\
    &\le 4 \sqrt{\frac{n}{p}} + 4 \sqrt{\frac{\log (2p)}{p}}, 
\end{align*}
with probability at least $1 - (2p)^{-2}$. 
\end{proof} 

\begin{lemma}\label{lem:app_basis_op2}
Let $\mb Y_{\mc I}$ be a submatrix of $\mb Y$ whose rows are indexed by the set $\mc I$. There exists a constant $C > 0$, such that when $p \ge Cn$ and $1/2 > \theta > 1/\sqrt{n}$, the following 
\begin{align*}
\norm{\ol{\mb Y}}_{\ell^2 \to \ell^1} & \le 3\sqrt{p}, \\
\norm{\mb Y_{\mc I}}_{\ell^2 \to \ell^1} & \le 7\sqrt{2\theta p}, \\
\norm{\mb G - \mb G'}_{\ell^2 \to \ell^1} & \le 4\sqrt{n} + 7\sqrt{\log (2p)} , \\
\norm{\ol{\mb Y}_{\mc I} - \mb Y_{\mc I}}_{\ell^2 \to \ell^1} & \le 20\sqrt{\frac{ n\log p}{\theta} }, \\
	\norm{\ol{\mb Y} - \mb Y}_{\ell^2 \to \ell^1} & \le 20\sqrt{\frac{ n\log p}{\theta} }  
\end{align*}
hold simultaneously with probability at least $1 -cp^{-2} $ for a positive constant $c$. 
\end{lemma}
\begin{proof}
First of all, we have 
\begin{align*}
\norm{\frac{\mb x_0 \mb x_0^\top}{\norm{\mb x_0}_2^2} \mb G \mb M}_{\ell^2 \to \ell^1} \le \frac{1}{\norm{\mb x_0}_2^2} \norm{\mb x_0}_{\ell^2 \to \ell^1} \norm{\mb x_0^\top \mb G \mb M}_{\ell^2 \to \ell^2} = \frac{2}{\norm{\mb x_0}_2^2} \norm{\mb x_0}_1 \norm{\mb x_0^\top \mb G}_2, 
\end{align*}
where in the last inequality we have applied the fact $\norm{\mb M} \le 2$ from Lemma~\ref{lem:app_basis_op1}. Now $\mb x_0^\top \mb G$ is an i.i.d. Gaussian vectors with each entry distributed as $\mc N\paren{0, \frac{\norm{\mb x_0}_2^2}{p}}$, where $\norm{\mb x_0}_2^2 = \frac{\abs{\mc I}}{\theta p}$. So by Gaussian concentration inequality in Lemma \ref{lem:gaussian-concentration}, we have 
\begin{align*}
\norm{\mb x_0^\top \mb G}_2 \le  2\norm{\mb x_0}_2 \sqrt{\frac{\log (2p)}{p}} 
\end{align*}
with probability at least $1- c_1 p^{-2}$. On the intersection with $\event_0$, this implies  
\begin{align*}
\norm{\frac{\mb x_0 \mb x_0^\top}{\norm{\mb x_0}_2^2} \mb G \mb M}_{\ell^2 \to \ell^1} \le 2\sqrt{2 \theta \log (2p)}, 
\end{align*}
with probability at least $1- c_2p^{-2}$ provided $\theta > 1/\sqrt{n}$. Moreover, when intersected with $\event_0$, Lemma~\ref{lem:gs_l2_l1} implies that when $p \ge C_1 n$, 
\begin{align*}
\norm{\mb G}_{\ell^2 \to \ell^1} \le \sqrt{p}, \quad \norm{\mb G_{\mc I}}_{\ell^2 \to \ell^1} \le \sqrt{2\theta p}
\end{align*}
with probability at least $1- c_3 p^{-2}$ provided $\theta > 1/\sqrt{n}$. Hence, by Lemma \ref{lem:app_basis_op1}, when $p>C_2 n$, 
\begin{align*}
\norm{\mb G - \mb G'}_{\ell^2 \to \ell^1} 
& \le \norm{\mb G}_{\ell^2 \to \ell^1} \norm{\mb I - \mb M} + \norm{\frac{\mb x_0 \mb x_0^\top}{\norm{\mb x_0}_2^2} \mb G \mb M}_{\ell^2 \to \ell^1} \\
& \le \sqrt{p} \paren{4\sqrt{\frac{n}{p}} + 4\sqrt{\frac{\log(2p)}{p}}} + 2\sqrt{2 \theta \log(2p)} \le 4\sqrt{n} + 7\sqrt{\log(2p)} , \\
\norm{\ol{\mb Y} }_{\ell^2 \to \ell^1} & \le \norm{\mb x_0}_{\ell^2 \to \ell^1} + \norm{\mb G}_{\ell^2 \to \ell^1} \le \norm{\mb x_0}_1 + \sqrt{p} \le 2\sqrt{\theta p} + \sqrt{p} \le 3\sqrt{p}, \\
\norm{\mb G'_{\mc I}}_{\ell^2 \to \ell^1} 
& \le \norm{\mb G_{\mc I}}_{\ell^2 \to \ell^1} \norm{\mb M}  + \norm{\frac{\mb x_0 \mb x_0^\top}{\norm{\mb x_0}_2^2} \mb G \mb M}_{\ell^2 \to \ell^1} \le 2\sqrt{2\theta p} + 2\sqrt{2\theta \log (2p)} \le 4\sqrt{2 \theta p}, \\
\norm{\mb G_{\mc I} - \mb G'_{\mc I}}_{\ell^2 \to \ell^1} 
& \le \norm{\mb G_{\mc I}}_{\ell^2 \to \ell^1} \norm{\mb I - \mb M}+ \norm{\frac{\mb x_0 \mb x_0^\top}{\norm{\mb x_0}_2^2} \mb G \mb M}_{\ell^2 \to \ell^1} \\
& \le \sqrt{2 \theta p} \paren{4\sqrt{\frac{n}{p}} + 4\sqrt{\frac{\log (2p)}{p}}} +  2\sqrt{2 \theta \log(2p)}  \le 4\sqrt{2\theta n} + 6\sqrt{2 \theta \log(2p)}, \\
\norm{\mb Y_{\mc I}}_{\ell^2 \to \ell^1}
& \le \norm{\frac{\mb x_0}{\norm{\mb x_0}_2}}_{\ell^2 \to \ell^1} + \norm{\mb G'_{\mc I}}_{\ell^2 \to \ell^1} \le \frac{\norm{\mb x_0}_1}{\norm{\mb x_0}_2} + 6\sqrt{2\theta p} \le 7\sqrt{2\theta p}
\end{align*}
with probability at least $1-c_4 p^{-2}$ provided $\theta > 1/\sqrt{n}$. Finally, by Lemma \ref{lem:frac-x_0} and the results above, we obtain
\begin{align*}
\norm{\ol{\mb Y} - \mb Y}_{\ell^2 \to \ell^1} 
& \le \abs{1-\frac{1}{\norm{\mb x_0}_2}} \norm{\mb x_0}_1 + \norm{\mb G - \mb G'}_{\ell^2 \to \ell^1} \le 20\sqrt{\frac{ n\log p}{\theta} }, \\
\norm{\ol{\mb Y}_{\mc I} - \mb Y_{\mc I}}_{\ell^2 \to \ell^1} 
& \le \abs{1-\frac{1}{\norm{\mb x_0}_2}} \norm{\mb x_0}_1 +  \norm{\mb G_{\mc I} - \mb G'_{\mc I}}_{\ell^2 \to \ell^1} \le 20\sqrt{\frac{ n\log p}{\theta} }, 
\end{align*}
holding with probability at least $1 - c_5 p^{-2}$. 
\end{proof}

\begin{lemma} \label{lem:app_basis_op3}
Provided $p \ge Cn$ and $\theta > 1/\sqrt{n}$, the following 
\begin{align*}
\norm{\mb G'}_{\ell^2 \to \ell^\infty} & \le 2\sqrt{\frac{n}{p}} + 8\sqrt{\frac{2\log(2p)}{p}}, \\
\norm{\mb G - \mb G'}_{\ell^2 \to \ell^\infty} & \le \frac{4n}{p} + \frac{8\sqrt{2} \log(2p)}{p} + \frac{21\sqrt{n \log(2p)}}{p}
\end{align*}
hold simultaneously with probability at least $1-cp^{-2}$ for some constant $c > 0$. 
\end{lemma}

\begin{proof}
First of all, we have when $p \ge C_1 n$, it holds with probability at least $1 - c_2 p^{-2}$ that 
\begin{align*}
\norm{\frac{\mb x_0 \mb x_0^\top}{\norm{\mb x_0}_2^2} \mb G \mb M}_{\ell^2 \to \ell^\infty} \le \frac{1}{\norm{\mb x_0}_2^2} \norm{\mb x_0}_{\ell^2 \to \ell^\infty} \norm{\mb x_0^\top \mb G \mb M}_{\ell^2 \to \ell^2} \le \frac{2}{\norm{\mb x_0}_2^2} \norm{\mb x_0}_\infty \norm{\mb x_0^\top \mb G}_2, 
\end{align*}
where at the last inequality we have applied the fact $\norm{\mb M} \le 2$ from Lemma~\ref{lem:app_basis_op1}. Moreover, from proof of Lemma~\ref{lem:app_basis_op2}, we know that $\norm{\mb x_0^\top \mb G}_2 \le 2\sqrt{\log(2p)/p} \norm{\mb x_0}_2$ with probability at least $1-c_3 p^{-2}$ provided $p \ge C_4 n$. Therefore, conditioned on $\event_0$, we obtain that  
\begin{align*}
\norm{\frac{\mb x_0 \mb x_0^\top}{\norm{\mb x_0}_2^2} \mb G \mb M}_{\ell^2 \to \ell^\infty} \le \frac{4\norm{\mb x_0}_{\infty}}{\norm{\mb x_0}_2} \sqrt{\frac{\log(2p)}{p}} \le \frac{4\sqrt{2\log(2p)}}{\sqrt{\theta}p}
\end{align*}
holds with probability at least $1- c_5 p^{-2}$ provided $\theta > 1/\sqrt{n}$. Now by Corollary~\ref{cor:gs_l2_linf}, we have that 
\begin{align*}
\norm{\mb G}_{\ell^2 \to \ell^\infty} \le \sqrt{\frac{n}{p}} + 2\sqrt{\frac{2\log(2p)}{p}}
\end{align*}
with probability at least $1- c_6 p^{-2}$. Combining the above estimates and Lemma~\ref{lem:app_basis_op1}, we have that with probability at least $1 -c_7 p^{-2}$
\begin{align*}
\norm{\mb G'}_{\ell^2 \to \ell^\infty} 
& \le \norm{\mb G \mb M}_{\ell^2 \to \ell^\infty} + \norm{\frac{\mb x_0 \mb x_0^\top}{\norm{\mb x_0}_2^2} \mb G \mb M}_{\ell^2 \to \ell^\infty}  \\
& \le \norm{\mb G}_{\ell^2 \to \ell^\infty} \norm{\mb M} + \norm{\frac{\mb x_0 \mb x_0^\top}{\norm{\mb x_0}_2^2} \mb G \mb M}_{\ell^2 \to \ell^\infty} \\
& \le 2\sqrt{\frac{n}{p}} + 4\sqrt{\frac{2\log(2p)}{p}} + \frac{4\sqrt{2\log(2p)}}{\sqrt{\theta}p} \le 2\sqrt{\frac{n}{p}} + 8\sqrt{\frac{2\log(2p)}{p}}, 
\end{align*}
where the last simplification is provided that $\theta > 1/\sqrt{n}$ and $p \ge C_8 n$ for a sufficiently large $C_8$. Similarly, 
\begin{align*}
\norm{\mb G - \mb G'}_{\ell^2 \to \ell^{\infty}} 
& \le \norm{\mb G}_{\ell^2 \to \ell^\infty} \norm{\mb I - \mb M} + \norm{\frac{\mb x_0 \mb x_0^\top}{\norm{\mb x_0}_2^2} \mb G \mb M}_{\ell^2 \to \ell^\infty} \\
& \le \frac{4n}{p} + \frac{8\sqrt{2} \log(2p)}{p} + \frac{(8\sqrt{2} + 4) \sqrt{n \log(2p)}}{p} + \frac{4\sqrt{2\log(2p)}}{\sqrt{\theta}p} \\
& \le \frac{4n}{p} + \frac{8\sqrt{2} \log(2p)}{p} + \frac{21\sqrt{n \log(2p)}}{p}, 
\end{align*}
completing the proof. 
\end{proof}

\section{Proof of $\ell^1/ \ell^2$ Global Optimality}

In this appendix, we prove the $\ell^1/\ell^2$ global optimality condition in Theorem \ref{thm:global} of Section \ref{sec:optimality}. 
\begin{proof}[Proof of Theorem \ref{thm:global}]
We will first analyze a canonical version, in which the input orthonormal basis is $\mb Y$ as defined in \eqref{eqn:Y-Z-R} of Section \ref{sec:algorithm}: 
\begin{align*}
\min_{\mb q \in \R^n} \norm{\mb Y \mb q}_1, \qquad \text{s.t.} \; \norm{\mb q}_2 = 1. 
\end{align*}
Let $\mb q = \begin{bmatrix}
	q_1 \\ \mb q_2
\end{bmatrix}  $ and let $\mc I$ be the support set of $\mb x_0$, we have
\begin{align*}
\norm{{\bf Yq}}_1 
& = \norm{\mb Y_{\mc I} \mb q}_1 + \norm{\mb  Y_{\mc I^c}\mb q}_1 \\
& \ge \abs{q_1} \norm{\frac{\mb x_0}{\norm{\mb x_0}_2}}_1 - \norm{\mb G'_{\mc I}\mb q_2}_1 + \norm{\mb G'_{\mc I^c}\mb q_2}_1   \\
& \ge \abs{q_1} \norm{\frac{\mb x_0}{\norm{\mb x_0}_2}}_1 - \norm{\mb G_{\mc I}\mb q_2}_1 - \norm{\paren{\mb G_{\mc I} - \mb G'_{\mc I}}\mb q_2}_1 + \norm{\mb G_{\mc I^c}\mb q_2}_1 - \norm{\paren{\mb G_{\mc I^c} - \mb G'_{\mc I^c}} \mb q_2}_1  \\
& \ge \abs{q_1} \norm{\frac{\mb x_0}{\norm{\mb x_0}_2}}_1 - \norm{\mb G_{\mc I}\mb q_2}_1 + \norm{\mb G_{\mc I^c} \mb q_2}_1 - \norm{\mb G - \mb G'}_{\ell^2 \to \ell^1} \norm{\mb q_2}_2,
\end{align*}
where $\mb G$ and $\mb G'$ are defined in \eqref{eqn:x_0-G} and \eqref{eqn:G'} of Appendix \ref{sec:app_bases}. By Lemma~\ref{lem:gs_l2_l1} and intersecting with $\event_0$ defined in \eqref{eq:global_support_event}, we have that as long as $p \ge C_1n$,  
\begin{align*}
\norm{\mb G_{\mc I} \mb q_2}_1 & \le \frac{2\theta p}{\sqrt{p}} \norm{\mb q_2}_2 = 2\theta \sqrt{p} \norm{\mb q_2}_2 \; \text{for all} \; \mb q_2 \in \R^{n-1}, \\
\norm{\mb G_{\mc I^c} \mb q_2}_1 & \ge \frac{1}{2}\frac{p - 2\theta p}{\sqrt{p}} \norm{\mb q_2}_2 = \frac{1}{2} \sqrt{p}\paren{1-2\theta} \norm{\mb q_2}_2 \; \text{for all} \; \mb q_2 \in \R^{n-1}, 
\end{align*}
hold with probability at least $1-c_2 p^{-2}$. Moreover, by Lemma~\ref{lem:app_basis_op2},
\begin{align*}
\norm{\mb G - \mb G'}_{\ell^2 \to \ell^1} \le 4\sqrt{n } + 7\sqrt{\log(2p)}
\end{align*}
holds with probability at least $1-c_3p^{-2} $ when $p \ge C_4 n$ and $\theta > 1/\sqrt{n}$. So we obtain that  
\begin{align*}
     \norm{{\mb Y\mb q}}_1 \geq g(\mb q) \doteq \abs{q_1} \norm{\frac{\mb x_0}{\norm{\mb x_0}_2}}_1 + \norm{\mb q_2}_2\paren{\frac{1}{2}\sqrt{p}\paren{1-2\theta} - 2\theta \sqrt{p} - 4\sqrt{n} - 7\sqrt{\log (2p)}} 
\end{align*}
holds with probability at least $1- c_5p^{-2} $. Assuming $\event_0$, we observe
\begin{align*}
\norm{\frac{\mb x_0}{\norm{\mb x_0}_2}}_1 \le \sqrt{\abs{\mc I}} \norm{\frac{\mb x_0}{\norm{\mb x_0}_2}}_2 \le \sqrt{2\theta p}.
\end{align*}
Now $g(\mb q)$ is a linear function in $\abs{q_1}$ and $\norm{\mb q_2}_2$. Thus,  whenever $\theta$ is sufficiently small and $p \ge C_6 n$ such that
\begin{align*}
\sqrt{2\theta p} < \frac{1}{2}\sqrt{p}\paren{1-2\theta} - 2\theta \sqrt{p} - 4\sqrt{n} - 7\sqrt{\log(2p)}, 
\end{align*}
$\pm \mb e_1$ are the unique minimizers of $g(\mb q)$ under the constraint $q_1^2 + \norm{\mb q_2}_2^2 = 1$. In this case, because $\norm{\mb Y(\pm \mb e_1) }_1 = g(\pm \mb e_1)$, and we have
\begin{align*}
	\norm{{\bf Yq}}_1 \geq g(\mb q) >g(\pm \mb e_1)
\end{align*}
for all $\mb q \not = \pm \mb e_1$, $\pm \mb e_1$ are the unique minimizers of $\norm{\mb Yq}_1$ under the spherical constraint. Thus there exists a universal constant $\theta_0 > 0$, such that for all $1/\sqrt{n} \le \theta \le \theta_0$, $\pm \mb e_1$ are the only global minimizers of~\eqref{eqn:syn_l1_l2} if the input basis is $\mb Y$.

Any other input basis can be written as $\widehat{\mb Y} = \mb Y \mb U$, for some orthogonal matrix $\mb U$. The program now is written as 
\begin{align*}
\min_{\mb q \in \R^n} \norm{\widehat{\mb Y} \mb q}_1, \qquad \text{s.t.} \; \norm{\mb q}_2 = 1, 
\end{align*}
which is equivalent to 
\begin{align*}
\min_{\mb q \in \R^n} \norm{\widehat{\mb Y} \mb q}_1, \qquad \text{s.t.} \; \norm{\mb U \mb q}_2 = 1, 
\end{align*}
which is obviously equivalent to the canonical program we analyzed above by a simple change of variable, i.e., $\overline{\mb q} \doteq \mb U \mb q$, completing the proof. 
\end{proof}

\section{Good Initialization}\label{app:initialization}

In this appendix, we prove Proposition \ref{prop:initialization}. We show that the initializations produced by the procedure described in Section \ref{sec:algorithm} are biased towards the optimal.

\begin{proof}[Proof of Proposition \ref{prop:initialization}]
Our previous calculation has shown that $\theta p/2 \le \abs{\mc I} \le 2 \theta p$ with probability at least $1 - c_1 p^{-2}$ provided $p \ge C_2n$ and $\theta > 1/\sqrt{n}$. Let $\mb Y= \brac{\mb y^1,\cdots,\mb y^p }^\top$ as defined in \eqref{eqn:Y-Z-R}. Consider any $i \in \mc I$. Then $x_{0}(i) = \frac{1}{\sqrt{\theta p}}$, and 
\begin{align*}
\innerprod{\mb e_1}{\mb y^i / \norm{\mb y^i}_2} = \frac{1/\sqrt{\theta p}}{\norm{\mb x_0}_2\norm{\mb y^{i}}_2} & \ge \frac{1/\sqrt{\theta p}}{\norm{\mb x_0}_2 \paren{\norm{\mb x_0}_{\infty}/\norm{\mb x_0}_2 + \norm{(\mb g')^i}_2}} \\
& \ge \frac{1/\sqrt{\theta p}}{\norm{\mb x_0}_2 \paren{\norm{\mb x_0}_{\infty}/\norm{\mb x_0}_2 + \norm{\mb g^i}_2 + \norm{\mb G - \mb G'}_{\ell^2 \to \ell^\infty}}}, 
\end{align*}
where $\mb g^i$ and $(\mb g')^i$ are the $i$-th rows of $\mb G$ and $\mb G'$, respectively. Since such $\mb g^i$'s are independent Gaussian vectors in $\R^{n-1}$ distributed as $\mc N(\mb 0, 1/p)$, by Gaussian concentration inequality and the fact that $\abs{\mc I} \ge p\theta/2$ w.h.p., 
\begin{align*}
\prob{\exists i \in \mc I: \norm{\mb g^i}_2 \le 2\sqrt{n/p}} \ge 1 - \exp\paren{-c_3n\theta p} \le c_4 p^{-2}, 
\end{align*}
provided $p \ge C_5 n$ and $\theta > 1/\sqrt{n}$. Moreover,
\begin{align*}
\norm{\mb x_0}_2 = \sqrt{\abs{\mc I} \times \frac{1}{\theta p}} \le  \sqrt{2 \theta p \times \frac{1}{\theta p}} = \sqrt{2}. 
\end{align*}
Combining the above estimates and result of Lemma \ref{lem:app_basis_op3}, we obtain that provided $p \ge C_6 n$ and $\theta > 1/\sqrt{n}$, with probability at least $1 - c_7 p^{-2}$, there exists an $i \in [p]$, such that if we set $\mb q^{(0)} = \mb y^i/\norm{\mb y^i}_2$, it holds that 
\begin{align*}
\abs{q_1^{(0)}} 
& \ge \frac{1/\sqrt{\theta p}}{1/\sqrt{\theta p} + 2\sqrt{2}\sqrt{n/p} + \sqrt{2}\paren{4n/p + 8\sqrt{2} \log(2p) /p + 21\sqrt{n \log (2p)}/p}}\\
& \ge \frac{1/\sqrt{\theta p}}{1/\sqrt{\theta p} + 6\sqrt{2}\sqrt{n/p}} \quad (\text{using $p \ge C_6 n$ to simplifiy the above line})\\
& =  \frac{1}{1+ 6\sqrt{2}\sqrt{\theta n}} \\
& \ge \frac{1}{(1+6\sqrt{2}) \sqrt{\theta  n}} \quad (\text{as $\theta > 1/\sqrt{n}$})\\
& \ge \frac{1}{10\sqrt{\theta n}}, 
\end{align*}
completing the proof. 
\end{proof}
We will next show that for an arbitrary orthonormal basis $\widehat{\mb Y} \doteq \mb Y \mb U$ the initialization still biases towards the target solution. To see this, suppose w.l.o.g. $\paren{\mb y^i}^\top$ is a row of $\mb Y$ with nonzero first coordinate. We have shown above that with high probability $\abs{\innerprod{\frac{\mb y^i}{\norm{\mb y^i}_2}}{\mb e_1}} \ge \frac{1}{10\sqrt{\theta n}}$ if $\mb Y$ is the input orthonormal basis. For $\mb Y$, as $\mb x_0 = \mb Y \mb e_1 = \mb Y\mb U \mb U^\top \mb e_1$, we know $\mb q_\star = \mb U^\top \mb e_1$ is the target solution corresponding to $\widehat{\mb Y}$. Observing that 
\begin{align*}
\abs{\innerprod{\mb U^\top \mb e_1}{\frac{\paren{\mb e_i^\top \widehat{\mb Y}}^\top}{\norm{\paren{\mb e_i^\top \widehat{\mb Y}}^\top}_2}}} 
= \abs{\innerprod{\mb U^\top \mb e_1}{\frac{\mb U^\top \mb Y^\top \mb e_i}{\norm{\mb U^\top \mb Y^\top \mb e_i}_2}}} 
= \abs{\innerprod{\mb e_1}{\frac{\paren{\mb Y}^\top \mb e_i}{\norm{\mb Y^\top \mb e_i}_2}}} = \abs{\innerprod{\mb e_1}{\frac{\mb y^i}{\norm{\mb y^i}_2}}} \\\ge \frac{1}{10\sqrt{n \theta}}, 
\end{align*} 
corroborating our claim.

\section{Lower Bounding Finite Sample Gap $G(\mb q)$}\label{app:gap-finite}
In this appendix, we prove Proposition \ref{prop:gap-bound-Y'}. In particular, we show that the gap $\mb G(\mb q)$ defined in \eqref{eqn:gap-G} is strictly positive over a large portion of the sphere $\bb S^{n-1}$.
\begin{proof}[Proof of Proposition \ref{prop:gap-bound-Y'}]
Without loss of generality, we work with the ``canonical'' orthonormal basis $\mb Y$ defined in \eqref{eqn:Y-Z-R}. Recall that $\mb Y$ is the orthogonalization of the planted sparse basis $\ol{\mb Y}$ as defined in \eqref{eqn:Y-bar}. We define the processes $\ol{\mb Q}(\mb q)$ and $\mb Q(\mb q)$ on $\mb q\in \bb S^{n-1}$, via
\begin{align*}
	\ol{\mb Q}(\mb q) = \frac{1}{p} \sum_{i=1}^p \ol{\mb y}^i S_\lambda\brac{\mb q^\top \ol{\mb y}^i },\quad \mb Q(\mb q) = \frac{1}{p} \sum_{i=1}^p \mb y^i S_\lambda\brac{\mb q^\top \mb y^i }.
\end{align*}
Thus, we can separate $\ol{\mb Q}(\mb q)$  as $\ol{\mb Q}(\mb q) = \left[ \begin{array}{c} \ol{Q}_1(\mb q) \\ \ol{\mb Q}_2(\mb q) \end{array} \right]$, where
\begin{align}\label{eqn:Q-1-2-bar}
	\ol{Q}_1(\mb q) = \frac{1}{p} \sum_{i=1}^p x_{0i} S_\lambda\brac{\mb q^\top \ol{\mb y}^i}\quad \text{and}\quad \ol{\mb Q}_2(\mb q) = \frac{1}{p} \sum_{i=1}^p \mb g_i S_\lambda \brac{\mb q^\top \ol{\mb y}^i },
\end{align}
and separate $\mb Q(\mb q)$ correspondingly. Our task is to lower bound the gap $G(\mb q)$ for finite samples as defined in \eqref{eqn:gap-G}. Since we can deterministically constrain $\abs{q_1}$ and $\norm{\mb q_2}_2$ over the set $\Gamma$ as defined in \eqref{eqn:Gamma-set} (e.g., $\frac{1}{10\sqrt{n \theta}}\le \abs{q_1} \le 3\sqrt{\theta}$ and $\norm{\mb q_2}_2 \ge \frac{1}{10}$, where the choice of $\frac{1}{10}$ for $\mb q_2$ is arbitrary here, as we can always take a sufficiently small $\theta$), the challenge lies in lower bounding $\abs{Q_1\paren{\mb q}}$ and upper bounding $\norm{\mb Q_2\paren{\mb q}}_2$, which depend on the orthonormal basis $\mb Y$. The unnormalized basis $\ol{\mb Y}$ is much easier to work with than $\mb Y$. Our proof will follow the observation that
\begin{align*}
\abs{Q_1\paren{\mb q}} & \ge \abs{\expect{\ol{Q}_1\paren{\mb q}}} - \abs{\ol{Q}_1\paren{\mb q} - \expect{\ol{Q}_1\paren{\mb q}}} - \abs{Q_1\paren{\mb q} - \ol{Q}_1\paren{\mb q}}, \\
\norm{\mb Q_2\paren{\mb q}} & \le \norm{\expect{\ol{\mb Q}_2\paren{\mb q}}}_2 + \norm{\ol{\mb Q}_2\paren{\mb q} - \expect{\ol{\mb Q}_2\paren{\mb q}}}_2 + \norm{\mb Q_2\paren{\mb q} - \ol{\mb Q}_2\paren{\mb q}}_2. 
\end{align*}
In particular, we show the following:
\begin{itemize}
\item Appendix~\ref{app:gap} shows that the expected gap is lower bounded for all $\mb q\in \bb S^{n-1}$ with $\abs{q_1} \le 3\sqrt{\theta}$: 
\begin{align*}
\ol{G}\paren{\mb q} \doteq \frac{\abs{\expect{\ol{Q}_1\paren{\mb q}}}}{\abs{q_1}} - \frac{\norm{\expect{\ol{\mb Q}_2\paren{\mb q}}}_2}{\norm{\mb q_2}_2} \ge \frac{1}{50}\frac{q_1^2}{\theta p}. 
\end{align*}
As $\abs{q_1}\geq \frac{1}{10\sqrt{n\theta}}$, we have 
\begin{align*}
\inf_{\mb q \in \Gamma} \frac{\abs{\expect{\ol{Q}_1\paren{\mb q}}}}{\abs{q_1}} - \frac{\norm{\expect{\ol{\mb Q}_2\paren{\mb q}}}_2}{\norm{\mb q_2}_2} \ge \frac{1}{5000}\frac{1}{\theta^2 n p}.
\end{align*}
\item Appendix~\ref{app:concentration}, as summarized in Proposition~\ref{lem:uniform_Q1_Q2}, shows that whenever $p \ge \Omega\paren{n^4\log n}$, it holds with high probability that 
\begin{align*}
& \sup_{\mb q \in \Gamma} \frac{\abs{\ol{Q}_1\paren{\mb q} - \expect{\ol{Q}_1\paren{\mb q}}}}{\abs{q_1}} + \frac{\norm{\ol{\mb Q}_2\paren{\mb q} - \expect{\ol{\mb Q}_2\paren{\mb q}}}_2}{\norm{\mb q_2}_2}  \\
\le\; & \frac{10\sqrt{\theta n}}{4\times 10^5 \theta^{5/2}n^{3/2}p} + \frac{10}{4\times 10^5\theta^2 np} = \frac{1}{2\times 10^4\theta^2 np}. 
\end{align*}
\item Appendix~\ref{app:q_prime_approx} shows that whenever $p \geq \Omega\paren{n^4 \log n}$, it holds with high probability that 
\begin{align*}
& \sup_{\mb q \in \Gamma} \frac{\abs{\ol{Q}_1\paren{\mb q} - Q_1\paren{\mb q}}}{\abs{q_1}} + \frac{\norm{\ol{\mb Q}_2\paren{\mb q} - \mb Q_2\paren{\mb q}}_2}{\norm{\mb q_2}_2}  \\
\le\; & \frac{10\sqrt{\theta n}}{4\times 10^5 \theta^{5/2}n^{3/2}p} + \frac{10}{4\times 10^5\theta^2 np} = \frac{1}{2\times 10^4\theta^2 np}.
\end{align*}
\end{itemize}
Observing that
\begin{align*}
\inf_{\mb q \in \Gamma} G(\mb q) 
& \ge \inf_{\mb q \in \Gamma} \paren{\frac{\abs{\expect{\ol{Q}_1\paren{\mb q}}}}{\abs{q_1}} - \frac{\norm{\expect{\ol{\mb Q}_2\paren{\mb q}}}_2}{\norm{\bm q_2}_2}} - \sup_{\mb q \in \Gamma} \paren{\frac{\abs{\ol{Q}_1\paren{\mb q} - \expect{\ol{Q}_1\paren{\mb q}}}}{\abs{q_1}} + \frac{\norm{\ol{\mb Q}_2\paren{\mb q} - \expect{\ol{\mb Q}_2\paren{\mb q}}}_2}{\norm{\mb q_2}_2} }  \\
& \qquad -  \sup_{\mb q \in \Gamma} \paren{\frac{\abs{\ol{Q}_1\paren{\mb q} - Q_1\paren{\mb q}}}{\abs{q_1}} + \frac{\norm{\ol{\mb Q}_2\paren{\mb q} - \mb Q_2\paren{\mb q}}_2}{\norm{\mb q_2}_2}}, 
\end{align*}
we obtain the result as desired.
\end{proof}

For the general case when the input orthonormal basis is $\widehat{\mb Y} = \mb Y \mb U$ with target solution $\mb q_\star = \mb U^\top \mb e_1$, a straightforward extension of the definition for the gap would be: 
\begin{align*}
G\paren{\mb q; \widehat{\mb Y} = \mb Y \mb U} \doteq \frac{\abs{\innerprod{\mb Q\paren{\mb q; \widehat{\mb Y}}}{\mb U^\top \mb e_1}}}{\abs{\innerprod{\mb q}{\mb U^\top \mb e_1}}} - \frac{\norm{\paren{\mb I - \mb U^\top \mb e_1 \mb e_1^\top \mb U}\mb Q\paren{\mb q; \widehat{\mb Y}}}_2}{\norm{\paren{\mb I - \mb U^\top \mb e_1 \mb e_1^\top \mb U} \mb q}_2}. 
\end{align*}
Since $\mb Q\paren{\mb q; \widehat{\mb Y}} = \frac{1}{p} \sum_{k=1}^p \mb U^\top \mb y^k S_{\lambda}\paren{\mb q^\top \mb U^\top \mb y^k}$, we have 
\begin{align} \label{eq:general_basis_identity}
\mb U \mb Q\paren{\mb q; \widehat{\mb Y}} = \frac{1}{p} \sum_{k=1}^p \mb U \mb U^\top \mb y^k S_{\lambda}\paren{\mb q^\top \mb U^\top \mb y^k} = \frac{1}{p} \sum_{k=1}^p \mb y^k S_{\lambda}\brac{\paren{\mb U \mb q}^\top \mb y^k} = \mb Q\paren{\mb U \mb q; \mb Y}. 
\end{align}
Hence we have 
\begin{align*}
G\paren{\mb q; \widehat{\mb Y} = \mb Y \mb U} = \frac{\abs{\innerprod{\mb Q\paren{\mb U \mb q; \mb Y}}{\mb e_1}}}{\abs{\innerprod{\mb U \mb q}{\mb e_1}}} - \frac{\norm{\paren{\mb I - \mb e_1 \mb e_1^\top} \mb Q\paren{\mb U \mb q; \mb Y}}_2}{\norm{\paren{\mb I - \mb e_1 \mb e_1^\top} \mb U \mb q}_2}. 
\end{align*}
Therefore, from Proposition~\ref{prop:gap-bound-Y'} above, we conclude that under the same technical conditions as therein, 
\begin{align*}
\inf_{\mb q \in \bb S^{n-1}: \frac{1}{10\sqrt{\theta n}} \le \abs{\innerprod{\mb U \mb q}{\mb e_1}} \le 3\sqrt{\theta}}  G\paren{\mb q; \widehat{\mb Y}} \ge \frac{1}{10^4\theta^2 np}
\end{align*}
with high probability. 

\subsection{Lower Bounding the Expected Gap $\ol{G}(\mb q)$}\label{app:gap}

In this section, we provide a nontrivial lower bound for the gap
\begin{align}
\ol{G}(\mb q) = \frac{\abs{\bb E\brac{\ol{Q}_1(\mb q)}}}{\abs{q_1}} - \frac{\norm{\bb E\brac{\ol{\mb Q}_2(\mb q)}}_2}{\norm{\mb q_2}_2}.\label{eqn:gap-original}
\end{align}
More specifically, we show that:
\begin{proposition}\label{lem:gap-lower}
There exists some numerical constant $\theta_0 >0$, such that for all $\theta \in\paren{0, \theta_0}$, it holds that 
\begin{align}
\ol{G}(\mb q) \geq \frac{1}{50} \frac{q_1^2}{\theta p}\label{eqn:gap-G-bound}
\end{align}
for all $\mb q\in \bb S^{n-1}$ with $\abs{q_1}\leq 3\sqrt{\theta}$.
\end{proposition}

Estimating the gap $\ol{G}(\mb q)$ requires delicate estimates for $\bb E\brac{\ol{Q}_1(\mb q)}$ and $\bb E\brac{\ol{\mb Q}_2(\mb q)}$. We first outline the main proof in Appendix \ref{sec:sketch-gap}, and delay these detailed technical calculations to the subsequent subsections.

\subsubsection{Sketch of the Proof} \label{sec:sketch-gap}
W.l.o.g., we only consider the situation that $q_1>0$, because the case of $q_1<0$ can be similarly shown by symmetry. By \eqref{eqn:Q-1-2-bar}, we have
\begin{align*}
\bb E\brac{ \ol{Q}_1(\mb q) } \;&=\; \bb E\brac{x_0 S_\lambda \brac{x_0q_1+\mb q_2^\top \mb g}}, \\
\bb E\brac{\ol{\mb Q}_2(\mb q)} \;&=\; \bb E\brac{\mb g  S_\lambda \brac{x_0q_1+\mb q_2^\top \mb g}},
\end{align*}
where $\mb g\sim \mc N\paren{\mb 0, \frac{1}{p}\mb I}$, and $x_0\sim \frac{1}{\sqrt{\theta p}}\text{Ber}(\theta)$. Let us decompose
\begin{align*}
\mb g = \mb g_\parallel + \mb g_\perp,
\end{align*}
with $\mb g_\parallel = \mc P_\parallel \mb g=  \frac{\mb q_2\mb q_2^\top }{\norm{\mb q_2}_2^2}\mb g$, and $\mb g_\perp = (\mb I- \mc P_\parallel) \mb g$. In this notation, we have
\begin{align*}
\bb E\brac{\ol{\mb Q}_2(\mb q)} \;&=\; \bb E\brac{\mb g_\parallel S_\lambda \brac{x_0q_1+\mb q_2^\top\mb g_\parallel } }+  \bb E\brac{ \mb g_\perp S_\lambda \brac{x_0q_1+\mb q_2^\top\mb g_\parallel } }  \\
\;&=\; \bb E\brac{\mb g_\parallel  S_\lambda \brac{x_0q_1+\mb q_2^\top\mb g }}+  \bb E\brac{\mb g_\perp}\bb E\brac{  S_\lambda \brac{x_0q_1+\mb q_2^\top\mb g }} \\
\;&=\;  \frac{\mb q_2}{\norm{\mb q_2}_2^2} \bb E\brac{ \mb q_2^\top \mb g  S_\lambda \brac{x_0q_1+\mb q_2^\top\mb g }},
\end{align*}
where we used the facts that $\mb q_2^\top \mb g = \mb q_2^\top \mb g_\parallel$, $\mb g_\perp$ and $\mb g_\parallel$ are uncorrelated Gaussian vectors and therefore independent, and $\bb E\brac{\mb g_\perp} =\mb 0$. Let $Z \doteq \mb g^\top \mb q_2 \sim \mc N(0,\sigma^2)$ with $\sigma^2 = \norm{\mb q_2}_2^2/p$, by partial evaluation of the expectations with respect to $x_0$, we get
\begin{align}
\bb E\brac{\ol{Q}_1(\mb q)} \;&=\; \sqrt{\frac{\theta}{p}} \bb E\brac{S_\lambda\brac{\frac{q_1}{\sqrt{\theta p}} +Z}  },\label{eqn:expect-Q-1-pre} \\
\bb E\brac{\ol{\mb Q}_2(\mb q) } \;&=\; \frac{\theta\mb q_2}{\norm{\mb q_2}_2^2} \bb E\brac{Z S_\lambda\brac{\frac{q_1}{\sqrt{\theta p}}+Z}}+ \frac{(1-\theta)\mb q_2}{\norm{\mb q_2}_2^2}\bb E\brac{Z S_\lambda\brac{Z} }.\label{eqn:expect-Q-2-pre}
\end{align}
Straightforward integration based on Lemma~\ref{lem:gaussian-integral} gives a explicit form of the expectations as follows
\begin{align}
\E\brac{\ol{Q}_1(\mb q)}  \;&=\; \sqrt{\frac{\theta}{p}} \Brac{\brac{\alpha\Psi\paren{-\frac{\alpha}{\sigma} }+\beta\Psi\paren{\frac{\beta}{\sigma} }}+\sigma\brac{\psi\paren{-\frac{\beta}{\sigma}}-\psi\paren{-\frac{\alpha}{\sigma}}}},\label{eqn:expect-Q-1}\\
\E\brac{\ol{\mb Q}_2(\mb q)} \;&=\; \Brac{\frac{2\paren{1-\theta}}{p}\Psi\paren{-\frac{\lambda}{\sigma}}+\frac{\theta}{p}\brac{\Psi\paren{-\frac{\alpha}{\sigma}}+\Psi \paren{\frac{\beta}{\sigma}}}}{\bf q}_2,\label{eqn:expect-Q-2}
\end{align}
where the scalars $\alpha$ and $\beta$ are defined as 
\begin{align*}
\alpha = \frac{q_1}{\sqrt{\theta p}}+\lambda,\quad\quad \beta = \frac{q_1}{\sqrt{\theta p}}-\lambda,
\end{align*}
and $\psi\paren{t}$ and $\Psi\paren{t}$ are \emph{pdf} and \emph{cdf} for standard normal distribution, respectively, as defined in Lemma~\ref{lem:gaussian-integral}. Plugging \eqref{eqn:expect-Q-1} and \eqref{eqn:expect-Q-2} into \eqref{eqn:gap-original}, by some simplifications, we obtain
\begin{align}
\ol{G}(\mb q) \;=\;& \frac{1}{q_1}\sqrt{\frac{\theta}{p}} \brac{\alpha\Psi\paren{-\frac{\alpha}{\sigma}}+ \beta \Psi\paren{\frac{\beta}{\sigma}}- \frac{2q_1}{\sqrt{\theta p}} \Psi\paren{-\frac{\lambda}{\sigma}} } - \frac{\theta}{p}\brac{\Psi\paren{-\frac{\alpha}{\sigma}} + \Psi\paren{\frac{\beta}{\sigma}}-2\Psi\paren{-\frac{\lambda}{\sigma}} } \nonumber \\
&+ \frac{\sigma}{q_1}\sqrt{\frac{\theta}{p}}\brac{\psi\paren{\frac{\beta}{\sigma}} - \psi\paren{-\frac{\alpha}{\sigma}} }.\label{eqn:gap-expectation}
\end{align}
With $\lambda = 1/\sqrt{p}$ and $\sigma^2 = \norm{\mb q_2}_2^2/p = (1-q_1^2)/p$, we have
\begin{align*}
-\frac{\alpha}{\sigma} \;=\; -\frac{\delta+1}{\sqrt{1-q_1^2}},\quad\quad  \frac{\beta }{\sigma} \;=\; \frac{\delta-1}{\sqrt{1-q_1^2}},\quad\quad \frac{\lambda}{\sigma} \;=\; \frac{1}{\sqrt{1-q_1^2}},
\end{align*}
where $\delta = q_1/\sqrt{\theta}$ for $q_1\leq 3\sqrt{\theta}$. To proceed, it is natural to consider estimating the gap $\ol{G}(\mb q)$ by Taylor's expansion. More specifically, we approximate $\Psi\paren{-\frac{\alpha}{\sigma}}$ and $\psi\paren{-\frac{\alpha}{\sigma}}$ around $-1-\delta$, and approximate $\Psi\paren{\frac{\beta}{\sigma}}$ and $\psi\paren{\frac{\beta}{\sigma}}$ around $-1+\delta$. Applying the estimates for the relevant quantities established in Lemma~\ref{lem:gap-taylor-gaussian}, we obtain 
\begin{align*}
\ol{G}(\mb q) \;&\geq\;\frac{1-\theta}{p}\Phi_1(\delta) - \frac{1}{\delta p }\Phi_2(\delta)+ \frac{1-\theta}{p}\psi(-1)q_1^2+ \frac{1}{p}\paren{\sigma\sqrt{p} +\frac{\theta}{2}-1}\eta_2(\delta)q_1^2 \\
& + \frac{1}{2\delta p}\brac{1 + \delta^2 - \theta \delta^2 - \sigma\paren{1+\delta^2} \sqrt{p}} q_1^2 \eta_1\paren{\delta} + \frac{\sigma}{\delta \sqrt{p}} \eta_1\paren{\delta} - \frac{5C_T\sqrt{\theta}q_1^3}{p}\paren{\delta + 1}^3, 
\end{align*}
where we define
\begin{align*}
&\Phi_1(\delta) \;=\; \Psi(-1-\delta) + \Psi(-1+\delta) - 2\Psi(-1), 
&\Phi_2(\delta) \;=\; \Psi(-1+\delta) - \Psi(-1-\delta), \\
&\eta_1(\delta) \;=\; \psi(-1+\delta) - \psi(-1-\delta), 
&\eta_2(\delta) \;=\; \psi(-1+\delta) + \psi(-1-\delta), 
\end{align*}
and $C_T$ is as defined in Lemma~\ref{lem:gap-taylor-gaussian}. Since $1 - \sigma \sqrt{p} \ge 0$, dropping those small positive terms $\frac{q_1^2}{p}(1-\theta) \psi(-1)$, $\frac{\theta q_1^2}{2p}\eta_2(\delta)$, and $\paren{1+\delta^2}\paren{1-\sigma \sqrt{p}} q_1^2 \eta_1\paren{\delta}/\paren{2\delta p}$, and using the fact that $\delta = q_1/\sqrt{\theta}$, we obtain 
\begin{align*}
\ol{G}(\mb q) 
& \ge \frac{1-\theta }{p} \Phi_1(\delta) - \frac{1}{\delta p}\brac{\Phi_2(\delta) -\sigma \sqrt{p} \eta_1(\delta)} - \frac{q_1^2}{p} \paren{1-\sigma \sqrt{p}}\eta_2(\delta) - \frac{\sqrt{\theta}}{2p}q_1^3 \eta_1\paren{\delta}- \frac{C_1\sqrt{\theta}q_1^3}{p} \max\paren{\frac{q_1^3}{\theta^{3/2}}, 1}  \\
& \ge \frac{1-\theta }{p} \Phi_1(\delta) - \frac{1}{\delta p}\brac{\Phi_2(\delta) -\eta_1(\delta)} - \frac{q_1^2}{p}\frac{\eta_1\paren{\delta}}{\delta} - \frac{q_1^2}{\theta p } \paren{ \frac{2\theta}{\sqrt{2\pi}} + \frac{3 \theta^2}{2\sqrt{2\pi}} + C_1\theta^2  }, 
\end{align*}
for some constant $C_1 > 0$, where we have used $q_1 \le 3 \sqrt{\theta}$ to simplify the bounds and the fact $\sigma \sqrt{p} = \sqrt{1-q_1^2} \ge 1 - q_1^2$ to simplify the expression. Substituting the estimates in Lemma~\ref{lem:gap-partial-bounds-2} and use the fact $\delta \mapsto \eta_1\paren{\delta}/\delta$ is bounded, we obtain 
\begin{align*}
\ol{G}\paren{p} 
& \ge \frac{1}{p} \paren{\frac{1}{40} - \frac{1}{\sqrt{2\pi}} \theta} \delta^2 - \frac{q_1^2}{\theta p}\paren{c_1 \theta + c_2 \theta^2} \\
& \ge \frac{q_1^2}{\theta p}\paren{\frac{1}{40} - \frac{1}{\sqrt{2\pi}} \theta - c_1 \theta - c_2 \theta^2}
\end{align*}
for some positive constants $c_1$ and $c_2$. We obtain the claimed result once $\theta_0$ is made sufficiently small. 

\subsubsection{Auxiliary Results Used in the Proof}

\begin{lemma}\label{lem:gap-taylor-gaussian}
Let $\delta \doteq q_1/\sqrt{\theta}$. There exists some universal constant $C_T > 0$ such that we have the follow polynomial approximations hold for all $q_1\in \paren{0, \frac{1}{2}}$: 
\begin{align*}
\abs{\psi\paren{-\frac{\alpha}{\sigma}} - \brac{ 1 -  \frac{1}{2} (1+\delta)^2 q_1^2 } \psi(-1-\delta) } \;&\leq \; C_T\paren{1+\delta}^2 q_1^4, \\
 \abs{\psi\paren{\frac{\beta}{\sigma}} - \brac{ 1 -  \frac{1}{2} (\delta-1)^2 q_1^2 } \psi(\delta-1) } \;&\leq \; C_T\paren{\delta-1}^2 q_1^4, \\
 \abs{\Psi\paren{-\frac{\alpha}{\sigma}} - \brac{\Psi(-1-\delta) - \frac{1}{2} \psi(-1-\delta)(1+\delta) q_1^2 } } \;&\leq \; C_T\paren{1+\delta}^2 q_1^4, \\
 \abs{\Psi\paren{\frac{\beta }{\sigma}} - \brac{\Psi(\delta-1) + \frac{1}{2} \psi(\delta-1)(\delta-1) q_1^2 } } \;&\leq \; C_T\paren{\delta-1}^2 q_1^4, \\
 \abs{\Psi\paren{-\frac{\lambda}{\sigma} } - \brac{\Psi(-1) - \frac{1}{2}\psi(-1)q_1^2 } } \;&\leq \; C_T q_1^4. 
 \end{align*}
\end{lemma}

\begin{proof}
First observe that for any $q_1 \in \paren{0, \frac{1}{2}}$ it holds that 
\begin{align*}
0 \le \frac{1}{\sqrt{1-q_1^2}}  -\paren{1+ \frac{q_1^2}{2}} \le  q_1^4. 
\end{align*}
Hence we have 
\begin{align*}
- (1+\delta) \paren{1+ \frac{1}{2} q_1^2 + q_1^4} & \le -\frac{\alpha}{\sigma}  \le - (1+\delta) \paren{1+ \frac{1}{2} q_1^2 },  \\
\paren{\delta - 1} \paren{1+ \frac{1}{2} q_1^2} & \le \frac{\beta}{\sigma} \le \paren{\delta - 1} \paren{1+ \frac{1}{2} q_1^2 + q_1^4}, \; \text{when}\; \delta \ge 1  \\
\paren{\delta - 1} \paren{1+ \frac{1}{2} q_1^2+ q_1^4} & \le \frac{\beta}{\sigma} \le \paren{\delta - 1} \paren{1+ \frac{1}{2} q_1^2 }, \; \text{when}\; \delta \le 1. 
\end{align*}
So we have 
\begin{align*}
\psi\paren{- (1+\delta) \paren{1+ \frac{1}{2} q_1^2 + q_1^4}} \le \psi\paren{-\frac{\alpha}{\sigma}} \le \psi\paren{- (1+\delta) \paren{1+ \frac{1}{2} q_1^2 }}. 
\end{align*}
By Taylor expansion of the left and right sides of the above two-side inequality around $-1-\delta$ using Lemma \ref{lem:gaussian-taylor}, we obtain
\begin{align*}
\abs{\psi\paren{-\frac{\alpha}{\sigma}} - \psi(-1-\delta) -  \frac{1}{2} (1+\delta)^2 q_1^2  \psi(-1-\delta) } \;&\leq \; C_T\paren{1+\delta}^2 q_1^4, 
\end{align*}
for some numerical constant $C_T > 0$ sufficiently large. In the same way, we can obtain other claimed results. 
\end{proof}

\begin{lemma}\label{lem:gap-partial-bound-1}
For any $\delta \in [0, 3]$, it holds that 
\begin{align}
\Phi_2(\delta) - \eta_1(\delta) \;\geq \; \frac{\eta_1\paren{3}}{9}\delta^3 \ge \frac{1}{20} \delta^3. \label{eqn:gap-Phi-eta-inequality}
\end{align}
\end{lemma}
\begin{proof}
Let us define
\begin{align*}
h(\delta) \;=\; \Phi_2(\delta) - \eta_1(\delta) - C\delta^3
\end{align*}
for some $C > 0$ to be determined later. Then it is obvious that $h(0)=0$. Direct calculation shows that
\begin{align}
\frac{d}{d\delta}\Phi_1(\delta) = \eta_1(\delta),\quad  \frac{d}{d\delta}\Phi_2(\delta) = \eta_2(\delta),\quad \frac{d}{d\delta}\eta_1(\delta) = \eta_2(\delta)-\delta\eta_1(\delta).\label{eqn:gap-inequality-1}
\end{align}
Thus, to show \eqref{eqn:gap-Phi-eta-inequality}, it is sufficient to show that $h'(\delta) \ge 0$ for all $\delta\in [0,3]$. By differentiating $h(\delta)$ with respect to $\delta$ and use the results in \eqref{eqn:gap-inequality-1}, it is sufficient to have
\begin{align*}
h'(\delta) = \delta\eta_1(\delta) -3C\delta^2 \ge 0 \Longleftrightarrow \eta_1(\delta)  \ge 3C\delta
\end{align*}
for all $\delta \in [0, 3]$. We obtain the claimed result by observing that $\delta \mapsto \eta_1\paren{\delta}/3\delta$ is monotonically decreasing over $\delta \in \brac{0, 3}$ as justified below. 

Consider the function 
\begin{align*}
p\paren{\delta} \doteq \frac{\eta_1\paren{\delta}}{3\delta} 
& = \frac{1}{3\sqrt{2\pi}} \exp\paren{-\frac{\delta^2 + 1}{2}} \frac{e^{\delta} - e^{-\delta}}{\delta}. 
\end{align*}
To show it is monotonically decreasing, it is enough to show $p'\paren{\delta}$ is always nonpositive for $\delta \in \paren{0, 3}$, or equivalently 
\begin{align*}
g\paren{\delta} \doteq \paren{e^{\delta} + e^{-\delta}} \delta - \paren{\delta^2 + 1}\paren{e^{\delta} - e^{-\delta}} \le 0
\end{align*}
for all $\delta \in \paren{0, 3}$, which can be easily verified by noticing that $g\paren{0} = 0$ and $g'\paren{\delta} \le 0$ for all $\delta \ge 0$. 
\end{proof}

\begin{lemma}\label{lem:gap-partial-bounds-2}
For any $\delta \in [0, 3]$, we have
\begin{align}
(1-\theta)\Phi_1(\delta) - \frac{1}{\delta }\brac{\Phi_2(\delta) - \eta_1(\delta) }\geq \paren{\frac{1}{40}- \frac{1}{\sqrt{2\pi}}\theta}\delta^2. \label{eqn:gap-inequality}
\end{align}
\end{lemma}

\begin{proof}
Let us define
\begin{align*}
g(\delta) \;= \; (1-\theta)\Phi_1(\delta) - \frac{1}{\delta }\brac{\Phi_2(\delta) - \eta_1(\delta) } - c_0\paren{\theta} \delta^2,
\end{align*}
where $c_0\paren{\theta} > 0$ is a function of $\theta$. Thus, by the results in \eqref{eqn:gap-inequality-1} and L'Hospital's rule, we have
\begin{align*}
\lim_{\delta\to 0} \frac{\Phi_2(\delta)}{\delta} = \lim_{\delta \to 0} \eta_2\paren{\delta} =  2\psi(-1), \quad 
\lim_{\delta\to 0} \frac{\eta_1(\delta)}{\delta} = \lim_{\delta \to 0}\brac{\eta_2(\delta ) - \delta \eta_1(\delta) } =2\psi(-1). 
\end{align*}
Combined that with the fact that $ \Phi_1(0) = 0$, we conclude $g\paren{0} = 0$. Hence, to show \eqref{eqn:gap-inequality}, it is sufficient to show that $g'(\delta)\geq 0$ for all $\delta\in [0, 3]$. Direct calculation using the results in \eqref{eqn:gap-inequality-1} shows that
\begin{align*}
g'(\delta) =  \frac{1}{\delta^2}\brac{\Phi_2(\delta) - \eta_1(\delta)} -\theta \eta_1(\delta) - 2c_0\paren{\theta}\delta.
\end{align*}
Since $\eta_1\paren{\delta}/\delta$ is monotonically decreasing as shown in Lemma~\ref{lem:gap-partial-bound-1}, we have that for all $\delta \in \paren{0, 3}$ 
\begin{align*}
\eta_1\paren{\delta} \le \delta \lim_{\delta \to 0} \frac{\eta\paren{\delta}}{\delta} \le \frac{2}{\sqrt{2\pi}} \delta. 
\end{align*}
Using the above bound and the main result from Lemma~\ref{lem:gap-partial-bound-1} again, we obtain
\begin{align*}
g'(\delta) \;\geq\; \frac{1}{20}\delta -\frac{2}{\sqrt{2\pi}}\theta \delta - 2c_0\delta. 
\end{align*}
Choosing $c_0\paren{\theta} = \frac{1}{40} - \frac{1}{\sqrt{2\pi}} \theta$ completes the proof. 
\end{proof}

\subsection{Finite Sample Concentration}\label{app:concentration}
In the following two subsections, we estimate the deviations around the expectations $\E \brac{\ol{Q}_1(\mb q)}$ and $\E\brac{\ol{\mb Q}_2(\mb q)}$, i.e., $\abs{\ol{Q}_1(\mb q) - \E \brac{\ol{Q}_1(\mb q)}}$ and $\norm{\ol{\mb Q}_2(\mb q) -\E \brac{\ol{\mb Q}_2(\mb q) } }_2$, and show that the total deviations fit into the gap $\ol{G}(\mb q)$ we derived in Appendix \ref{app:gap}. Our analysis is based on the scalar and vector Bernstein's inequalities with moment conditions. Finally, in Appendix \ref{app:uniform_Q1_Q2}, we uniformize the bound by applying the classical discretization argument.
\subsubsection{Concentration for $\ol{Q}_1(\mb q)$} \label{app:Q1}
\begin{lemma}[Bounding $\abs{\ol{Q}_1(\mb q) - \E \brac{\ol{Q}_1(\mb q)}}$] \label{lem:Q1_deviation}
For each $\mb q \in \bb S^{n-1}$, it holds for all $t > 0$ that 
\begin{align*}
\prob{ \abs{ \ol{Q}_1(\mb q) - \bb E \brac{\ol{Q}_1(\mb q)} } \ge t} \le 2\exp\paren{-\frac{\theta p^3t^2}{8+4pt}}. 
\end{align*}
\end{lemma}

\begin{proof}
By \eqref{eqn:Q-1-2-bar}, we know that 
\begin{align*}
\ol{Q}_1(\mb q) = \frac{1}{p}\sum_{k=1}^p X_k^1,\quad X_k^1 = x_0(k)\mc S_{\lambda}\brac{x_0(k)q_1+Z_k}
\end{align*}
where $Z_k = \mb q_2^\top \mb g_k \sim \mc N\paren{0,\frac{\norm{\mb q_2}_2^2}{p}}$. Thus, for any $m\geq 2$, by Lemma \ref{lem:gaussian_moment}, we have
\begin{align*}
\bb E\brac{\abs{X_k^1}^m} \;&\leq\;  \theta \paren{\frac{1}{\sqrt{\theta p}}}^m\bb E\brac{\abs{\frac{q_1}{\sqrt{\theta p}}+ Z_k}^m}  \\
\;&= \;\theta \paren{\frac{1}{\sqrt{\theta p}}}^m\sum_{l=0}^m {m\choose l}\paren{\frac{q_1}{\sqrt{\theta p}}}^l \bb E\brac{\abs{Z_k}^{m-l}}  \\
\;&= \;  \theta \paren{\frac{1}{\sqrt{\theta p}}}^m\sum_{l=0}^m {m\choose l}\paren{\frac{q_1}{\sqrt{\theta p}}}^l (m-l-1)!!\paren{\frac{\norm{\mb q_2}_2}{\sqrt{p}}}^{m-l}  \\
\;&\leq \; \frac{m!}{2}\theta \paren{\frac{1}{\sqrt{\theta p}}}^m \paren{\frac{q_1}{\sqrt{\theta p } }+ \frac{\norm{\mb q_2}_2}{\sqrt{p}}}^m  \\
\;&\leq \; \frac{m!}{2} \theta\paren{\frac{2}{\theta p }}^m = \frac{m!}{2} \frac{4}{\theta p^2}\paren{\frac{2}{\theta p }}^{m-2} 
\end{align*}
let $\sigma_X^2 = 4/(\theta p^2) $ and $R = 2/(\theta p) $, apply Lemma \ref{lem:mc_bernstein_scalar}, we get
\begin{align*}
\bb P\brac{\abs{\ol{Q}_1(\mb q)-\bb E\brac{\ol{Q}_1(\mb q)}}\geq t } \leq 2\exp\paren{-\frac{\theta p^3t^2}{8+4pt}}.
\end{align*}
as desired.
\end{proof}

\subsubsection{Concentration for $\ol{\mb Q}_2(\mb q)$}\label{app:Q2}
\begin{lemma} [Bounding $\norm{\ol{\mb Q}_2(\mb q) - \E \brac{\ol{\mb Q}_2(\mb q)}}_2$] \label{lem:Q2_deviation}
For each $\mb q \in \bb S^{n-1}$, it holds for all $t > 0$ that 
\begin{align*}
\prob{  \norm{ \ol{\mb Q}_2(\mb q) - \bb E \brac{\ol{\mb Q}_2(\mb q)} }_2 >  t} \le 2(n+1) \exp\paren{-\frac{\theta p^3t^2}{128n + 16 \sqrt{\theta n}p t}}. 
\end{align*}
\end{lemma}
Before proving Lemma \ref{lem:Q2_deviation}, we record the following useful results.

\begin{lemma}\label{lem:gaussian-chi-moment}
For any positive integer $s,l>0$, we have
\begin{align*}
\bb E\brac{\norm{\mb g^k}_2^s \abs{\mb q_2^\top\mb g^k}^l }\leq \frac{(l+s)!!}{2} \norm{\mb q_2}_2^{l} \frac{\paren{2\sqrt{n}}^s}{\paren{ \sqrt{p}}^{s+l} }. 
\end{align*} 
In particular, when $s=l$, we have
\begin{align*}
\bb E\brac{\norm{\mb g^k}_2^l \abs{\mb q_2^\top\mb g^k}^l }\leq \frac{l!}{2} \norm{\mb q_2}_2^{l} \paren{\frac{4\sqrt{n}}{p }}^l
\end{align*}
\end{lemma}

\begin{proof}
Let $\mc P_{\mb q_2^\parallel} = \frac{\mb q_2\mb q_2^\top}{\norm{\mb q_2}_2^2} $ and $\mc P_{\mb q_2^\perp} = \paren{ \mb I - \frac{1}{\norm{\mb q_2}_2^2}\mb q_2\mb q_2^\top}$ denote the projection operators onto $\mb q_2$ and its orthogonal complement, respectively. By Lemma \ref{lem:gaussian_moment}, we have
\begin{align*}
\bb E\brac{\norm{\mb g^k}_2^s \abs{\mb q_2^\top\mb g^k}^l } \;&\leq \; \bb E\brac{\paren{\norm{\mc P_{\mb q_2^\parallel}\mb g^k}_2+ \norm{\mc P_{\mb q_2^\perp}\mb g^k}_2}^{s} \abs{\mb q_2^\top\mb g^k}^l }  \\
\;&=\; \sum_{i=0}^{s}{s \choose i} \bb E\brac{\norm{\mc P_{\mb q_2^\perp}\mb g^k }_2^{i} } \bb E\brac{ \abs{\mb q_2^\top\mb g^k}^l \norm{\mc P_{\mb q_2^\parallel}\mb g^k}_2^{s-i}}  \\
\;&=\;\sum_{i=0}^{s}{s \choose i} \bb E\brac{\norm{\mc P_{\mb q_2^\perp}\mb g^k }_2^{i} } \bb E\brac{ \abs{\mb q_2^\top\mb g^k}^{l+s-i}} \frac{1}{\norm{\mb q_2}_2^{s-i}} \\
\;&\leq \;\norm{\mb q_2}_2^l \sum_{i=0}^{s}{s \choose i}\bb E\brac{\norm{\mc P_{\mb q_2^\perp}\mb g^k }_2^{i} } \paren{\frac{1}{\sqrt{p}}}^{l+s-i}(l+s-i-1)!!.
\end{align*}
Using Lemma \ref{lem:chi_moment} and the fact that $ \norm{\mc P_{\mb q_2^\perp}\mb g^k }_2\le \norm{\mb g^k}_2$, we obtain
\begin{align*}
\bb E\brac{\norm{\mb g^k}_2^s \abs{\mb q_2^\top\mb g^k}^l }
\;&\leq\; \norm{\mb q_2}_2^l \sum_{i=0}^{s}{s \choose i} \paren{\frac{\sqrt{n}}{\sqrt{p}}}^i i!! \paren{\frac{1}{\sqrt{p}}}^{l+s-i}(l+s-i-1)!!  \\
\;&\leq\;  \norm{\mb q_2}_2^l \paren{\frac{1}{\sqrt{p}}}^l \frac{(l+s)!!}{2} \paren{\frac{\sqrt{n}}{\sqrt{p} } +\frac{1}{\sqrt{p}}}^{s}  \\
\;&\leq\;  \frac{(l+s)!!}{2} \norm{\mb q_2}_2^{l} \frac{\paren{2\sqrt{n}}^s}{\paren{ \sqrt{p}}^{s+l} }.
\end{align*}
\end{proof}

Now, we are ready to prove Lemma \ref{lem:Q2_deviation},
\begin{proof}
By \eqref{eqn:Q-1-2-bar}, note that 
\begin{align*}
\ol{\mb Q}_2 = \frac{1}{p}\sum_{k=1}^p \mb X_k^2,\quad \mb X_k^2 = \mb g^k \mc S_\lambda\brac{x_0(k) q_1+Z_k}
\end{align*}
where $Z_k =  \mb q_2^\top\mb g^k$. Thus, for any $m\geq 2$, by Lemma \ref{lem:gaussian-chi-moment}, we have
\begin{align*}
\bb E\brac{\norm{\mb X_k^2}_2^m} \;&\leq\; \theta \bb E\brac{\norm{\mb g^k}_2^m\abs{\frac{q_1}{\sqrt{\theta p}}+ \mb q_2^\top\mb g^k}^m }+ (1-\theta) \bb E\brac{\norm{\mb g^k}_2^m\abs{\mb q_2^\top\mb g^k}^m } \\
 \;&\leq\;\theta \sum_{l=0}^m{m\choose l}\bb E\brac{\abs{\mb q_2^\top\mb g^k}^{l}\norm{\mb g^k}_2^{m}}\abs{\frac{q_1}{\sqrt{\theta p}}}^{m-l}  + (1-\theta) \bb E\brac{\norm{\mb g^k}_2^m\abs{\mb q_2^\top\mb g^k}^m }  \\
  \;&\leq\;\theta\paren{\frac{2\sqrt{n}}{\sqrt{p}}}^m \sum_{l=0}^m {m\choose l}\frac{(m+l)!!}{2}\paren{\frac{\norm{\mb q_2}_2}{\sqrt{p}}}^l\abs{\frac{q_1}{\sqrt{\theta p}}}^{m-l} +(1-\theta)\frac{m!}{2} \norm{\mb q_2}_2^{m} \paren{\frac{4\sqrt{n}}{p }}^m \\
   \;&\leq\;\theta \frac{m!}{2}\paren{\frac{4\sqrt{n}}{\sqrt{p}}}^m\paren{\frac{\norm{\mb q_2}_2 }{\sqrt{p}}+\frac{q_1}{\sqrt{\theta p}}}^m+(1-\theta)\frac{m!}{2} \norm{\mb q_2}_2^{m} \paren{\frac{4\sqrt{n}}{p }}^m \\
   \; &\le\; \frac{m!}{2}\paren{\frac{8\sqrt{n}}{\sqrt{\theta}p}}^m. 
 \end{align*}
Taking $\sigma_X^2 = 64n/(\theta p^2)$ and $R = 8\sqrt{n}/(\sqrt{\theta}p)$ and using vector Bernstein's inequality in Lemma \ref{cor:vector-bernstein}, we obtain
\begin{align*}
\bb P\brac{\norm{\ol{\mb Q}_2(\mb q) -\bb E\brac{\ol{\mb Q}_2(\mb q)}}_2 \geq t}\;\leq\; 2(n+1) \exp\paren{-\frac{\theta p^3t^2}{128n + 16 \sqrt{\theta n} p t}},
\end{align*}
as desired.
\end{proof}

\subsection{Union Bound}\label{app:uniform_Q1_Q2}


\begin{proposition}[Uniformizing the Bounds] \label{lem:uniform_Q1_Q2}
Suppose that $\theta > 1/\sqrt{n}$. Given any $\xi > 0$, there exists some constant $C\paren{\xi}$, such that whenever $p \geq C\paren{\xi} n^4 \log n$, we have
\begin{align*}
\abs{\ol{Q}_1(\mb q) - \bb E \brac{\ol{Q}_1(\mb q)} } \;&\leq\; \frac{2\xi}{\theta^{5/2} n^{3/2} p}, \\
 \norm{ \ol{\mb Q}_2(\mb q) - \bb E \brac{\ol{\mb Q}_2(\mb q)} }_2 \;&\leq\; \frac{2\xi}{ \theta^2 n p}
\end{align*}
hold uniformly for all $\mb q \in \bb S^{n-1}$, with probability at least $1- c(\xi) p^{-2}$ for a positive constant $c(\xi)$. 
\end{proposition}

\begin{proof}
We apply the standard covering argument. For any $\eps \in \left(0, 1\right)$, by Lemma \ref{lem:eps-net-cover}, the unit hemisphere of interest can be covered by an $\eps$-net $\mc N_\eps$ of cardinality at most $\left(3/\eps\right)^n$. For any $\mb q \in \bb S^{n-1}$, it can be written as 
\begin{align*}
\mb q = \mb q' + \mb e 
\end{align*}
where $\mb q' \in \mc N_\eps$ and $\norm{\mb e}_2 \leq \eps$. Let a row of $\ol{\mb Y}$ be $\ol{\mb y}^k =\brac{x_0(k),\mb g^k}^\top  $, which is an independent copy of $\overline{\mb y} = \brac{x_0,\mb g }^\top$. By \eqref{eqn:Q-1-2-bar}, we have
\begin{align*}
&\abs{\ol{Q}_1(\mb q) - \bb E \brac{\ol{Q}_1(\mb q)}}  \\
&=\;  \abs{ \frac{1}{p}\sum_{k=1}^p  \Brac{x_0(k)\mc S_\lambda \brac{\innerprod{\ol{\mb y}^k}{\mb q'+\mb e} }  - \bb E\brac{x_0(k) \mc S_\lambda \brac{\innerprod{\ol{\mb y}^k}{\mb q'+\mb e}}} }}  \\
&\leq \; \abs{ \frac{1}{p}\sum_{k=1}^p x_0(k) \mc S_\lambda \brac{\innerprod{\ol{\mb y}^k}{\mb q'+\mb e} }  -  \frac{1}{p}\sum_{k=1}^p x_0(k) \mc S_\lambda \brac{\innerprod{\ol{\mb y}^k}{\mb q'} } } + \abs{ \frac{1}{p}\sum_{k=1}^p x_0(k) \mc S_\lambda \brac{\innerprod{\ol{\mb y}^k}{\mb q'} }  -  \bb E\brac{x_{0} \mc S_\lambda \brac{\innerprod{\overline{\mb y} }{\mb q'}}}  }  \\
& + \abs{\bb E\brac{x_{0} \mc S_\lambda \brac{\innerprod{\overline{\mb y} }{\mb q'}}} - \bb E\brac{x_{0} \mc S_\lambda \brac{\innerprod{\overline{\mb y} }{\mb q'+\mb e}}}}. 
\end{align*}
Using Cauchy-Schwarz inequality and the fact that $\mc S_\lambda \brac{\cdot}$ is a nonexpansive operator, we have
\begin{align*}
\abs{\ol{Q}_1(\mb q) - \bb E \brac{\ol{Q}_1(\mb q)}} 
& \le \abs{ \ol{Q}_1(\mb q') - \bb E \brac{\ol{Q}_1(\mb q')} } 
+ \paren{\frac{1}{p}\sum_{k=1}^p \abs{x_0(k)} \norm{\ol{\mb y}^k}_2+ \bb E\brac{\abs{x_{0}}\norm{\overline{\mb y} }_2 }}\norm{\mb e}_2 \\
& \le \abs{ \ol{Q}_1(\mb q') - \bb E \brac{\ol{Q}_1(\mb q')} }  + \eps \frac{1}{\sqrt{\theta p}} \paren{\frac{2}{\sqrt{\theta p}} + \max_{k \in [p]} \norm{\mb g^k}_2  + \bb E\brac{ \norm{\mb g}_2 } }. 
\end{align*}
By Lemma~\ref{lem:gauss_seq_norm}, $\max_{k \in [p]} \norm{\mb g^k}_2 \le \sqrt{n/p} + 2\sqrt{2\log(2p) /p}$ with probability at least $1- c_1 p^{-3}$. Also $\bb E\brac{\norm{\mb g}_2 } \leq \paren{\bb E\brac{\norm{\mb g}_2^2  } }^{1/2} \leq \sqrt{n/p}$. Taking $t = \xi \theta^{-5/2} n^{-3/2} p^{-1}$ in Lemma~\ref{lem:Q1_deviation} and applying a union bound with $\eps = \xi \theta^{-2} n^{-2} (\log 2p)^{-1/2}/7$, and combining with the above estimates, we obtain that 
\begin{align*}
\abs{\ol{Q}_1(\mb q) - \bb E \brac{\ol{Q}_1(\mb q)}} \le \frac{\xi}{\theta^{5/2} n^{3/2} p} + \frac{\xi}{7\theta^{5/2} n^2\sqrt{\log (2p)} p} \paren{4\sqrt{n} + 2\sqrt{2\log(2p)}} \le \frac{2\xi}{\theta^{5/2} n^{3/2} p}
\end{align*} 
holds for all $\mb q \in \bb S^{n-1}$, with probability at least 
\begin{align*}
1- c_1p^{-3} - 2\exp\paren{-c_3\paren{\xi} p/(\theta^4 n^3) + c_4\paren{\xi}n \log n + c_5(\xi) n \log \log (2p)}.
\end{align*} 
Similarly, by \eqref{eqn:Q-1-2-bar}, we have
\begin{align*}
\norm{\ol{\mb Q}_2(\mb q) - \bb E\brac{\ol{\mb Q}_2(\mb q)}}_2 \;&=\; \norm{\frac{1}{p} \sum_{k=1}^p \Brac{\mb g^k \mc S_\lambda\brac{\innerprod{\overline{\mb y}^k}{\mb q'+\mb e}} - \bb E\brac{\mb g \mc S_\lambda\brac{\innerprod{\overline{\mb y} }{\mb q'+\mb e}} }}  }_2  \\
\;&\leq \; \norm{\ol{\mb Q}_2(\mb q') - \bb E\brac{\ol{\mb Q}_2(\mb q')}}_2+ \paren{ \frac{1}{p}\sum_{k=1}^p \norm{\mb g^k}_2 \norm{\overline{\mb y}^k}_2 + \bb E\brac{\norm{\mb g}_2 \norm{\overline{\mb y}}_2} }\norm{\mb e}_2  \\
\;&\leq \; \norm{\ol{\mb Q}_2(\mb q') - \bb E\brac{\ol{\mb Q}_2(\mb q')}}_2+  \eps \brac{\max_{k \in [p]} \norm{\mb g^k}_2 \paren{\frac{1}{\sqrt{\theta p}} + \max_{k \in [p]} \norm{\mb g^k}_2} + \frac{\sqrt{n}}{\sqrt{\theta}p} + \frac{n}{p} }. 
\end{align*}
Applying the above estimates for $\max_{k \in [p]} \norm{\mb g^k}_2$, and taking $t = \xi \theta^{-2}n^{-1}p^{-1}$ in Lemma~\ref{lem:Q2_deviation} and applying a union bound with $\eps = \xi \theta^{-2} n^{-2} \log^{-1}(2p)/30$, we obtain that 
\begin{align*}
\norm{\ol{\mb Q}_2(\mb q) - \bb E\brac{\ol{\mb Q}_2(\mb q)}}_2 
& \le \frac{\xi}{\theta^2 np} + \frac{\xi}{30\theta^2 n^2 \log(2p)} \Brac{4\paren{\sqrt{\frac{n}{p}} + \sqrt{\frac{2\log(2p)}{p}}}^2 + \frac{2n}{p}} \\
& \le \frac{\xi}{\theta^2 np} + \frac{\xi}{30\theta^2 n^2 \log(2p)} \Brac{\frac{16\log(2p)}{p} + \frac{10n}{p}} \\
& \le \frac{2\xi}{\theta^2 np}
\end{align*}
holds for all $\mb q \in \bb S^{n-1}$, with probability at least 
\begin{align*}
1 - c_1 p^{-3} - \exp\paren{-c_6\paren{\xi} p/(\theta^3 n^3) + c_7(\xi) n \log n + c_8(\xi) n \log \log(2p) }.
\end{align*} 
Taking $p \ge C_9(\xi) n^4 \log n$ and simplifying the probability terms complete the proof. 
\end{proof}

\subsection{ $\mb Q(\mb q)$ approximates $\ol{\mb Q}(\mb q)$ } \label{app:q_prime_approx}

\begin{proposition} \label{lem:q-appx}
Suppose $\theta > 1/\sqrt{n}$. For any $\xi > 0$, there exists some constant $C\paren{\xi}$, such that whenever $p \ge C\paren{\xi} n^4 \log n$, the following bounds 
\begin{align}
\sup_{\mb q\in \bb S^{n-1}} \abs{ Q_1(\mb q) - \ol{Q}_1(\mb q) }  &\leq \frac{\xi}{\theta^{5/2} n^{3/2} p} \label{eqn:Q-Q-bar-1} \\
\sup_{\mb q\in \bb S^{n-1}} \norm{ \mb Q_2(\mb q) - \ol{\mb Q}_2(\mb q) }_2 &\leq \frac{\xi}{ \theta^2 n p }, \label{eqn:Q-Q-bar-2}
\end{align}
hold with probability at least $1-c(\xi)p^{-2}$ for a positive constant $c(\xi)$. 
\end{proposition}

\begin{proof}
First, for any $\mb q\in \bb S^{n-1}$, from \eqref{eqn:Q-1-2-bar}, we know that
\begin{align*}
& \abs{\ol{Q}_1(\mb q) - Q_1(\mb q)}  \\
\;&=\; \abs{ \frac{1}{p} \sum_{k=1}^p x_0(k) \mc S_\lambda \brac{\mb q^\top \ol{\mb y}^k } - \frac{1}{p} \sum_{k=1}^p  \frac{x_0(k)}{\norm{\mb x_0}_2} \mc S_\lambda \brac{\mb q^\top \mb y^k} }  \\
\;&\leq \; \abs{ \frac{1}{p} \sum_{k=1}^p x_0(k) \mc S_\lambda \brac{\mb q^\top \ol{\mb y}^k} - \frac{1}{p} \sum_{k=1}^p  x_0(k) \mc S_\lambda \brac{\mb q^\top \mb y^k} }+ \abs{ \frac{1}{p} \sum_{k=1}^p x_0(k) \mc S_\lambda \brac{\mb q^\top \mb y^k} - \frac{1}{p} \sum_{k=1}^p  \frac{x_0(k)}{\norm{\mb x_0}_2}\mc S_\lambda \brac{\mb q^\top \mb y^k} }  \\
\;&\leq \;\frac{1}{p}\sum_{k=1}^p \abs{x_0(k)} \abs{ \mc S_\lambda \brac{\mb q^\top \ol{\mb y}^k}- \mc S_\lambda \brac{\mb q^\top \mb y^k}} + \frac{1}{p}\sum_{k=1}^p \abs{x_0(k)}\abs{1 - \frac{1}{\norm{\mb x_0}_2}}\abs{ \mc S_\lambda \brac{\mb q^\top \mb y^k} }.
\end{align*}
For any $\mc I =\text{supp}(\mb x_0)$, using the fact that $\mc S_\lambda[\cdot]$ is a nonexpansive operator, we have 
\begin{align*}
\sup_{\mb q \in \bb S^{n-1}} \abs{\ol{Q}_1(\mb q) - Q_1(\mb q)}
&\le \frac{1}{p} \sup_{\mb q \in \bb S^{n-1}}\sum_{k\in \mc I} \abs{x_0(k)} \abs{\mb q^\top \paren{\ol{\mb y}^k - \mb y^k}} +   \abs{1 - \frac{1}{\norm{\mb x_0}_2}} \frac{1}{p} \sup_{\mb q \in \bb S^{n-1}}\sum_{k\in \mc I} \abs{x_0(k)}\abs{\mb q^\top \mb y^k }  \\
& = \frac{1}{\sqrt{\theta}p^{3/2}} \paren{\norm{\ol{\mb Y}_{\mc I}-\mb Y_{\mc I}}_{\ell^2 \to \ell^1} +\abs{1 - \frac{1}{\norm{\mb x_0}_2}} \norm{\mb Y_{\mc I} }_{\ell^2 \to \ell^1} }.
\end{align*}
By Lemma~\ref{lem:frac-x_0} and Lemma~\ref{lem:app_basis_op2} in Appendix \ref{sec:app_bases}, we have the following holds 
\begin{align*}
\sup_{\mb q \in \bb S^{n-1}} \abs{\ol{Q}_1(\mb q) - Q_1(\mb q)}\;& \leq\; \frac{1}{\sqrt{\theta}p^{3/2}}\paren{20\sqrt{\frac{n\log p}{\theta }}+ \frac{4\sqrt{2}}{5}\sqrt{\frac{n\log p}{\theta^2 p}}\times 7\sqrt{2\theta p}} \leq \frac{32}{\theta p^{3/2}}\sqrt{n\log p},
\end{align*}
with probability at least $1-c_1p^{-2}$, provided $p \ge C_2 n$ and $\theta > 1/\sqrt{n}$. Simple calculation shows that it is enough to have $p \geq C_3\paren{\xi} n^4 \log n$ for some sufficiently large $C_1\paren{\xi}$ to obtain the claimed result in \eqref{eqn:Q-Q-bar-1}.
Similarly, by Lemma~\ref{lem:app_basis_op2} and Lemma~\ref{lem:app_basis_op3} in Appendix \ref{sec:app_bases}, we have 
\begin{align*}
& \sup_{\mb q \in \bb S^{n-1}} \norm{\ol{\mb Q}_2(\mb q) - \mb Q_2(\mb q)}_2 \\
= & \sup_{\mb q \in \bb S^{n-1}} \norm{\frac{1}{p}\sum_{k=1}^p \mb g^k \mc S_\lambda \brac{\mb q^\top \ol{\mb y}^k }- \frac{1}{p}\sum_{k=1}^p \mb g'^k \mc S_\lambda \brac{\mb q^\top \mb y^k } }_2  \\
\le & \sup_{\mb q \in \bb S^{n-1}}\norm{\frac{1}{p}\sum_{k=1}^p \mb g^k  \mc S_\lambda \brac{\mb q^\top \ol{\mb y}^k } - \frac{1}{p}\sum_{k=1}^p \mb g'^k  \mc S_\lambda \brac{\mb q^\top \overline{\mb y}^k } }_2+ \norm{\frac{1}{p}\sum_{k=1}^p \mb g'^k  \mc S_\lambda \brac{\mb q^\top \ol{\mb y}^k } - \frac{1}{p}\sum_{k=1}^p \mb g'^k  \mc S_\lambda \brac{\mb q^\top \mb y^k } }_2  \\
\le & \frac{1}{p} \sup_{\mb q \in \bb S^{n-1}}\sum_{k=1}^p \norm{\mb g^k - \mb g'^k}_2 \abs{\mb q^\top \ol{\mb y}^k} + \frac{1}{p} \sup_{\mb q \in \bb S^{n-1}}\sum_{k=1}^p \norm{\mb g'^k}_2 \abs{\mb q^\top \paren{\ol{\mb y}^k - \mb y^k} } \\
\le & \frac{1}{p} \paren{\norm{\mb G-\mb G' }_{\ell^2\to \ell^\infty} \norm{\ol{\mb Y} }_{\ell^2\to \ell^1} + \norm{\mb G' }_{\ell^2 \to \ell^\infty}  \norm{\ol{\mb Y}-\mb Y}_{\ell^2 \to \ell^1} }  \\
\le & \frac{1}{p}\paren{\frac{120\max(n, \log(2p))}{\sqrt{p}}+ \frac{300 \sqrt{n \log(2p)} \max(\sqrt{n}, \sqrt{\log(2p)})}{\sqrt{\theta p}} } \;\leq\; \frac{420\sqrt{n \log(2p)} \max(\sqrt{n}, \sqrt{\log(2p)})}{\theta^{1/2}p^{3/2}}
\end{align*}
with probability at least $1 - c_4 p^{-2}$ provided $p \ge C_4 n$ and $\theta > 1/\sqrt{n}$. It is sufficient to have $p \ge C_5\paren{\xi} n^4 \log n$ to obtain the claimed result \eqref{eqn:Q-Q-bar-2}. 
\end{proof}

\section{Large $\abs{q_1}$ Iterates Staying in Safe Region for Rounding}\label{app:safe-region}


In this appendix, we prove Proposition \ref{lem:safe} in Section \ref{sec:analysis}.

\begin{proof}[Proof of Proposition \ref{lem:safe}]
For notational simplicity, w.l.o.g. we will proceed to prove assuming $q_1 > 0$. The proof for $q_1 < 0$ is similar by symmetry. It is equivalent to show that 
\begin{align*}
\frac{\norm{\mb Q_2\paren{\mb q}}_2}{\abs{Q_1\paren{\mb q}}} < \sqrt{\frac{1}{4\theta} - 1}, 
\end{align*}
which is implied by 
\begin{align*}
\mc L\left(\mb q\right) \doteq \frac{\norm{ \bb E \brac{\ol{\mb Q}_2(\mb q)} }_2 + \norm{ \mb Q_2(\mb q) - \bb E\brac{ \ol{\mb Q}_2(\mb q)} }_2 }{\bb E \brac{\ol{Q}_1\left(\mb q\right)} - \abs{Q_1\left(\mb q\right) - \bb E \brac{\ol{Q}_1\left(\mb q\right)} }} < \sqrt{\frac{1}{4\theta} - 1} 
\end{align*}
for any $\mb q\in \bb S^{n-1}$ satisfying $q_1 > 3 \sqrt{\theta}$. Recall from~\eqref{eqn:expect-Q-1} that
\begin{align*}
\E \brac{\ol{Q}_1(\mb q)} = \sqrt{\frac{\theta}{p}} \Brac{\brac{\alpha\Psi\paren{-\frac{\alpha}{\sigma} }+\beta\Psi\paren{\frac{\beta}{\sigma} }}+\sigma\brac{\psi\paren{\frac{\beta}{\sigma}}-\psi\paren{-\frac{\alpha}{\sigma}}}},
\end{align*} 
where 
\begin{align*}
\alpha = \frac{1}{\sqrt{p}} \left(\frac{q_1}{\sqrt{\theta}} + 1\right), \quad \beta = \frac{1}{\sqrt{p}} \left(\frac{q_1}{\sqrt{\theta}} - 1\right), \quad \sigma = \norm{\mb q_2}_2/\sqrt{p}. 
\end{align*}
Noticing the fact that 
\begin{align*}
\psi\paren{\frac{\beta}{\sigma}}-\psi\paren{-\frac{\alpha}{\sigma}} 
& \ge 0, \\
\Psi\left(\frac{\beta}{\sigma}\right) 
& = \Psi\left(\frac{1}{\sqrt{1-q_1^2}}\left(\frac{q_1}{\sqrt{\theta}} - 1\right)\right) \ge  \Psi\left(2\right) \ge \frac{19}{20} \quad \text{for}\; q_1 > 3\sqrt{\theta}, 
\end{align*}
we have 
\begin{align*}
\bb E \brac{\ol{Q}_1\left(\mb q\right)} & \ge \frac{\sqrt{\theta}}{p}\Brac{\frac{q_1}{\sqrt{\theta}}\brac{\Psi\paren{-\frac{\alpha}{\sigma}} + \Psi\paren{\frac{\beta}{\sigma}}} + \Psi\paren{-\frac{\alpha}{\sigma}} - \Psi\paren{\frac{\beta}{\sigma}}} \ge \frac{2\sqrt{\theta}}{p} \Psi\paren{\frac{\beta}{\sigma}} \ge \frac{19}{10} \frac{\sqrt{\theta}}{p}. 
\end{align*}
Moreover, from~\eqref{eqn:expect-Q-2}, we have 
\begin{align*}
\norm{\E \brac{\ol{\mb Q}_2\left(\mb q\right)} }_2
& = \norm{\mb q_2}_2 \Brac{\frac{2\paren{1-\theta}}{p}\Psi\paren{-\frac{\lambda}{\sigma}}+\frac{\theta}{p}\brac{\Psi\paren{-\frac{\alpha}{\sigma}}+\Psi \paren{\frac{\beta}{\sigma}}}} \\
& \le \frac{2\paren{1-\theta}}{p} \Psi\paren{-1} + \frac{\theta}{p}\brac{\Psi\paren{-1} + 1} \le \frac{2}{p} \Psi\paren{-1} + \frac{\theta}{p} \le \frac{2}{5p} + \frac{\theta}{p}, 
\end{align*}
where we have used the fact that $-\lambda/\sigma \le -1$ and $-\alpha/\sigma \le -1$. Moreover, from results in Proposition~\ref{lem:uniform_Q1_Q2} and Proposition~\ref{lem:q-appx} in Appendix \ref{app:gap-finite}, we know that 
\begin{align*}
\sup_{\mb q \in \bb S^{n-1}} \abs{Q_1(\mb q) - \bb E \brac{\ol{Q}_1(\mb q) } }
& \le \sup_{\mb q \in \bb S^{n-1}} \abs{Q_1(\mb q) - \ol{Q}_1(\mb q) }  + \sup_{\mb q \in \bb S^{n-1}} \abs{\ol{Q}_1(\mb q) - \bb E\brac{ \ol{Q}_1(\mb q) } }  \le  \frac{1}{2\times 10^5\theta^{5/2}n^{3/2}p},  \\
\sup_{\mb q \in \bb S^{n-1}} \norm{\mb Q(\mb q) - \bb E\brac{ \ol{\mb Q}(\mb q)} }_2 & \le \sup_{\mb q \in \bb S^{n-1}} \norm{ \mb Q(\mb q) - \ol{\mb Q}(\mb q) }_2 + \sup_{\mb q \in \bb S^{n-1}} \norm{ \ol{\mb Q}(\mb q) - \bb E \brac{\ol{\mb Q}(\mb q)} }_2 \le \frac{1}{2\times 10^5\theta^2 np}
\end{align*}
hold with probability at least $1-c_1p^{-2}$ provided that $p \ge \Omega\paren{n^4 \log n}$. Hence, with high probability, we have 
\begin{align*}
\mc L\paren{\mb q} \le \frac{2/(5p) + \theta/p + (2\times 10^5\theta^2 np)^{-1}}{19\sqrt{\theta}/(10p) -  (2\times 10^5\theta^{5/2}n^{3/2}p)^{-1}} \le \frac{3/5}{18\sqrt{\theta}/10} \le \frac{1}{3\sqrt{\theta}} < \sqrt{\frac{1}{4\theta} - 1}, 
\end{align*}
whenever $\theta$ is sufficiently small. This completes the proof. 
\end{proof}

Now, keep the notation in Appendix \ref{app:gap-finite} for general orthonormal basis $\widehat{\mb Y} = \mb Y\mb U$. For any current iterate $\mb q\in \bb S^{n-1}$ that is close enough to the target solution, i.e., $\abs{\innerprod{\mb q}{\mb U^\top \mb e_1}} = \abs{\innerprod{\mb U \mb q}{\mb e_1}}\ge 3\sqrt{\theta}$, we have 
\begin{align*}
\frac{\abs{\innerprod{\mb Q\paren{\mb q; \widehat{\mb Y}}}{\mb U^\top \mb e_1}}}{\norm{\mb Q\paren{\mb q; \widehat{\mb Y}}}_2} = \frac{\abs{\innerprod{\mb U \mb Q\paren{\mb q; \widehat{\mb Y}}}{\mb e_1}}}{\norm{\mb U \mb Q\paren{\mb q; \widehat{\mb Y}}}_2} = \frac{\abs{\innerprod{\mb Q\paren{\mb U \mb q; \mb Y}}{\mb e_1}}}{\norm{\mb Q\paren{\mb U \mb q; \mb Y}}_2}, 
\end{align*}
where we have applied the identity proved in~\eqref{eq:general_basis_identity}. Taking $\mb U \mb q \in \bb S^{n-1}$ as the object of interest, by Proposition~\ref{lem:safe}, we conclude that 
\begin{align*}
\frac{\abs{\innerprod{\mb Q\paren{\mb U \mb q; \mb Y}}{\mb e_1}}}{\norm{\mb Q\paren{\mb U \mb q; \mb Y}}_2} \ge 2\sqrt{\theta}
\end{align*}
with high probability. 

\section{Bounding Iteration Complexity}
 \label{app:iter_cplx}
In this appendix, we prove Proposition \ref{prop:iter-complexity} in Section \ref{sec:analysis}.
\begin{proof}[Proof of Proposition \ref{prop:iter-complexity}]
Recall from Proposition \ref{prop:gap-bound-Y'} in Section \ref{sec:analysis}, the gap 
\begin{align*}
G(\mb q) = \frac{\abs{Q_1(\mb q)}}{\abs{q_1}} - \frac{\norm{\mb Q_2(\mb q)}_2}{\norm{\mb q}_2}\;&\geq \;\frac{1}{10^4\theta^2 np }
\end{align*}
holds uniformly over $\mb q\in \bb S^{n-1}$ satisfying $\frac{1}{10\sqrt{\theta n}} \le \abs{q_1} \le 3\sqrt{\theta}$, with probability at least $1- c_1p^{-2}$, provided $p \ge C_2 n^4 \log n$. The gap $G(\mb q)$ implies that 
\begin{align*}
& \abs{\widetilde{Q}_1\paren{\mb q}} \doteq \frac{\abs{Q_1(\mb q)}}{\norm{\mb Q\paren{\mb q}}_2} \ge  \frac{\abs{q_1}\norm{\mb Q_2(\mb q)}_2}{\norm{\mb q}_2 \norm{\mb Q\paren{\mb q}}_2} + \frac{\abs{q_1}}{10^4\theta^2 np \norm{\mb Q\paren{\mb q}}_2} \\
\Longleftrightarrow &  \abs{\widetilde{Q}_1\paren{\mb q}} \ge \frac{\abs{q_1}}{\norm{\mb q_2}_2} \sqrt{1-\abs{\widetilde{Q}_1\paren{\mb q}}^2} + \frac{\abs{q_1}}{10^4\theta^2 np \norm{\mb Q\paren{\mb q}}_2} \\
\Longrightarrow & \abs{\widetilde{Q}_1\paren{\mb q}}^2 \ge \abs{q_1}^2 \paren{1 + \frac{\norm{\mb q_2}_2^2}{10^8 \theta^4 n^2 p^2 \norm{\mb Q\paren{\mb q}}_2^2}}. 
\end{align*}
Given the set $\Gamma$ defined in \eqref{eqn:Gamma-set}, now we know that 
\begin{align*}
\sup_{\mb q \in \Gamma} \norm{\mb Q\paren{\mb q}}_2 
& \le \sup_{\mb q \in \Gamma} \abs{\expect{\overline{Q}_1(\mb q)}} + \sup_{\mb q \in \bb S^{n-1} } \abs{\expect{\overline{Q}_1(\mb q)}- \overline{Q}_1\paren{\mb q}} + \sup_{\mb q \in \bb S^{n-1} } \abs{Q_1(\mb q)- \overline{Q}_1\paren{\mb q}}\\
& \quad + \sup_{\mb q \in \Gamma} \norm{\expect{\overline{\mb Q}_2(\mb q)}}_2 + \sup_{\mb q \in \bb S^{n-1} } \norm{\expect{\overline{\mb Q}_2(\mb q)} - \overline{\mb Q}_2\paren{\mb q}}_2 + \sup_{\mb q \in \bb S^{n-1} } \norm{\mb Q_2(\mb q) - \overline{\mb Q}_2\paren{\mb q}}_2 \\
& \le \sup_{\mb q \in \Gamma} \abs{\expect{\overline{Q}_1(\mb q)}} + \sup_{\mb q \in \Gamma} \abs{\expect{\overline{\mb Q}_2(\mb q)}} + \frac{1}{pn}
\end{align*}
with probability at least $1 - c_3 p^{-2}$ provided $p \ge C_4 n^4 \log n$ and $\theta > 1/\sqrt{n}$. Here we have used Proposition~\ref{lem:uniform_Q1_Q2} and Proposition~\ref{lem:q-appx} to bound the magnitudes of the four difference terms. To bound the magnitudes of the expectations, we have 
\begin{align*}
\abs{\expect{\overline{Q}_1(\mb q)}} 
& = \abs{\expect{\frac{1}{p}\sum_{k=1}^p x_0(k) S_{\lambda}\brac{x_0(k)q_1 + \mb q_2^\top \mb g^k}}} \le \frac{1}{\sqrt{\theta p}} \paren{\frac{1}{\sqrt{\theta p}} + \expect{\norm{\mb g}_2}} \le \frac{3\sqrt{n}}{\sqrt{\theta} p} \le \frac{3n}{p}, \\
\norm{\expect{\overline{\mb Q}_2(\mb q)}}_2
& = \norm{\expect{\frac{1}{p} \sum_{k=1}^p \mb g^k S_{\lambda}\brac{x_0(k) q_1 + \mb q_2^\top \mb g^k}}}_2 \le \frac{1}{\sqrt{\theta p}} \expect{\norm{\mb g}_2} + \expect{\norm{\mb g}_2^2} \le \frac{3n}{p}
\end{align*}
hold uniformly for all $\mb q \in \Gamma$, provided $\theta > 1/\sqrt{n}$. Thus, we obtain that 
\begin{align*}
\sup_{\mb q \in \Gamma} \norm{\mb Q\paren{\mb q}}_2  \le \frac{3n}{p} + \frac{3n}{p} + \frac{1}{np} \le \frac{7n}{p}
\end{align*} 
with probability at least $1 - c_3 p^{-2}$ provided $p \ge C_4 n^4 \log n$ and $\theta > 1/\sqrt{n}$. So we conclude that 
\begin{align*}
\frac{\abs{\widetilde{Q}_1\paren{\mb q}} }{\abs{q_1}} \ge \sqrt{1 + \frac{1-9\theta}{10^8 \times 7^2 \times \theta^4 n^4}}. 
\end{align*}
Thus, starting with any $\mb q\in \bb S^{n-1}$ such that $\abs{q_1} \ge \frac{1}{10\sqrt{\theta n}}$, we will need at most 
\begin{align*}
T = \frac{2\log\paren{3\sqrt{\theta}/\frac{1}{10\sqrt{\theta n}}}}{\log\paren{1 + \frac{1-9\theta}{10^8 \times 7^2 \times \theta^4 n^4}} } = \frac{2\log\paren{30\theta \sqrt{n}}}{\log\paren{1 + \frac{1-9\theta}{10^8 \times 7^2 \times \theta^4 n^4}} } \le \frac{2\log\paren{30\theta \sqrt{n}}}{\paren{\log 2} \frac{1-9\theta}{10^8 \times 7^2 \times \theta^4 n^4} } \le C_5 n^4 \log n
\end{align*}
steps to arrive at a $\overline{\mb q}\in \bb S^{n-1}$ with $\abs{\bar{q_1}} \ge 3\sqrt{\theta}$ for the first time. Here we have assumed $\theta_0 < 1/9$ and used the fact that $\log\paren{1+x} \ge x\log 2$ for $x \in \brac{0, 1}$ to simplify the final result. 
\end{proof}

\section{Rounding to the Desired Solution}\label{app:rounding}


In this appendix, we prove Proposition \ref{lem:rounding} in Section \ref{sec:analysis}. For convenience, we will assume the notations we used in Appendix~\ref{sec:app_bases}. Then the rounding scheme can be written as 
\begin{align} \label{eqn:rounding_analysis}
\min_{\mb q}\; \norm{\mb Y \mb q}_1, \quad \mathrm{s.t.} \; \innerprod{\overline{\mb q}}{\mb q} = 1.
\end{align}
We will show the rounding procedure get us to the desired solution with high probability, regardless of the particular orthonormal basis used.

\begin{proof}[Proof of Proposition \ref{lem:rounding}]
The rounding program~\eqref{eqn:rounding_analysis} can be written as  
\begin{align} \label{eqn:round-origin}
\inf_{\mb q}\; \norm{\mb Y\mb q}_1,\quad \mathrm{s.t.}\;\; \overline{q}_1 q_1+ \innerprod{\overline{\mb q}_2}{\mb q_2}=1.
\end{align}
Consider its relaxation
\begin{align}  \label{eqn:round-relax-1}
\inf_{\mb q}\; \norm{\mb Y\mb q}_1,\quad \mathrm{s.t.}\;\; \overline{q}_1 q_1+ \norm{\overline{\mb q}_2}_2 \norm{\mb q_2}_2\geq 1.
\end{align}
It is obvious that the feasible set of \eqref{eqn:round-relax-1} contains that of \eqref{eqn:round-origin}. So if 
$\mb e_1/\overline{q}_1 $ is the unique optimal solution (UOS) of \eqref{eqn:round-relax-1}, it is also the UOS of \eqref{eqn:round-origin}. Let $\mc I = \text{supp}(\mb x_0)$, and consider a modified problem
\begin{align}\label{eqn:round-relax-2}
\inf_{\mb q}\; \norm{\frac{\mb x_0}{\norm{\mb x_0}_2}}_1\abs{q_1} -\norm{\mb G'_{\mc I}\mb q_2}_1+ \norm{\mb G'_{\mc I^c}\mb q_2}_1,\quad \mathrm{s.t.}\;\; \overline{q}_1 q_1+ \norm{\overline{\mb q}_2}_2 \norm{\mb q_2}_2\geq 1. 
\end{align}
The objective value of \eqref{eqn:round-relax-2} lower bounds the objective value of \eqref{eqn:round-relax-1}, and are equal when $\mb q = \mb e_1/ \overline{q}_1 $. So if $\mb q=\mb e_1/ \overline{q}_1 $ is the UOS to \eqref{eqn:round-relax-2}, it is also UOS to \eqref{eqn:round-relax-1}, and hence UOS to \eqref{eqn:round-origin} by the argument above. Now 
\begin{align*}
-\norm{\mb G'_{\mc I}\mb q_2}_1+ \norm{\mb G'_{\mc I^c}\mb q_2}_1 
& \ge -\norm{\mb G_{\mc I} \mb q_2}_1 + \norm{\mb G_{\mc I^c} \mb q_2}_1 - \norm{\paren{\mb G - \mb G'} \mb q_2}_1 \\
& \ge -\norm{\mb G_{\mc I} \mb q_2}_1 + \norm{\mb G_{\mc I^c} \mb q_2}_1 - \norm{\mb G - \mb G'}_{\ell^2 \to \ell^1} \norm{\mb q_2}_2. 
\end{align*}
When $p \ge C_1 n$, by Lemma~\ref{lem:gs_l2_l1} and Lemma~\ref{lem:app_basis_op2}, we know that 
\begin{multline*}
-\norm{\mb G_{\mc I} \mb q_2}_1 + \norm{\mb G_{\mc I^c} \mb q_2}_1 - \norm{\mb G - \mb G'}_{\ell^2 \to \ell^1} \norm{\mb q_2}_2\\
 \ge -\frac{6}{5}\sqrt{\frac{2}{\pi}}2\theta\sqrt{p}\norm{\mb q_2}_2 + \frac{24}{25}\sqrt{\frac{2}{\pi}}\paren{1-2\theta}\sqrt{p}\norm{\mb q_2}_2 - 4\sqrt{n}\norm{\mb q_2}_2 - 7\sqrt{\log(2p)}\norm{\mb q_2}_2 \doteq \zeta \norm{\mb q_2}_2
\end{multline*}
holds with probability at least $1-c_2p^{-2}$. Thus, we make a further relaxation of problem \eqref{eqn:round-origin} by
\begin{align}\label{eqn:round-relax-3}
\inf_{\mb q}\; \norm{\frac{\mb x_0}{\norm{\mb x_0}_2}}_1\abs{q_1} + \zeta\norm{\mb q_2}_2 ,\quad \mathrm{s.t.}\;\; \overline{q}_1 q_1+ \norm{\overline{\mb q}_2}_2 \norm{\mb q_2}_2\geq 1, 
\end{align}
whose objective value lower bounds that of \eqref{eqn:round-relax-2}. By similar arguments, if $\mb e_1/ \overline{q}_1 $ is UOS to \eqref{eqn:round-relax-3}, it is UOS to \eqref{eqn:round-origin}. At the optimal solution to \eqref{eqn:round-relax-3}, notice that it is necessary to have $\text{sign}(q_1) =\text{sign}(\overline{q}_1)$ and $\overline{q}_1 q_1+ \norm{\overline{\mb q}_2}_2 \norm{\mb q_2}_2= 1$. So \eqref{eqn:round-relax-3} is equivalent to
\begin{align}\label{eqn:round-relax-4}
\inf_{\mb q}\; \norm{\frac{\mb x_0}{\norm{\mb x_0}_2}}_1\abs{q_1} + \zeta\norm{\mb q_2}_2 ,\quad \mathrm{s.t.}\;\; \overline{q}_1 q_1+ \norm{\overline{\mb q}_2}_2 \norm{\mb q_2}_2=1.
\end{align}
which is further equivalent to
\begin{align}\label{eqn:round-relax-5}
\inf_{q_1}\;  \norm{\frac{\mb x_0}{\norm{\mb x_0}_2}}_1\abs{q_1} +\zeta \frac{1-\abs{\overline{q}_1} \abs{q_1}}{\norm{\overline{\mb q}_2}_2},\quad \mathrm{s.t.}\;\; \abs{q_1}\leq \frac{1}{\abs{\overline{q}_1}}.
\end{align}
Notice that the problem in \eqref{eqn:round-relax-5} is linear in $\abs{q_1}$ with a compact feasible set. Since the objective is also monotonic in $\abs{q_1}$, it indicates that the optimal solution only occurs at the boundary points $\abs{q_1} =0 $ or $\abs{q_1} = 1/\abs{\overline{q}_1}$ 
 Therefore, $\mb q = \mb e_1/\overline{q}_1 $ is the UOS of \eqref{eqn:round-relax-5} if and only if
\begin{align*}
\frac{1}{\abs{\overline{q}_1}}\norm{\frac{\mb x_0}{\norm{\mb x_0}_2}}_1\;<\; \frac{\zeta}{\norm{\overline{\mb q}_2}_2}. 
\end{align*}
Since $\norm{\frac{\mb x_0}{\norm{\mb x_0}_2}}_1 \le \sqrt{2\theta p}$ conditioned on $\event_0$, it is sufficient to have 
\begin{align*}
\frac{\sqrt{2\theta p}}{2\sqrt{\theta}} \le \zeta = \frac{24}{25}\sqrt{\frac{2}{\pi}} \sqrt{p}\paren{1-\frac{9}{2}\theta - \frac{25}{6}\sqrt{\frac{\pi}{2}}\sqrt{\frac{n}{p}} - \frac{175}{24}\sqrt{\frac{\pi}{2}}\sqrt{\frac{\log(2p)}{p}}}. 
\end{align*}
Therefore there exists a constant $\theta_0 >0$, such that whenever $\theta \leq  \theta_0$ and $p \ge C_3(\theta_0) n$, the rounding returns $\mb e_1/\overline{q}_1 $. A bit of thought suggests one can take a universal $C_3$ for all possible choice of $\theta_0$, completing the proof. 
\end{proof}
When the input basis is $\widehat{\mb Y} = \mb Y \mb U$ for some orthogonal matrix $\mb U \neq \mb I$, if the ADM algorithm produces some $\overline{\mb q} = \mb U^\top \mb q'$, such that $q'_1 > 2\sqrt{\theta}$. It is not hard to see that now the rounding~\eqref{eqn:rounding_analysis} is equivalent to 
\begin{align*}
\min_{\mb q}\; \norm{\mb Y \mb U \mb q}_1, \quad \mathrm{s.t.} \; \innerprod{\mb q'}{\mb U\mb q} = 1. 
\end{align*}
Renaming $\mb U\mb q$, it follows from the above argument that at optimum $\mb q_\star$ it holds that $\mb U \mb q_\star= \gamma \mb e_1$ for some constant $\gamma$ with high probability.

{
\bibliographystyle{ieeetr} 
\bibliography{IT,ncvx}
}

\end{document}